\documentclass{article}

\usepackage{amsmath}
\usepackage{amsthm}
\usepackage{amssymb}
\usepackage{todonotes}
\usepackage{rotating}
\usepackage{array}
\usepackage{enumitem}
\usepackage{mathtools}
\usepackage{setspace}
\usepackage{authblk}
\usepackage{fullpage}
\usepackage[section]{placeins}
\usepackage{cases}

\usepackage[ruled]{algorithm2e}

\SetCommentSty{mycommfont}

\newtheorem{theorem}{Theorem}
\newtheorem{definition}{Definition}
\newtheorem{example}{Example}

\def \ie {\textit{i.e.}}
\def \eg {\textit{e.g.}}

\def \etal {\textit{et~al.}}

\DeclareMathOperator* {\argmax} {arg\,max}
\DeclareMathOperator* {\argmin} {arg\,min}

\newcommand{\N}{\mathbb{N}}
\newcommand{\R}{\mathbb{R}}
\newcommand{\et}{\tau}
\newcommand{\en}{\eta}
\newcommand{\SR}{\Phi}
\newcommand{\sr}{\phi}
\newcommand{\vs}{v_S}
\newcommand{\seq}{\gamma}
\newcommand{\Seq}{\Gamma}
\newcommand{\avs}{\hat{\Seq}}
\newcommand{\Su}{\sigma}
\newcommand{\dist}{\mathit{dist}}

\begin{document}

\title{Strategic Attack \& Defense in Security Diffusion Games}

\author[a,b,c,*]{Marcin Waniek}
\author[c]{Tomasz P. Michalak}
\author[b]{Aamena Alshamsi}

\renewcommand*{\Affilfont}{\normalsize}

\affil[a]{Computer Science Department, New York University Abu Dhabi, Abu Dhabi, UAE}
\affil[b]{Masdar Institute, Khalifa University of Science and Technology, Abu Dhabi, UAE}
\affil[c]{Institute of Informatics, University of Warsaw, Warsaw, Poland}
\affil[*]{To whom correspondence should be addressed: mjwaniek@nyu.edu}

\date{}

\maketitle

\begin{abstract}
Security games model the confrontation between a defender protecting a set of targets and an attacker who tries to capture them.
A variant of these games assumes security interdependence between targets, facilitating contagion of an attack. 
So far only stochastic spread of an attack has been considered.
In this work, we introduce a version of security games, where the attacker strategically drives the entire spread of attack and where interconnections between nodes affect their susceptibility to be captured.

We find that the strategies effective in the settings without contagion or with stochastic contagion are no longer feasible when spread of attack is strategic.
While in the former settings it was possible to efficiently find optimal strategies of the attacker, doing so in the latter setting turns out to be an NP-complete problem for an arbitrary network.
However, for some simpler network structures, such as cliques, stars, and trees, we show that it is possible to efficiently find optimal strategies of both players.
For arbitrary networks, we study and compare the efficiency of various heuristic strategies.
As opposed to previous works with no or stochastic contagion, we find that centrality-based defense is often effective when spread of attack is strategic, particularly for centrality measures based on the Shapley value.
\end{abstract}

\section{Introduction}


Security games~\cite{roy2010survey}---a particular type of Stackelberg games~\cite{von1934marktform}---has recently attracted growing interest in the literature. 
These games consider a confrontation between two players: a defender, \eg, a police force trying to protect certain resources or locations, and an attacker, \eg, a criminal organization intending to gain control over the locations or the resources.
One of the milestones proving the effectiveness of such models is the ARMOR system used to protect the LAX airport~\cite{pita2008deployed}.
Its success led to the development of many other systems based on security games~\cite{tsai2009iris,tambe2011security,shieh2012protect}.

Some security games settings consider a potential spread of the attack after hitting the initial target~\cite{kunreuther2003interdependent,johnson2010uncertainty}.
In particular, Bachrach~\etal~\cite{bachrach2013contagion} created a variation of the security game model with stochastic contagion between two targets. The idea was later expanded by adding a whole network of interdependencies between targets~\cite{acemoglu2016network, lou2017multidefender}. However, so far the literature on security games has addressed only the stochastic spread of attacks. While this fits certain scenarios such as the spread of computer viruses or fires, it is not suitable for scenarios in which the attack is more likely to diffuse be strategic  of the attack. An example of such a setting can be an adversarial recruitment of people in social networks which is usually not contagious but carefully planned. As we discuss below, the adversarial recruitment of people may occur in either competitive scenarios, such as customer poaching and brand switching~\cite{fudenberg2000customer}, or malicious ones, such as recruitment of individuals to join criminal and covert organizations~\cite{leukfeldt2017cybercriminal,saad2018terrorist}.

An example of a competitive scenario is an archetypical Stackelberg game, i.e. the competition between two market players---an incumbent and an entrant~\cite{von1934marktform,etro2007stackelberg}. One of the ways in which such an entry can be executed, especially in the time of growing popularity of personalised marketing campaigns~\cite{matloka2010destination}, is by poaching the customers of the incumbent~\cite{chen2010dynamic}, e.g. by offering them discounts and free samples. As a pre-emptive move, the incumbent may try to secure the loyalty of the customers by compensating them with competitive offers. In this game, the incumbent is the defender, who acts in advance to prevent the customers from switching to the entrant, who is the attacker. Furthermore, given that the decision to switch to new products or services, especially in high-tech markets \cite{hall2003adoption}, is often affected by social influence, the consideration of the social network is key to capture the complexity of this scenario~\cite{iyengar2011opinion,karsai2014complex}.

An example of a malicious scenario concerns the recruitment of individuals to join covert organizations. It is well known that various covert organization recruit members by carefully selecting potential candidates using social ties and they do this typically by taking advantage of the kinship and friendship ties~\cite{leukfeldt2017cybercriminal}.
In such scenarios, the attacker is a covert organization, who intends to strategically select potential candidates from the acquaintances of the already-recruited members.
On the other hand, the defender is either the law enforcement or a charitable institution, whose goal is to prevent the criminal degeneration or radicalization of a local community and who has limited security resources to distribute over the social networks in an attempt to achieve this goal.

Unfortunately, the existing models are inadequate to grasp complexity of these scenarios.
On one hand, the aforementioned defender-attacker literature~\cite{bachrach2013contagion, acemoglu2016network, lou2017multidefender} assumes a stochastic contagion of the attack over the network.
On the other hand, the existing models of strategic diffusion~\cite{alshamsi2018optimal} do not include an optimizing defender counteracting the spread, \ie, they only model the diffusion aspect of the problem. Furthermore, there exist models of market competition, where one of the competitors strategically schedules an order in which members of the network make decision on adopting either her product, or the product of her competitor, based on their network neighborhood~\cite{chierichetti2014schedule,arthur1989competing}.
However, even though these works study market competition, only one competitor in these models makes strategic decisions, while the other remains passive.
Moreover, the temporal aspect is limited to selecting a sequence of nodes, with each node in the sequence making her decision immediately.
In reality, more resilient nodes may take more time to be convinced by a competitor, adding another layer of complexity to the problem.
Another body of literature considers strategically affecting the spread of diffusion by selecting the exact moment when the seed nodes are activated, instead of activating them all at the beginning of a marketing campaign process~\cite{jankowski2017balancing,jankowski2018strategic}.
Nevertheless, after activating the seed nodes the remainder of the process is purely stochastic, whereas we intend to model a fully strategic diffusion process.
Also, these settings do not include any competitive elements, whereas the existence of an entity intending to prevent the spread of diffusion is a key element of our approach.

Against this background, we propose and analyze the first Stackelberg game between a defender and an attacker of the network, where the attacker strategically guides the course of network diffusion and a central defender strategically attempts to counter the expected attack in advance. At the beginning of the game, the defender distributes security resources among nodes of the network.
Having observed the security efforts of the defender, the attacker chooses the sequence in which the nodes are to be attacked, where the time it takes to successfully attack a certain node is inversely proportional to the number of already captured neighbors of this node.
We use complex contagion to model the susceptibility level of each member of the network, as it has been shown to better describe diffusion processes in social networks~\cite{centola2007complex, centola2010spread, centola2011experimental}.
Within this framework, we investigate two main questions: \textit{how difficult it is for the attacker to find an optimal way of strategically attacking a network} and \textit{what are the effective ways of defending a network from a fully sequential strategic attack?} Our study mainly focuses on finding the best strategies for both players.

We first prove that in the general case, i.e., for an arbitrary network, the problem of finding an optimal plan of the strategic attack is NP-complete. Given this hardness result, we next analyse whether the same holds when we focus our attention on some simple network structures. In doing so, we follow the studies of the previous models in the literature (with no strategic spread of attack) for which polynomial results were obtained for some simple networks~\cite{bachrach2013contagion, acemoglu2016network, lou2017multidefender}, \eg, by considering only a very limited number of nodes~\cite{bachrach2013contagion} or by considering only symmetric networks~\cite{acemoglu2016network}. The literature shows that in such cases finding the optimal strategy of an attacker is typically relatively easy---she just selects the target offering the highest expected utility. Given this, the remaining analysis was typically focused on analyzing the strategies of the defender. Interestingly, we are still able to find optimal strategies for some some simple network structures in our more challenging model. In particular, for cliques and stars, we develop polynomial algorithms to search for the optimal strategies for both the attacker and the defender. 

Next, for arbitrary but small networks, we develop a dynamic programming algorithm that finds an optimal attack more efficiently than checking all possible sequences, albeit still in exponential time. We also show an effective way of protecting the network from a fixed set of attacker's strategies, similar to that proposed by Lou~\etal~\cite{lou2017multidefender}.

Our hardness result shows that considering the strategic spread of an attack, rather than a stochastic one, significantly increases the challenge faced by both players.  
Indeed, in the Stackelberg game, the NP-completeness of finding an optimal attack constitutes a challenge not only to the attacker but also to the defender. This is because, it may be prohibitively difficult for the defender to choose her best strategy given that this choice is based on the prediction of the best response by the attacker.

A popular way of dealing with such computational challenges efficiently, although most probably suboptimally, is to resort to heuristics. Hence, we propose a variety of heuristic strategies for both players. As for the attacker's strategies, we consider the ones found effective in strategic diffusion models~\etal~\cite{alshamsi2018optimal}, as well as strategies based on exploration-exploitation techniques~\cite{sutton1998reinforcement}.
As for the defender's strategies, the previous studies found that the optimal defender's utility is achieved when security levels of all targets are equal~\cite{bachrach2013contagion, lou2017multidefender}.
However, our findings indicate that making all nodes equally difficult to target is not particularly effective against attack that spreads strategically.
Instead, in most cases the defense based on one of the network centrality measures provides better results.

\section{Preliminaries}
\label{sec:preliminaries}

\noindent Let $G = (V, E)$ be a network, where $V=\{v_1,\ldots,v_n\}$ denotes the set of $n$ nodes and $E \subseteq V \times V$ denotes the set of edges.
Let $\mathbb{G}(V)$ be a set of all networks over the set of nodes $V$.
We denote an edge between nodes $v_i$ and $v_j$ by $(v_i,v_j)$.
In this work, we consider only \textit{undirected} networks, \ie, networks in which we do not discern between edges $(v_i,v_j)$ and $(v_j,v_i)$.
We also assume that networks do not contain self-loops, \ie, $\forall_{v_i \in V}(v_i,v_i) \notin E$.
We denote by $N_G(v_i)$ the set of \textit{neighbours} of $v_i$ in $G$, \ie, $N_G(v_i) = \{v_j \in V : (v_i,v_j) \in E\}$.
Furthermore, we denote by $d_G(v_i)$ the \textit{degree} of $v_i$ in $G$, \ie, $d_G(v_i) = |N_G(v_i)|$. We will often omit the network itself from the notation when it is clear from the context, \eg, by writing $N(v)$ instead of $N_G(v)$.

We will also use a concept of a path in the graph.
In particular, a path is a sequence of distinct nodes, $\langle v_l, \ldots, v_k\rangle$, such that every two consecutive nodes are connected by an edge.
The length of a path is the number of edges in that path. 
Let $\Seq(V)$ denote the set of all \textit{sequences} of elements from $V$ without repetitions.
Let $\seq_i$ denote the $i$-th element of sequence $\seq \in \Seq(V)$.
Finally, let $|\seq|$ denote the number of elements in $\seq \in \Seq(V)$.

In graph theory and social network analysis, the importance of nodes is quantified using functions of the form: $c : \mathbb{G}(V) \times V \rightarrow \R$, called \textit{centrality measures}.
Three such fundamental centrality measures are the \textit{degree} centrality~\cite{shaw1954group}, the \textit{closeness} centrality~\cite{beauchamp1965improved}, and the \textit{betweeness} centrality~\cite{anthonisse1971rush, freeman1977set}. Specifically, given a node $v_i \in V$ and an undirected network, we have that the formula of degree centrality is: $c_{degr}(v_i) = \frac{d(v_i)}{n - 1}$. In words, degree centrality ranks the nodes simply according to their number of neighbors. Next,  the formula of closeness centrality is:
$$
	c_{clos}(v_i) = \frac{n - 1}{\sum_{v_j \in V}\dist(v_i,v_j)},
$$
\noindent where $\dist(v_i,v_j)$ denotes the distance between $v_i$ and $v_j$.
Intuitively, closeness centrality focuses on distances among nodes and gives high value to the nodes that are, on average, close to all other nodes. Furthermore, betweenness centrality is defined as:
$$
	c_{betw}(v_i) = \frac{2}{(n-1)(n-2)}
		\sum_{v_j,v_k \in V \setminus \{ v_i \}}
			\frac
				{|\{ p \in sp(v_j,v_k) : v_i \in p \}|}
				{|sp(v_j,v_k)|},
$$
\noindent where $sp_G(v_i,v_j)$ denotes the set of all shortest paths between any pair of nodes, $v_i$ and $v_j$.
Thus, in words, betweenness centrality considers all shortest paths between any two nodes in the network.
The more such shortest paths a particular node belongs to, the more important it is.

Apart from the above three fundamental centrality measures, many other have been proposed in the literature. One interesting class are \textit{game-theoretic centrality} measures.
The key idea behind game-theoretic centrality measures is to treat nodes as players in a cooperative game, where the value of each coalition of nodes is determined by certain graph-theoretic properties~\cite{michalak2013efficient,skibski2018axiomatic}. The key advantage of this approach is that nodes are ranked not only according to their individual roles in the network but also according to how they contribute to the role played by all possible subsets (or groups) of nodes. In game theory, the function  that assigns a value to any subset of nodes, $u: 2^V \rightarrow \R$, is called the \textit{characteristic function}. Now, the best known game-theoretic method to evaluate the role of a player in the game is the Shapley value:
$$
\psi(v_i, u) = \sum_{C \subseteq V \setminus \{v_i\}} \frac{|C|!(n-|C|-1)!}{n!} \left( u(C \cup \{ v_i \}) - u(C) \right).
$$
\noindent The value of $\psi(v_i, u)$ is the Shapley value of player $v_i$ in game $u$. Now,  in  each Shapley value-based centrality measure $c$ the centrality score of a node $v_i$ is equal to its Shapley value in the game $u$ associated with this centrality measure, \ie, $c(v_i) = \psi(v_i, u)$.
We consider the following centralities based on the Shapley value:
\begin{itemize}
\item The \textit{Shapley value-based degree centrality}~\cite{michalak2013efficient} of a node is its Shapley value in a game with the characteristic function:
$$
u(C) = \left|C \cup \bigcup_{v_i \in C} N(v_i)\right|.
$$
In words, in this game the value of a given coalition is the size of the coalition and its neighborhood.

\item The \textit{Shapley value-based closeness centrality}~\cite{michalak2013efficient} is defined by:\footnote{\footnotesize Note that there are many other ways in which one could define the game-theoretic extension of  closeness centrality~\cite{tarkowski2018efficient}.}
$$
u(C) = \left|C \cup \{v_j \in V: \exists_{v_i \in C} dist(v_i,v_j) \leq \delta\}\right|.
$$
In words, in this game the value of a given coalition is the number of its nodes and all nodes in distance at most $\delta$. In this work we assume $\delta=3$.

\item \textit{Shapley value-based betweenness centrality}~\cite{szczepanski2016efficient} is defined by:
$$
u(C) = \sum_{v_j,v_k \in V \setminus C}
			\frac
				{|\{ p \in sp(v_j,v_k) : \exists_{v_i \in C} v_i \in p \}|}
				{|sp(v_j,v_k)|}.
$$
In words, the value of a given coalition depends on the percentage of the shortest paths between pairs of nodes not from the coalition, that the members of the coalition belong to.
\end{itemize}


\section{The Model of Security Diffusion Games} 
\label{sec:game-definition}

\noindent
We define a security diffusion game by combining ideas of security games and strategic diffusion in networks. 

\begin{figure*}[t]
\centering
\includegraphics[width=.8\linewidth]{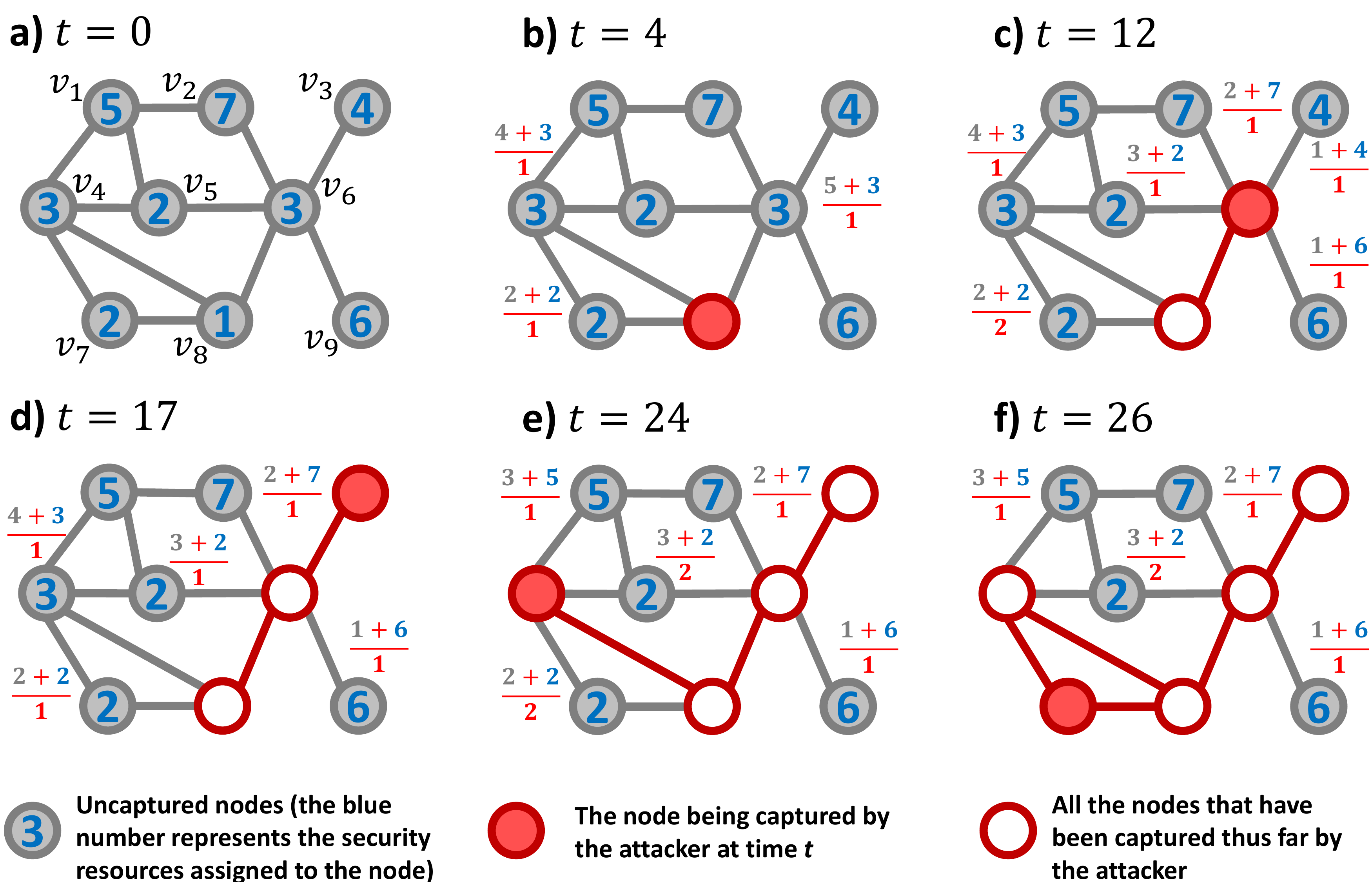}
\caption{Figure (a) depicts a sample network with the set of nodes $\{v_1,\ldots,v_9\}$. Figures (b)-(f) depict sequence $\seq = (v_8,v_6,v_3,v_4,v_7)$ in which the attacker chooses to spread the attack. Note that, for the sake of clarity, we omit the labels of nodes in Figures (b)-(f). 
}
\label{fig:example}
\end{figure*}

\begin{definition}[Security Diffusion Game]
The security diffusion game is a Stackelberg game between two players, the defender (the leader) and the attacker (the follower).
It is defined by a tuple $(G, \SR, T)$, where $G = (V,E)$ is a given network, $\SR \in \R^+$ is the amount of security resources available to the defender, and $T \in \R^+$ is the attacker's time limit before the defender discovers the attack and stops it completely. The game has two steps:

\begin{itemize}
\item Step~1---at the beginning of the game, the defender distributes security resources among nodes of the network, \ie, she chooses $\sr : V \rightarrow \R$ such that $\forall_{v \in V} \sr(v) \geq 0$ and $\sum_{v \in V} \sr(v) \leq \SR$.
\item Step~2--- next, having observed the security efforts of the defender in Step~1, the attacker chooses the sequence, $\seq \in \Seq(V)$,  in which nodes will be targeted for attacks. A given node in the sequence can be attacked only when all preceding nodes have been successfully captured.
\end{itemize}
The activation time of each node $v$,  denoted $\et(v)$, is a function of the network topology, the set of already activated  nodes, and the distribution of the security resources (to improve readability we omit the additional arguments in the notation). The utilities of the attacker and the defender are denoted $u_A$ and $u_D$, respectively, and are functions of the network, the distribution of security resources, the sequence of the attack, and the time limit $T$: 
\begin{align*}
u_{D} = f^D(G,\sr,\seq,T); \\
u_{A} = f^A(G,\sr,\seq,T).
\end{align*}
\end{definition}

In this work we make the following assumptions about the particular forms of the utility functions and the activation time function. As for the utility functions, we assume that $u_A = \en(\seq,T)$, where $\en(\seq,T)$ denotes the number of nodes activated within time limit $T$ when using attack sequence $\seq$ (to simplify notation we assume that the network and the distribution of security resources are known from the context). 
Furthermore, to express the competitive nature of the game, we define the utility of the defender as $u_D = -u_A = -\en(\seq,T)$, which makes the security diffusion game a zero-sum game.

We define the activation time of a node as a function of its assigned security resources, its degree, and the number of the activated nodes among its neighbours:\footnote{\footnotesize 
The function in equation \eqref{cs:time-second} is a form of the Tullock (contest success) function which is widely used in the security games literature~\cite{bachrach2013contagion, acemoglu2016network, lou2017multidefender} and has various desirable properties~\cite{skaperdas1996contest, bachrach2013contagion}.
Furthermore, we note that in cyber-security security domain a similar idea is expressed by the Byzantine generals problem~\cite{lamport1982byzantine}.}
\begin{numcases}{\et(v) = }
d(v) + \sr(v) & if $v$ is the seed node, \ie, the first node in the sequence $\seq$ \label{cs:time-first} \\
\frac{d(v) + \sr(v)}{|N(v) \cap I|} & otherwise, \label{cs:time-second}
\end{numcases}
\noindent where $I$ is the set of currently activated nodes. The above form of activation time function combines notions from strategic diffusion~\cite{alshamsi2018optimal}  and from security games~\cite{paruchuri2008playing}. In particular:
\begin{itemize}
    \item The more security resources assigned to $v$, the more time it takes to activate this node (see the term $\sr(v)$ in equation~\ref{cs:time-first} for the seed node and the nominator of equation~\ref{cs:time-second} for all the other nodes);
    \item The time of activation decreases with the number of activated neighbours (see the term $|N(v) \cap I|$ in the denominator of equation~\ref{cs:time-second}). In other words, the more neighbours of $v$ have been already activated, the more this node is susceptible to become activated itself.
\end{itemize}
Furthermore, the activation time increases with the total degree of the node (see the term $\sr(v)$ in equation~\ref{cs:time-first} for the seed node and the nominator of equation~\ref{cs:time-second} for all the other nodes).  This is a well established idea in the diffusion models~\cite{kempe2003maximizing, alshamsi2018optimal} and complex contagion~\cite{centola2007complex} literature and it has the following intuitive interpretation. If the number of already activated nodes around $v$ makes only a small proportion of the total number of $v$'s neighbours then $v$ is less likely to become active itself.

We also note that the activation time function (\ref{cs:time-first}-\ref{cs:time-second}) implies that, while we allow the attacker to choose any node in the network as the seed node, any other node needs at least one activated neighbour to become  activated itself. Hence, the strategy of the attacker (\ie, the sequence of nodes that she chooses to attack) constitutes by definition a path in the graph.

\begin{example}
Figure~\ref{fig:example} presents a sample network of $n=9$ nodes and a possible sequence, $\seq = (v_8,v_6,v_3,v_4,v_7)$, in which the attacker  captures the nodes within the time limit $T \geq 26$, \ie:
\begin{itemize}
    \item Figure~\ref{fig:example} a) depicts a not-yet-attacked network at time $t=0$, where the number inside a node shows security resources assigned to this node by the defender;
    \item Figure~\ref{fig:example} b) depicts the moment at which the attacker, having chosen to attack node $v_8$, captures it (at time $t = d(v_8)+\sr(v_8) = 3+1 = 4$);
    \item Figure~\ref{fig:example} c) depicts the activation of $v_6$ at time $t = 4+ \left(d(v_6)+\sr(v_6)\right)/\left(|N(v_6) \cap I| \right) = 4+(5+3)/1 = 12$);
    \item Figure~\ref{fig:example} d) depicts the activation of $v_3$ at time $t = 12+ \left(d(v_3)+\sr(v_3)\right)/\left(|N(v_3) \cap I| \right) = 12+(1+4)/1 = 17$);    
    \item Figure~\ref{fig:example} e) depicts the activation of $v_4$ at time $t = 17+ \left(d(v_4)+\sr(v_4)\right)/\left(|N(v_4) \cap I| \right) = 17+(4+3)/1 = 24$); and
    \item Figure~\ref{fig:example} f) depicts the activation of $v_7$ at time $t = 24+ \left(d(v_7)+\sr(v_7)\right)/\left(|N(v_7) \cap I| \right) = 24+(2+2)/2 = 26$).
\end{itemize}
\end{example}
We can also relate our model to both motivating scenarios presented in the introduction.  Our sample competitive scenario involves two companies---an incumbent and an entrant. The former one is the defender who can invest resources into securing higher levels of loyalty of its customer base. Conversely, the entrant is the attacker whose goal is to attract as many customers of the defender as possible. The activation time for each customer (node) represents here the difficulty for the entrant of doing so. As per equation~(\ref{cs:time-first}-\ref{cs:time-second}), in our model this difficulty for any customer $v$ is assumed to be in the inverse proportion to the share of the already attracted (activated) customers in the neighbourhood of $v$. As we already argued, this is a well-established approach in the complex contagion literature~\cite{centola2007complex}. In our context, such a definition of the activation time reflects that in many markets (especially hi-tech ones) the customers are more likely to switch to a competing company (or a technology) if more of their friends already use it. We also note that the strategy of the attacker (the sequence of nodes) induces a connected subgraph in the network;  hence, our model fits the settings of the word-of-mouth marketing~\cite{kozinets2010networked}.

Our sample malicious scenario involves either a criminal or extremist organisation who selects the candidates for new members using social ties of their existing rank and file. Here, the defender is either the law enforcement or a charitable institution whose goal is to protect a community. The activation time for each node represents the threshold of an individual to  become attracted by the covert organisation. The meaning of equation~(\ref{cs:time-first}-\ref{cs:time-second}) in this context is that any individual $v$ is assumed to be more attracted to become a member of a covert organisation if relatively many of her neighbours have already done so. This reflects many studies of illegal behaviours in the literature---for instance, Bikhchandani~\etal~\cite[p. 994]{bikhchandani1992theory} ``consider a teenager deciding whether or not to try drugs. A strong motivation for trying out drugs is the fact that friends are doing so. Conversely, seeing friends reject drugs could help persuade the teenager to stay clean''. Furthermore, we note that the fact that the strategy of the attacker induces a connected subgraph is natural in this scenario as it ensures that the covert network remains connected.

\section{Finding Optimal Strategies}

We begin our analysis of the model by considering the problem of finding optimal strategies. Since the Stackelberg game is solved by backward induction,  we first consider the strategy  of the attacker (given an arbitrary strategy of the defender). In particualr,  we formally define the Optimal Attack Problem as follows:

\begin{definition}[Optimal Attack Problem]
This problem is defined by a tuple $(G,\sr,\vs,T)$, where $G=(V,E)$ is a given network, $\sr : V \rightarrow \R$ is the defender's distribution of security resources, $\vs \in V$ is the seed node of the attacker, and $T$ is the attacker's time limit.
The goal is to identify $\seq^* \in \Seq(V)$ such that $\seq^*$ is in $\argmax_{\seq \in \Seq(V) : \seq_1 = \vs} \en(\seq,T)$.\footnote{Recall that $\en(\seq,T)$ denotes the number of nodes activated within time limit $T$ when using attack sequence $\seq$.}
\end{definition}

As mentioned in the introduction, in the existing literature, where the attack spreads in the form of a stochastic process, the attacker's strategy consists only of choosing a seed node; hence, it is usually easy to pick an optimal strategy of the attacker. However, finding the optimal strategy of the attacker in our model is no longer trivial. Indeed, as we show in Section~\ref{sec:npc}, in the general case, \ie, when we allow a network to have any structure, the problem of finding a sequence of nodes providing the largest number of nodes activated within a time limit is NP-complete.
Hence, even if the attacker is given complete information about the activities of the defender, finding an optimal strategy of the attack is computationally intractable.

However, we are able to find optimal strategies of both attacker and defender for more restricted network structures.
In particular, we find that, if the network under consideration is a star, the attacker should attack nodes in non-decreasing order of their activation times, with the additional constraint that the center of the star has to be one of the first two nodes in the sequence (see Section~\ref{app:stars}). As for the defender strategy, if it is possible not to let the attacker activate even a single node given the time limit, then the defender should try to make the time of activation equal for all nodes (if possible) or all peripheral nodes (if there are not enough security resources).
Alternatively, the defender should assign all security resources to the central node (since again, it has to be one of the first two nodes in the sequence). We also analyse the optimal strategies for cliques, where the optimal attack sequence orders the nodes non-decreasingly according to assigned security resources, while the defender should spread her security resources uniformly (see Appendix~\ref{app:cliques}). We can also efficiently find an optimal attack sequence for a tree (see Appendix~\ref{app:trees}).

In the case we need to find an optimal attacker's strategy for an arbitrary (but small) network, we present an algorithm based on dynamic programming technique that returns solution in time $\mathcal{O}(2^n n^2)$, where $n$ is the number of nodes in the network.
Its pseudocode can be found as Algorithm~\ref{alg:dynamic-programming} in Appendix~\ref{app:dynamic}.
While it is more effective than simple exhaustive search, its exponential time complexity makes it suitable for rather small networks.

Finally, we also present a mixed-integer-linear-programming method to compute optimal defense strategies when the strategies of the attacker can be enumerated, \ie, when the set of the attacker's strategies is of reasonable size.
The formulation has $\mathcal{O}(|\avs|n)$ constraints, where $\avs$ is the set of strategies available to the attacker (see Section~\ref{app:milp}).
While this method would be highly impractical if we assume that attacker can use any strategy, it can help defend the network effectively when some additional constraints are put on the attacker, \eg, we can ensure that some ways of attack are completely implausible.

\subsection{Complexity Analysis of the Optimal Attack Problem}
\label{sec:npc}

The following result holds.

\begin{theorem}
\label{thrm:optimal-attack-npcomplete}
The Optimal Attack problem is NP-complete.
\end{theorem}

\begin{proof}

The decision version of the optimization problem is the following: given a network $G=(V,E)$, a distribution of security resources $\sr$, a seed node $\vs$, a time limit $T$, and a value $r^* \in \N$ expressing the number of nodes to be activated within time limit, does there exist a sequence of nodes $\seq^* \in \Seq(V)$ such that $\en(\seq^*,T) \geq r^*$?

This problem clearly is in NP. This is because, given a solution, \ie, a sequence $\seq^* \in \Seq(V)$, we can compute the time of activation of every node in the sequence, and verify in polynomial time whether $\en(\seq^*,T) \geq r^*$.

To prove NP-hardness, we will show a reduction from the NP-complete 3-Set Cover problem. To this end, we will build a network that reflects the structure of a given 3-Set Cover problem instance, and use it as an input for the Optimal Attack problem. Finally, we will show that an optimal solution of the Optimal Attack problem corresponds to a solution of the given instance of the 3-Set Cover problem.

More formally, an instance of the NP-complete 3-Set Cover problem is defined by a universe $U=\{u_1, \ldots, u_m\}$, a collection of sets $S = \{S_1, \ldots, S_k\}$ such that $\forall_j S_j \subset U$ and $\forall_j |S_j| = 3$, and an integer $b \leq k$.
The goal is to determine whether there exist $b$ elements of $S$, the union of which equals $U$.
In what follows, let $\Su(u_i)=|\{S_j \in S : u_i \in S_j\}|$, \ie, $\Su(u_i)$ is the number of sets in $S$ that contain $u_i$.

First, let us create a network $G$ that reflects the structure of the instance of the 3-Set Cover problem under consideration. It is shown in Figure~\ref{fig:npc}, where:
\begin{itemize}
\item \textbf{The set of nodes:}
For every $S_i \in S$, we create a single node, denoted by $S_i$.
For every $u_i \in U$, we create a single node, denoted by $u_i$, as well as $4k$ nodes $a_{i,1}, \ldots, a_{i,k}$.
Moreover, for every $u_j \in U$ and every $S_i \in S$ such that $u_j \in S_i$, we create a single node $l_{i,j}$.
Additionally, we create a single node $\vs$.
\item \textbf{The set of edges:}
For every node $S_i \in S$ we create an edge $(S_i,\vs)$.
For every node $a_{i,j}$ we create an edge $(a_{i,j},u_i)$.
Finally, for every node $l_{i,j}$ we create two edges, $(S_i,l_{i,j})$ and $(u_j,l_{i,j})$.
\end{itemize}

Next, let $\sr^*$ be a distribution of security resources such that:
\begin{itemize}
\item $\sr^*(\vs) = 0$,
\item $\sr^*(a_{i,j}) = 0$ for every node $a_{i,j}$,
\item $\sr^*(u_i) = k - \Su(u_i)$ for every node $u_i$,
\item $\sr^*(S_i) = 20mk$ for every node $S_i$,
\item $\sr^*(l_{i,j}) = 20mk + 2$ for every node $l_{i,j}$.
\end{itemize}

This particular distribution of security resources has the following properties:

\begin{itemize}
\item it gives each node $u_i$ the same time of activation (assuming an equal number of activated neighbors), and 
\item it gives each of the nodes $S_i$ and $l_{i,j}$ the minimal time of activation that is greater than the maximal time of activation of nodes $u_i$, \ie, activating node $S_i$ or $l_{i,j}$ always takes more time than activating node $u_i$, no matter what the number of active neighbors is.
\end{itemize}

These properties will play an important role later in the proof. Furthermore, let $T^* = (b + m)(20mk+4) + 9mk + 1$ and $r^* = b + (4k+2)m$.
These are the time and the number of activated nodes of the strategy that corresponds to a solution to the 3-Set Cover problem.

\begin{figure}[t]
\centering
\includegraphics[width=.3\linewidth]{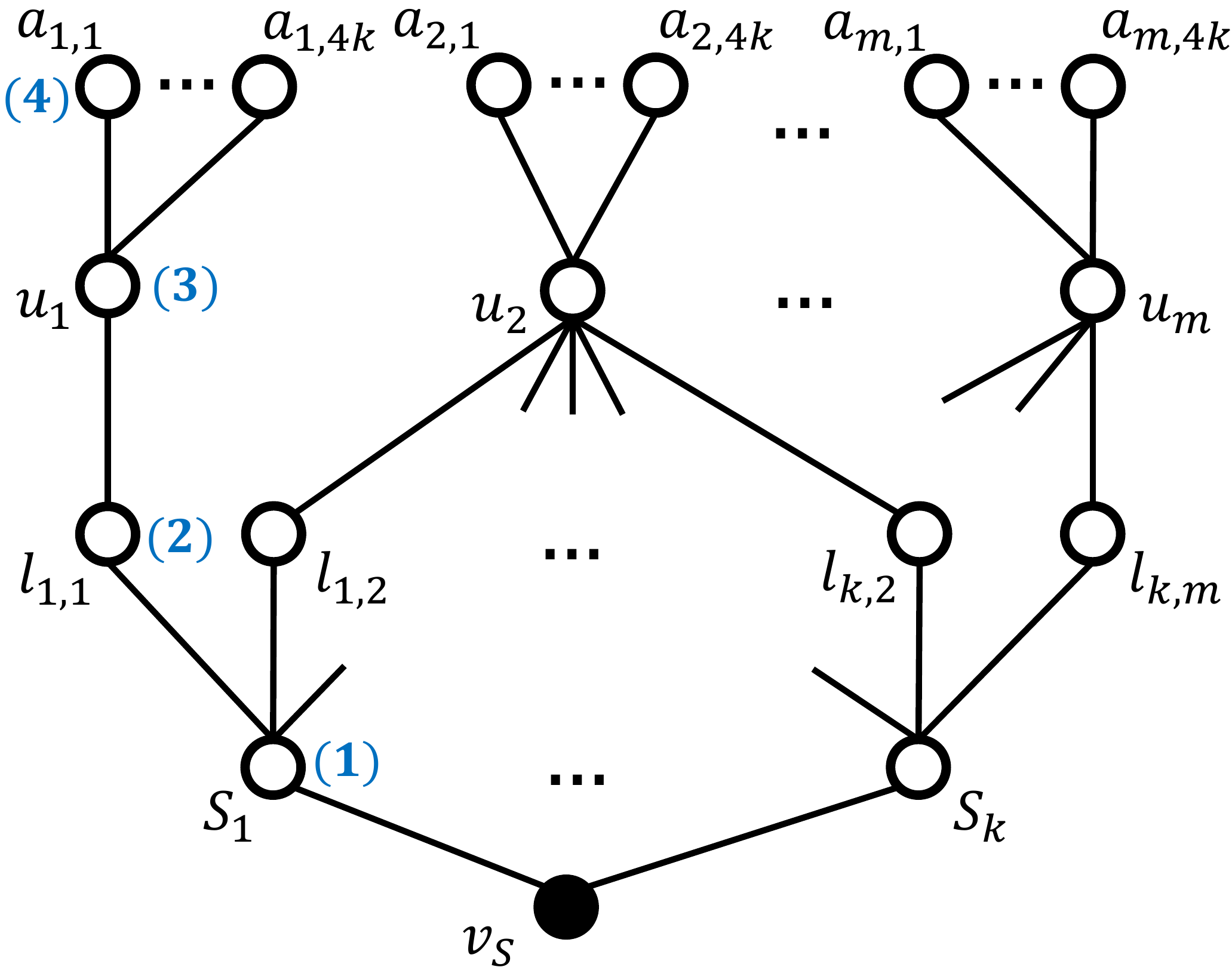}
\caption{The network used to reduce the 3-Set Cover problem to the Optimal Attack problem.
The seed node, $\vs$, is marked black.
Blue numbers in the parenthesis show the order of nodes that have to be activated before activating node $a_{1,1}$.}
\label{fig:npc}
\end{figure}

Now, consider the instance of the Optimal Attack problem in the form of $(G,\sr^*,\vs,T^*,r^*)$.
We will show next that an optimal solution to this instance corresponds to an optimal solution to the 3-Set Cover problem.
We begin by observing that the following holds for the times of activation of the nodes in $V$:
\begin{itemize}
\item $\et(a_{i,j}) = 1$ for every node $a_{i,j}$, as we have $d(a_{i,j}) + \sr^*(a_{i,j}) = 1$ and $|N(a_{i,j}) \cap I| \leq 1$;
\item $\et(u_i) \leq 5k$ for every $u_i \in U$, as we have $d(u_i) + \sr^*(u_i) = 5k$ and $|N(u_i) \cap I| \geq 1$;
\item $5mk+1 \leq \et(S_i) \leq 20mk+4$ for every $S_i \in S$, as we have $d(S_i) + \sr^*(S_i) = 20mk + 4$ and $1 \leq |N(S_i) \cap I| \leq 4$;
\item $10mk+2 \leq \et(l_{i,j}) \leq 20mk+4$ for every node $l_{i,j}$, as we have $d(l_{i,j}) + \sr^*(l_{i,j}) = 20mk + 4$ and $1 \leq |N(l_{i,j}) \cap I| \leq 2$.
\end{itemize}

We will now show that if there exists a solution $S^* \subset S$ to the given instance of the 3-Set Cover problem, then there exists a solution $\seq^*$ to the constructed instance of the Optimal Attack problem such that $\en(\seq^*,T^*) = r^* = b + (4k+2)m$.
Indeed, we can construct such solution by activating every node $S_i \in S^*$ ($b$ nodes activated in time $b(20mk+4)$), choosing for every $u_j \in U$ node $S_i \in S^*$ such that $u_j \in S_i$ and activating nodes $l_{i,j}$ and $u_j$ ($2m$ nodes activated in time $(20mk+4)m+5mk$), and finally activating all nodes $a_{i,j}$ ($4mk$ nodes activated in time $4mk$).

In what follows, let $\seq^*_T$ denote the maximum prefix of $\seq^*$ such that $\et(\seq^*_T) < T^*$. Now, we have to show that if there exists a solution $\seq^*$ to the constructed instance of the Optimal Attack problem, then there exists also a solution $S^* \subset S$ to the given instance of the 3-Set Cover problem.
To this end, we will show that for the prefix $\seq^*_T$ of any such solution $\seq^*$ it must be that, for every $u_j \in U$, there exists a node $S_i \in \seq^*_T$ such that $u_j \in S_i$ and $|\seq^*_T \cap S| \leq b$.
Notice that when this is the case we can obtain the solution to the given instance of the 3-Set Cover problem by taking $S^* = \seq^*_T \cap S$.

First, we observe that any sequence $\seq^*_T$, that does not contain all nodes $u_i$, cannot have the required number of $r^*$ activated nodes within the time limit.
Assume to the contrary, that there exists such a sequence.
Since there exists node $u_i$ that is not activated, neither of the nodes $a_{i,j}$ are activated.
Therefore this sequence has to activate $4k+1$ of the nodes $S_i$ or $l_{i,j}$ instead.
However, there are only $4k$ nodes $S_i$ and $l_{i,j}$.
Hence, sequence $\seq^*_T$ that is a solution of the constructed instance of the Optimal Attack problem must activate all nodes $u_j$.
In order to activate given node $u_j$ we need to activate at least one node $l_{i,j}$, and in order to achieve that we need to activate node $S_i$.
Therefore, because of the way we constructed the network, for every $u_j \in U$ there exists a node $S_i \in \seq^*_T$ such that $u_j \in S_i$.

Now we need to show that for $\seq^*_T$ being a solution of the constructed instance of the Optimal Attack problem we have $|\seq^*_T \cap S| \leq b$, \ie, that any sequence being solution cannot activate more than $b$ nodes $S_i$ within time limit.
As shown above, $\seq^*_T$ has to activate all $u_j$ nodes.
Notice that in order to activate $u_j$ we have to activate a node $l_{i,j}$ with only one activated neighbor (its only other neighbor being $u_j$).
Activating $m$ such nodes takes $m(20mk+4)$ time.
Hence, in order to activate more than $b$ nodes $S_i$ within time limit we would have to do it in time shorter (as we still have to activate nodes $u_j$) than $T^*-m(20mk+4)=b(20mk+4) + 9mk$.
Now, activating node $S_i$ (including the time necessary to activate neighboring nodes $l_{i,j}$) takes time:
\begin{itemize}
\item $20mk+4$ when activating with only one active neighbor $\vs$;
\item $\frac{20mk+4}{2}+(20mk+4)$ when activating with two active neighbors;
\item $\frac{20mk+4}{3}+2(20mk+4)$ when activating with three active neighbors;
\item $\frac{20mk+4}{4}+3(20mk+4)$ when activating with four active neighbors.
\end{itemize}
This is because, in order to be used to speed up the activation of node $S_i$, node $l_{i,j}$ has to be activated beforehand, in time $20mk+4$.
Hence, activating a single node $S_i$ takes at least $20mk+4$ time and it is not possible to activate more than $b$ nodes $S_i$ in time $b(20mk+4) + 9mk$.

This implies that the optimal solution to the constructed instance of the Optimal Attack problem must correspond to the optimal solution to the given instance of the 3-Set Cover problem, thus concluding the proof.
\end{proof}

\subsection{Optimal Strategies for Stars}
\label{app:stars}
For stars, we have the  following result:
\begin{theorem}
\label{thrm:star}
Let network $G=(V,E)$ be a star with center $a$ and peripheral nodes $b_1,\ldots,b_{n-1}$. Let $\sr$ be a particular distribution of security resources (assume that nodes $b_i$ are ordered non-decreasingly according to $\sr(b_i)$).
An optimal attack sequence is then:
\begin{itemize}
\item $\seq^* = \langle a, b_1, b_2, \ldots, b_{n-1} \rangle$ if $\sr(b_1) \geq \sr(a)+n-2$ (i.e., when $a$ has the lowest time of activation)
\item $\seq^* = \langle b_1, a, b_2, \ldots, b_{n-1} \rangle$ otherwise.
\end{itemize}
An optimal defense strategy against an optimal attack strategy is then:
\begin{itemize}
\item $\sr^*(a) = \frac{\SR-(n-1)(n-2)}{n}$ and $\sr^*(b_i) = \frac{\SR+n-2}{n}$ if $T \leq \frac{\SR+2(n-1)}{n}$,
\item $\sr^*(a) = 0$ and $\sr^*(b_i) = \frac{\SR}{n-1}$ if $\frac{\SR+2(n-1)}{n} < T \leq \frac{\SR+n-1}{n-1}$,
\item $\sr^*(a) = \SR$ and $\sr^*(b_i) = 0$ otherwise.
\end{itemize}
\end{theorem}

In words, an optimal strategy of the attacker is to attack nodes in non-decreasing order of their activation times, with the additional constraint that the center of the star has to be one of the first two nodes in the sequence.
As for the defender strategy, if it is possible not to let the attacker activate even a single node given the time limit, then the defender should try to make the time of activation equal for all nodes (if possible) or all peripheral nodes (if there are not enough security resources).
Alternatively, the defender should assign all security resources to the central node (since again, it has to be one of the first two nodes in the sequence).

\begin{proof}

First we prove our claim about the strategy of the attacker. Notice that node $a$ has to be either first or second node in the attack sequence, as the attacker either starts with it, or it is the only neighbor of the starting node (if the attacker starts with one of the peripheral nodes).
Hence, it has activation time of $\et(a) = n - 1 + \sr(a)$. 
Similarly, any node $b_i$ is either first in the sequence or is activated when exactly one of its neighbors (central node $a$) is active.
Hence, its activation time is $\et(b_i) = 1 + \sr(b_i)$. If it is impossible to activate more than one node within time limit, the sequence $\seq^*$ is optimal because the first node has the lowest time of activation in the entire network.

Now consider a case where it is possible to activate at least two nodes within the time limit.
Assume to the contrary, that in an optimal attack sequence, $\seq$, we have $b_j$ activated before $b_i$ for $i < j$ (as mentioned above, node $a$ has to be either on the first or the second position in the sequence and we can swap these nodes without changing the total activation time).
However, such sequence can be improved by swapping elements $b_i$ and $b_j$.
The argument follows the same logic as in the proof of Theorem~2.

Now we move to proving our claim about the strategy of the defender. 
First, consider a case in which it is possible not to let attacker activate even a single node, \ie, $\forall_{v \in V} \et(v) \geq T$.
It is then optimal for the defender  to have the same activation time for all nodes (as the attacker will pick the one with lowest activation time as the first in sequence).
Hence, we have $\et(a) = \sr^*(a)+n-1 = \sr^*(b_i)+1 = \et(b_i)$ and $\sr^*(a)+(n-1)\sr^*(b_i)=\SR$.
After solving this set of equations we get $\sr^*(a) = \frac{\SR-(n-1)(n-2)}{n}$ and $\sr^*(b_i) = \frac{\SR+n-2}{n}$.
We then have $\et(a) = \et(b_i) = \frac{\SR+2(n-1)}{n}$.
Since we consider a case in which it is possible not to let attacker activate even a single node, we need to have $\et(a) = \et(b_i) = \frac{\SR+2(n-1)}{n} \geq T$.

Notice that it is possible that the defender does not have enough defense resources to make time of activation of all nodes equal (when $\SR<(n-1)(n-2)$).
She should then spread them uniformly among peripheral nodes, to maximize the minimal time of activation in the network.
We then have $\sr^*(a) = 0$ and $\sr^*(b_i) = \frac{\SR}{n-1}$.
Minimal time of activation in the network is then $\et(b_i) = \frac{\SR+n-1}{n-1}$ and in order for this strategy to be optimal we need to have $\et(b_i) = \frac{\SR+n-1}{n-1} \geq T$.

Now, consider the case in which the attacker is able to activate $k \geq 1$ nodes.
One of the first two nodes in the sequence has to be $a$.
Hence, assigning all security resources to $a$ either maximizes the time required to activate available number of nodes---if $k \geq 2$ as their total activation time is $\sr(a)+\sum_{i=1}^{k-1} \sr(b_i)$), or it maximizes the time that would be necessary to activate second node in the sequence---if $k=1$.
\end{proof}

\subsection{Optimal Defense Against a Subset of Strategies}
\label{app:milp}

In this section, we present a mixed-integer linear programming method to compute optimal defense strategies when the strategies of the attacker can be enumerated, \ie, when the set of the attacker's strategies is of reasonable size.

Assume that the attacker has at her disposal only a limited set of strategies, $\avs \subseteq \Seq(V)$.
We assume that each sequence in this set is a correct strategy of the attack, \ie, there is no attempt to activate nodes with no active neighbors.
In such case the problem of finding the optimal distribution of defense resources can be formulated as a mixed-integer linear programming problem.
The formulation is:
$$
\begin{array}{rll}
\max_{\sr_v, a_{\seq,i}, k} k & & \\
\text{subject to} \quad\quad\quad \sr_v & \geq 0 & \forall_{v \in V} \\
\sum_{v \in V} \sr_v & \leq \SR & \\
a_{\seq,i} & \in \{0,1\} & \forall_{\seq \in \avs} \forall_{i \in I} \\
\sum_{j \leq i} \frac{d(\seq_j)+\sr_{\seq_j}}{c_{\seq,j}} & \geq a_{\seq,i} T & \forall_{\seq \in \avs} \forall_{i \in I} \\
\sum_{i \in I} a_{\seq,i} & \geq k & \forall_{\seq \in \avs}
\end{array}
$$
\noindent where $I = \{1, \ldots, n\}$ and $c_{\seq,j}$ is the number of active neighbors of $\seq_j$ at the moment of its activation, when using attack sequence $\seq$ (we set $c_{\seq,1}=1$ for every $\seq \in \avs$).

The first two sets of constraints guarantee that we have a valid distribution of security resources. When the attacker is using attack sequence $\seq$, we want variable $a_{\seq,i}$ to be equal $1$ if and only if node $\seq_i$ is not activated within the time limit.
This is guaranteed with the fourth set of constraints. In particular, we have $a_{\seq,i}=1$ only when the total time of activation of the first $i$ nodes in the sequence (the left-hand side of the fourth constraint) exceeds the time limit, $T$.
The last set of constraints guarantees that $k$ is the minimum number of inactive nodes over all the choices of $\seq$ (the attacker minimizes this value, while the defender maximizes it).
The formulation has $\mathcal{O}(|\avs|n)$ constraints.

In the next section, we explore the efficiency of various heuristic strategies that may be used by both players.

\section{Attack and Defense Strategies based on Heuristics}

\noindent
Another approach that may be pursued by both the defender and the attacker is to resort to heuristic strategies.
In fact, one can think of many potential heuristics, depending on the network topology and the particular setting at hand.
Hence, in this section we propose a number of different heuristics that are inspired by a variety of scenarios and the results from the literature.

\subsection{Heuristic Strategies of the Defender}

\noindent
To find a suitable heuristic, we should take into account the information possessed by the players.
In particular, the defender does not know \textit{a priori}, where the attack is going to take place.
Hence, the heuristic strategy of assigning available security resources should be based on some properties of either the security game or the network as no other information is available at this stage. The formulas of the defender's strategies are presented in Table~\ref{tab:defender-formulas}.
In these formulas, $r_i$ is drawn uniformly at random from $[0,1]$, $s^*$ is the maximal $s$ such that $\sum_{i=1}^{n}\max(0,s-d(v_i)) \leq \SR$, and $c$ denotes a given centrality measure for centrality-based strategies.
\begin{table}[t]
\begin{tabular}{lcl}
\hline
Equality & & $\sr(v_i) = \max(0,s^*-d(v_i))$ \\
Uniform & \ \ \ \ \ \ \ \ \ \ \ \ \ \ \ \ \ \ \ \ & $\sr(v_i) = \frac{1}{n} \SR$ \\
Random & & $\sr(v_i) = \frac{r_i}{\sum_{i=1}^{n} r_i} \SR$ \\
High Centrality $c$ & & $\sr(v_i) = \frac{c(v_i)}{\sum_{i=1}^{n} c(v_i)} \SR$ \\
Low Centrality $c$ & & $\sr(v_i) = \frac{c^{-1}(v_i)}{\sum_{i=1}^{n} c^{-1}(v_i)} \SR$ \\
\hline
\end{tabular}
\caption{Formulas for the defender's heuristics.}
\label{tab:defender-formulas}
\end{table}
We now present an intuition behind each heuristic for the defender:
\begin{itemize}
\item \emph{Equality}---the defender tries to make the activation time equal for as many nodes as possible (as it may not be possible for all nodes). This heuristic is inspired by the work of Bachrach~\etal~\cite{bachrach2013contagion} who found that in their model the optimum is achieved when security levels are equal for all targets:
\item \emph{Uniform}---the security resources are distributed uniformly among all nodes. This heuristic is related to the previous one but now the resources are simply divided equally.
\item \emph{Random}---the security resources are distributed randomly among all nodes. This heuristic is a natural benchmark for other ones.
\end{itemize}

Next, we propose twelve heuristics for the defender that are based on the topology of the network.
The basic idea is to assign available security resources either proportionally or inverse proportionally to the centrality, \ie, the importance, of nodes in the network.
The following set of heuristics for the defender are based on the centrality measures:

\begin{itemize}
\item \emph{High Degree}---the defender focuses on defending nodes with high degrees.
The inspiration for this strategy can be found in a number of works that indicate that hubs are of the key importance in the diffusion process~\cite{pastor2001epidemic, dezsHo2002halting}.

\item \emph{Low Degree}---the defender focuses on defending the nodes with low degrees.
This strategy is inspired by the observation that some criminal groups tend to recruit lonely, socially-isolated people~\cite{ceballo2000neighborhood}.

\item \emph{High Betweenness}---the defender focuses on defending the nodes with high betweenness.
Betweenness centrality is considered particularly important in the context of financial networks, where it is arguably the most sensible measure of the systemic danger of an individual bank within the financial system~\cite{freixas1998contagion}.
In particular, banks with high betweenness centrality are able to bring down the entire financial system if consecutive occurrences of illiquidity materialize, given that there are no interventions by relevant authorities.
\end{itemize}

The next two strategies are inspired by the observation that nodes with low betweenness centrality and (typically) high closeness centrality play the role of the ``pulse-takers'' in their organizations~\cite{henderson2013department}.
They are easily accessible by other central nodes as well as to the rest of the network.
The pulse-takers are key to the preservation and the development of companies~\cite{Stephenson:2004} and as such are more likely to be a subject of various attacks, such as an advanced ransomware attack~\cite{Vanderburg:2017}.

\begin{itemize}
\item \emph{Low Betweenness}---the defender focuses on defending the nodes with low betweenness.

\item \emph{High Closeness}---the defender focuses on defending the nodes with high closeness.
\end{itemize}

We also consider the heuristic strategy focused on the nodes with low closeness.
In the social network context, such nodes tend to be vulnerable (easier to target) as more isolated and dependent on information from few neighbours~\cite{rana2015centrality}.

\begin{itemize}
\item \emph{Low Closeness}---the defender focuses on defending the nodes with low closeness.
\end{itemize}

Finally, we consider heuristics based on the Shapley value of each node (see Section~\ref{sec:preliminaries} for more details). As the Shapley value expresses synergy between nodes, it might prove effective in choosing which areas of the network to defend more intensively.
The Shapley value of a given game can be used to define centrality measure, where the centrality score of each node is the payoff assigned to it by the Shapley value.
In this article we consider the following heuristics based on the Shapey value:

\begin{itemize}
\item \emph{High/Low Shapley Value Degree}---the defender focuses on defending the nodes with high/low payoff in a game where the value of the coalition $C$ is the number of nodes in $C$ and those immediately reachable from $C$~\cite{michalak2013efficient}.
    
\item \emph{High/Low Shapley Value Closeness}---the defender focuses on defending the nodes with high/low payoff in a game where the value of the coalition $C$ is the number of nodes in $C$ and those not further than $\delta$ hops away (in our experiments we use $\delta=3$)~\cite{michalak2013efficient}.

\item \emph{High/Low Shapley Value Betweenness}---the defender focuses on defending the nodes with high/low payoff in a game where the value of the coalition $C$ is the percentage of the shortest paths between the nodes from outside $C$ controlled by the nodes in $C$~\cite{szczepanski2016efficient}.
\end{itemize}

\subsection{Heuristic Strategies of the Attacker}
\label{sec:attacker-strategies}

\noindent
Unlike the defender, who has to choose her moves \textit{a priori}, the attacker knows at each step of her decision-making process which nodes in the network have been already activated and which have not, as well as she knows the distribution of security resources chosen by the defender.
This immediately suggests the following two heuristics in which the attacker may use this information to her advantage:

\begin{itemize}
\item \emph{Greedy}---the heuristic in which the next target is a not-yet-activated node with the lowest time of activation; and
\item \emph{Majority}---the heuristic in which the next target is a not-yet-activated node with the highest number of activated neighbours.
\end{itemize}

In comparison to the Greedy heuristic, the Majority heuristic has a tendency to activate nodes with high degrees earlier in the process.
One can also mix both heuristics:

\begin{itemize}
\item \emph{Mixed}---a mix of the two previous heuristics that we will denote by Mixed$(p)$, where $p$ is the probability of using the Greedy heuristic and $1-p$ is the probability of using the Majority heuristic.
\end{itemize}

The above three heuristics were already studied in strategic network diffusion~\cite{alshamsi2018optimal} however not in the setting that involves a strategic defender or another player. 

A parallel can be drawn between using the Greedy and Majority heuristics and the tradeoff between \textit{exploration} and \textit{exploitation}~\cite{kaelbling1996reinforcement}.
In particular, the Greedy strategy can be considered the exploitation component of the setting, as it maximizes short-term gains.
At the same time, the Majority heuristic has a tendency to target high degree nodes, thus exploring new potential venues of an attack for the future, although the immediate cost of such an attack is usually higher than for the Greedy heuristic. Following this trail, we propose yet another two heuristics for the attacker that are inspired by the exploration-exploitation literature~\cite{sutton1998reinforcement}:

\begin{itemize}
\item \emph{Epsilon-Decreasing}---a heuristic that exhibits exploratory behaviour at the beginning and exploitative behaviour at the end of the process.
In more detail, we set the probability of selecting the target using the Greedy heuristic to $p = \frac{t}{T}$, where $t$ is the time elapsed since the beginning of the attack process and $T$ is the total time limit.
Otherwise, we select the target using the Majority heuristic.
\item \emph{Epsilon-First}---a heuristic that at first exhibits only exploratory behaviour, followed by only exploitative behaviour. In more detail, we select the target using the Majority heuristic if $t \leq \epsilon T$, where $\epsilon$ is the parameter of the heuristics, while $t$ and $T$ are defined as for the Epsilon-Decreasing heuristic. Otherwise we select the target using the Greedy heuristic.
\end{itemize}

In the next section, we will evaluate all the above heuristic strategies by the means of numerical simulations.

\section{Experimental Analysis}
\label{sec:experimental-analysis}

\noindent
We will now test the heuristics proposed in the previous section across different combinations of networks and strategies.

\subsection{Experiment Design}\label{sec:experiment-design}

\noindent
The setting of each of our experiments is characterised by (a) a network, $G$, (b) a strategy of the defender, $\sr$, and (c) a strategy of the attacker $\seq$. As for (a), we use randomly generated networks, as well as real-life datasets. As for (b)---\ie, the defender's strategy---we consider all the defense heuristics proposed in the previous section. 
In our experiments, we assume that the amount of security resources is $\SR = 10 n$ and the time limit is $T = n$.
These parameters' values allow us to clearly present trends appearing for different settings of our experiments. As for (c)---\ie, the attacker's strategy---we again consider all attack heuristics proposed in the previous section.
We benchmark them against the optimal attack strategy which we find using the dynamic-programming algorithm.
However, the dynamic-programming algorithm is time and memory intensive and we are able to apply it in practice only for smaller networks.
Furthermore, following the usual assumption in the security games literature~\cite{bachrach2013contagion, acemoglu2016network, lou2017multidefender}, we choose the seed node of the attacker that provides the highest number of activated nodes (breaking ties uniformly at random).

\paragraph{Network Generation Models} We use the following models to generate the networks:

\begin{itemize}
\item \emph{Preferential attachment} networks generated using the Barab{\'a}si-Albert model~\cite{barabasi1999emergence}.
Unless stated otherwise, in our experiments we add $d=3$ links with each new node and we set the size of the initial clique to $d$;

\item \emph{Scale-free} networks generated using the configuration model \cite{newman2003structure}.
Unless stated otherwise, in our experiments we assume the minimal degree to be $d_{min}=2$, the maximal degree to be $d_{max}=10$, and the configuration model parameter to be $\lambda=3$;

\item \emph{Random graphs} generated using the Erd{\H{o}}s-R{\'e}nyi model \cite{erdds1959random}.
Unless stated otherwise, in our experiments we assume the expected average degree to be $d=10$;

\item \emph{Small world} networks generated using the Watts-Strogatz model \cite{watts1998collective}.
Unless stated otherwise, in our experiments we assume the expected average degree to be $d=10$ and the probability of rewiring to be $p=\frac{1}{4}$;

\item \emph{Random trees} generated using Pr\"ufer sequences~\cite{prufer1918neuer}.
We use sequences where each element is chosen uniformly at random from set $1,\ldots,n$.
\end{itemize}

\paragraph{Real-Life Networks} In addition to the experiments on randomly generated networks, we also perform simulations for a number of real-life network structures. In particular:

\begin{itemize}
\item \textit{Karate Club}~\cite{zachary1977information}---the network of people attending karate classes in one of the American universities.
The network consists of 34 nodes and 78 edges.

\item \textit{Facebook (small)}~\cite{leskovec2012learning}---a small fragment of the social network of the users of Facebook.
The network consists of 61 nodes and 272 edges.

\item \textit{Facebook (medium)}~\cite{leskovec2012learning}---a medium fragment of the social network of the users of Facebook.
The network consists of 333 nodes and 2523 edges.

\item \textit{Facebook (large)}~\cite{leskovec2012learning}---a large fragment of the social network of the users of Facebook.
The network consists of 786 nodes and 14027 edges.

\end{itemize}

The four datasets can be considered examples of networks appearing in different applications of our model (they are standard social networks generated based on surveys and social media data) since all nodes are inactive at the beginning.

\subsection{Numerical Results}

\noindent
We now move to describing the results of our simulations.
Results for randomly generated networks with any combination of attacker's and defender's strategies are presented in Figures~\ref{fig:heat-80} and~\ref{fig:heat-1000}, whereas Figure~\ref{fig:best-bars} presents results under assumption that the attacker chooses the optimal strategy (in the case of smaller networks) or the best of considered heuristic strategies.
Results for real-life networks are presented in Figures~\ref{fig:heat-reallife} and \ref{fig:best-bars-realife-soc}. 
We describe results of each type of networks separately.

\begin{figure}[t]
\centering
\setlength\tabcolsep{0pt}
\begin{tabular}{m{.32\textwidth}m{.32\textwidth}m{.32\textwidth}}
\multicolumn{1}{c}{Preferential attachment networks} &
\multicolumn{1}{c}{Scale free networks} &
\multicolumn{1}{c}{Random graph networks} \\
\includegraphics[width=\linewidth]{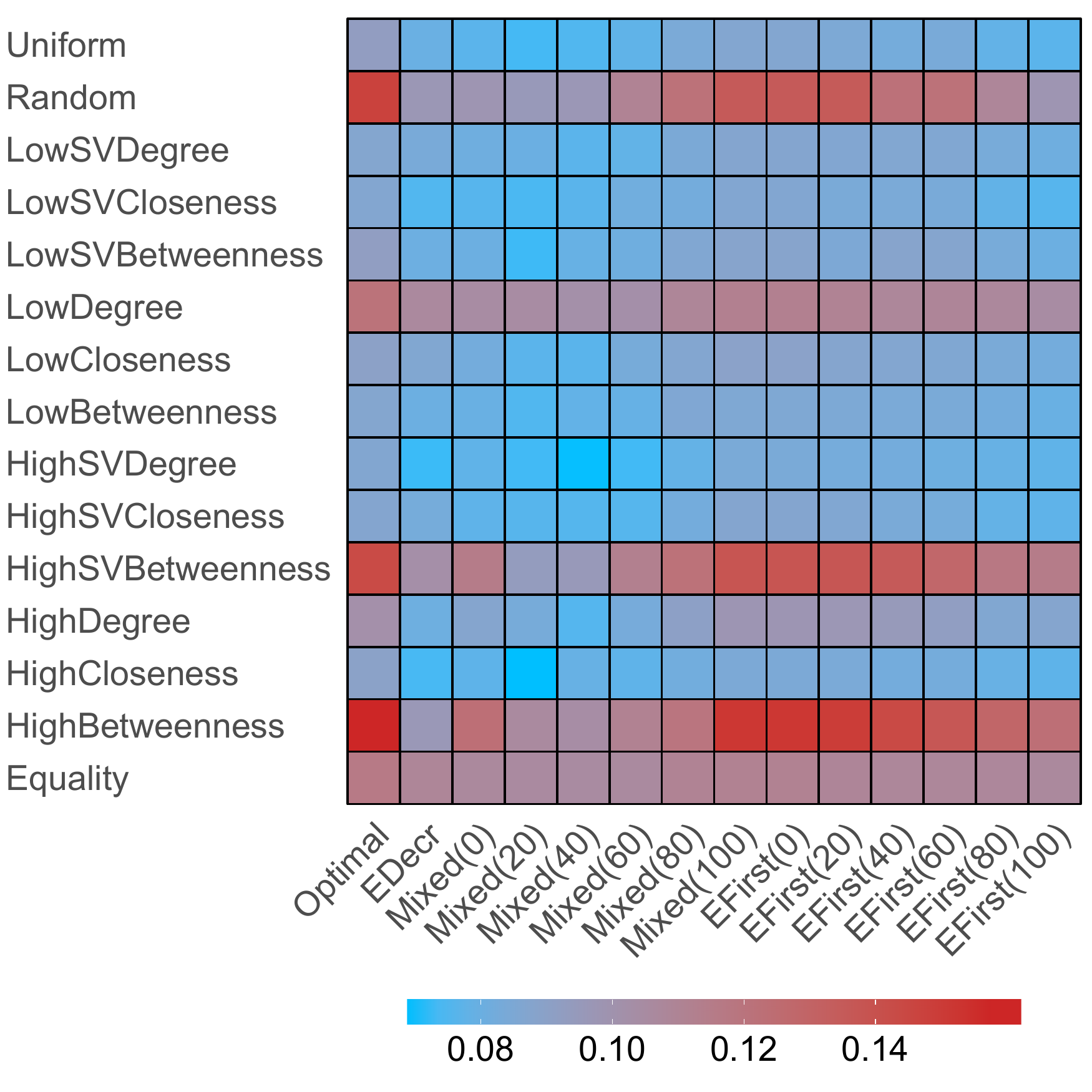} &
\includegraphics[width=\linewidth]{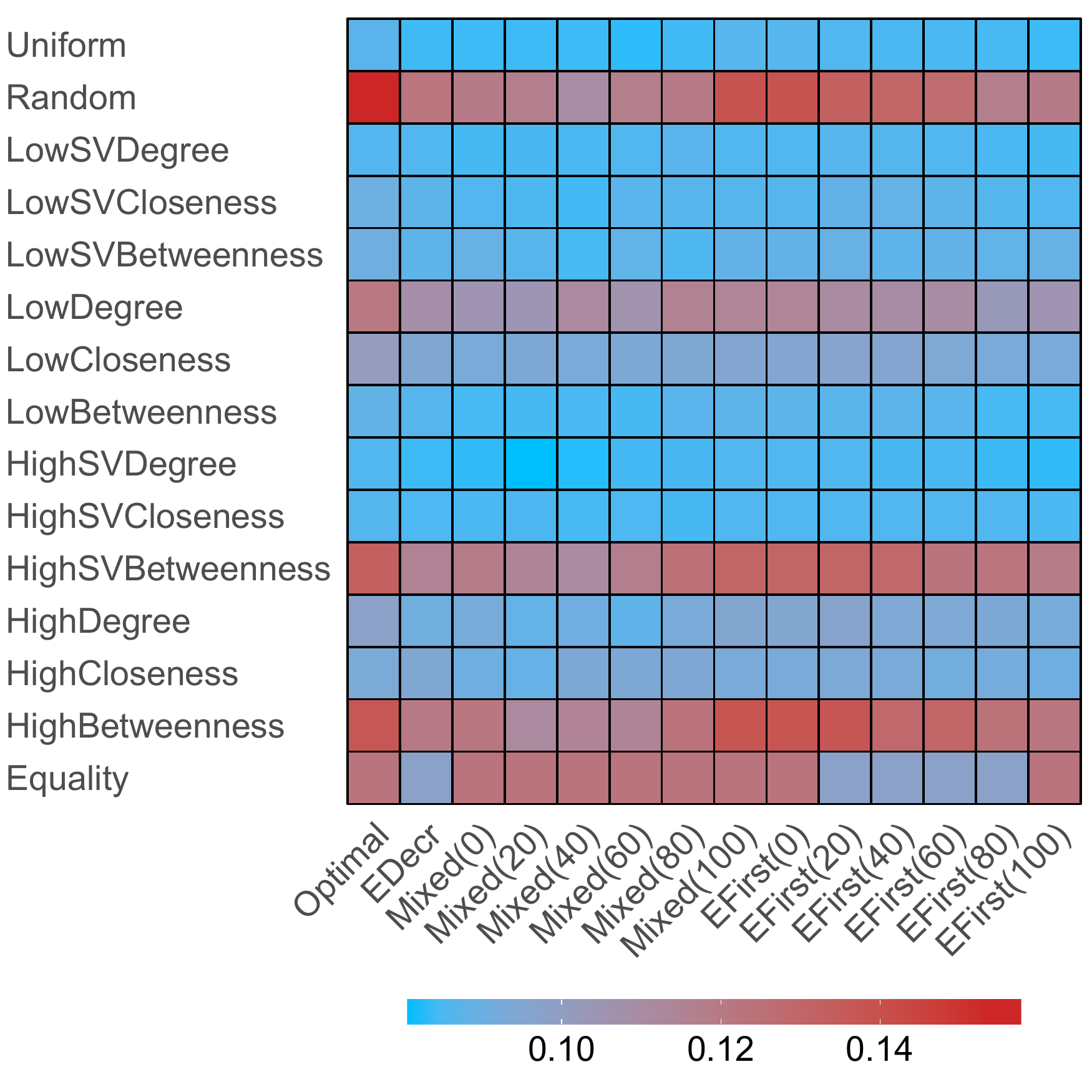} &
\includegraphics[width=\linewidth]{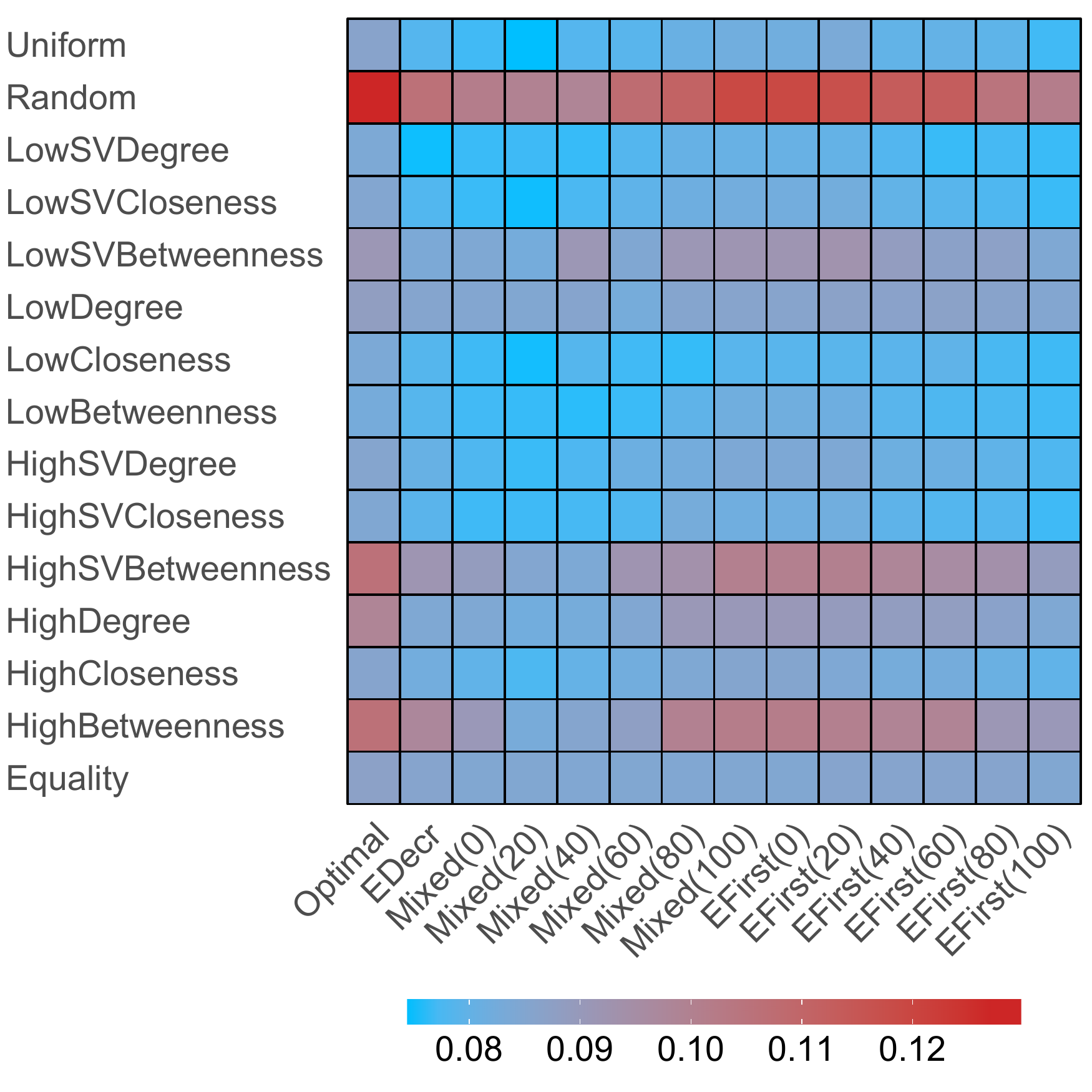} \\
\end{tabular}
\begin{tabular}{m{.32\textwidth}m{.32\textwidth}}
\multicolumn{1}{c}{Small world networks} &
\multicolumn{1}{c}{Random trees} \\
\includegraphics[width=\linewidth]{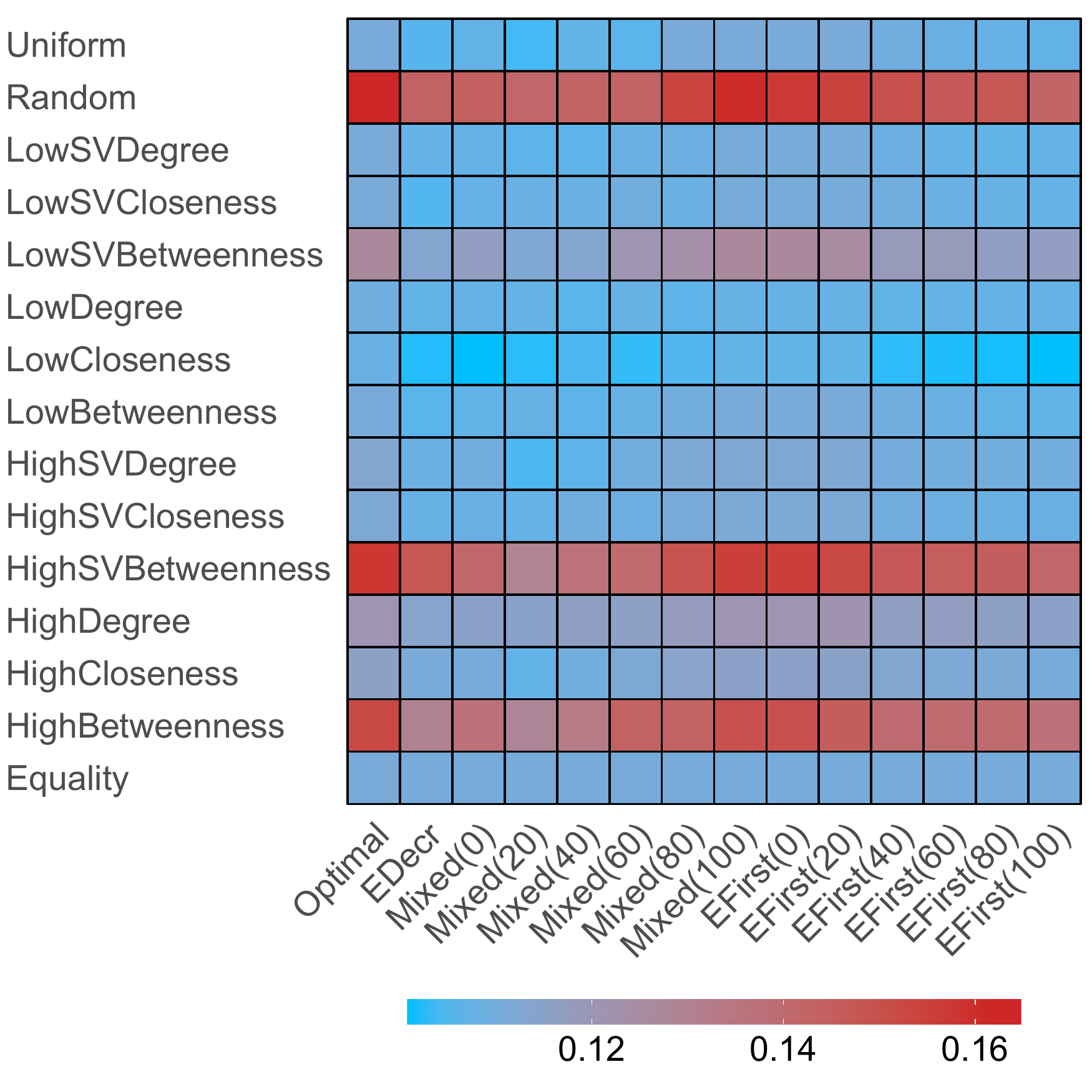} &
\includegraphics[width=\linewidth]{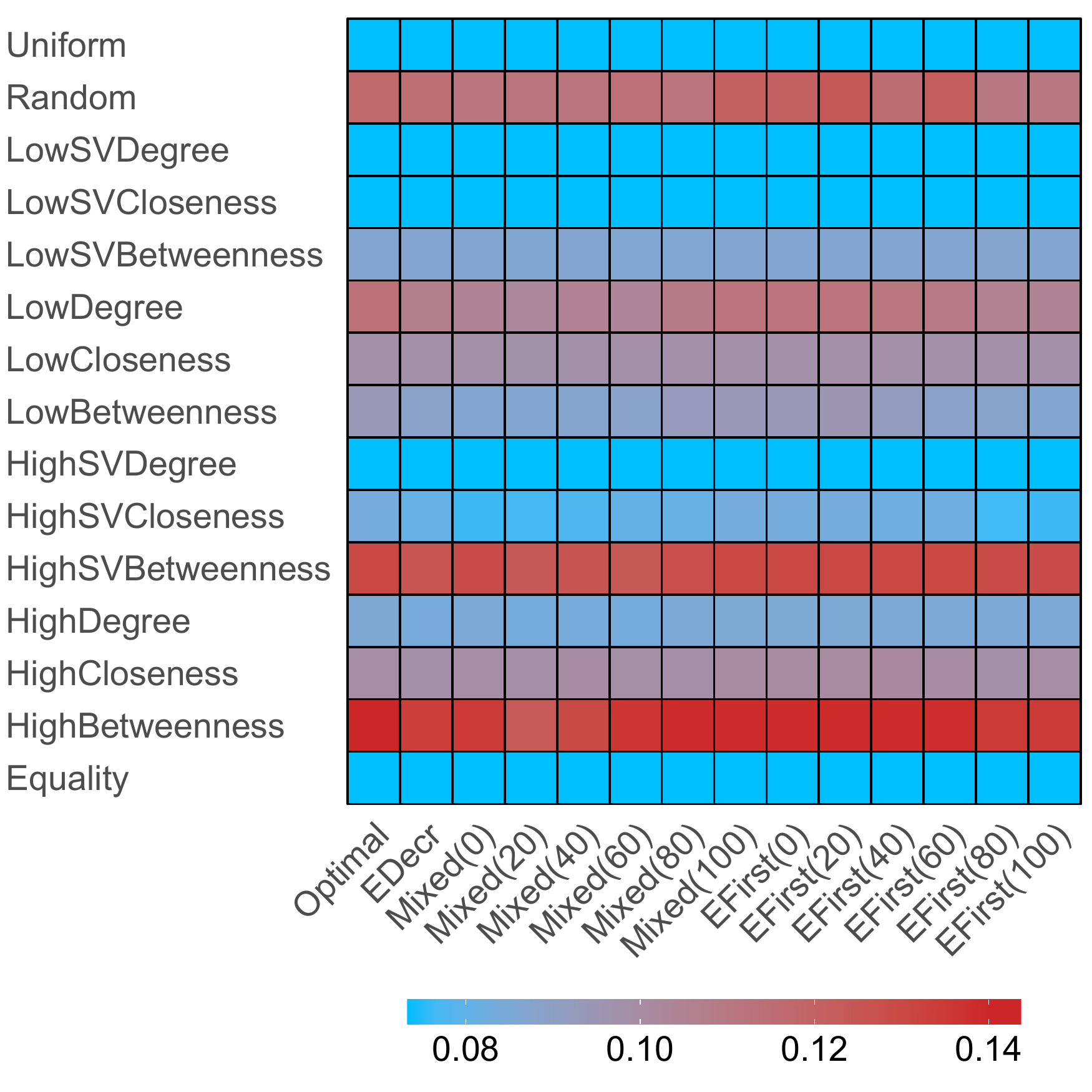} \\
\end{tabular}
\caption{
Comparison of effectiveness of defender's strategies for different attacker's strategies on networks with $80$ nodes.
Color of each cell represents the expected percentage of nodes successfully activated by the attacker.
Results are taken as an average over $100$ simulations, with a new network generated for each simulation using one of the models.
}
\label{fig:heat-80}
\end{figure}

\begin{figure}[t]
\centering
\setlength\tabcolsep{0pt}
\begin{tabular}{m{.32\textwidth}m{.32\textwidth}m{.32\textwidth}}
\multicolumn{1}{c}{Preferential attachment networks} &
\multicolumn{1}{c}{Scale free networks} &
\multicolumn{1}{c}{Random graph networks} \\
\includegraphics[width=\linewidth]{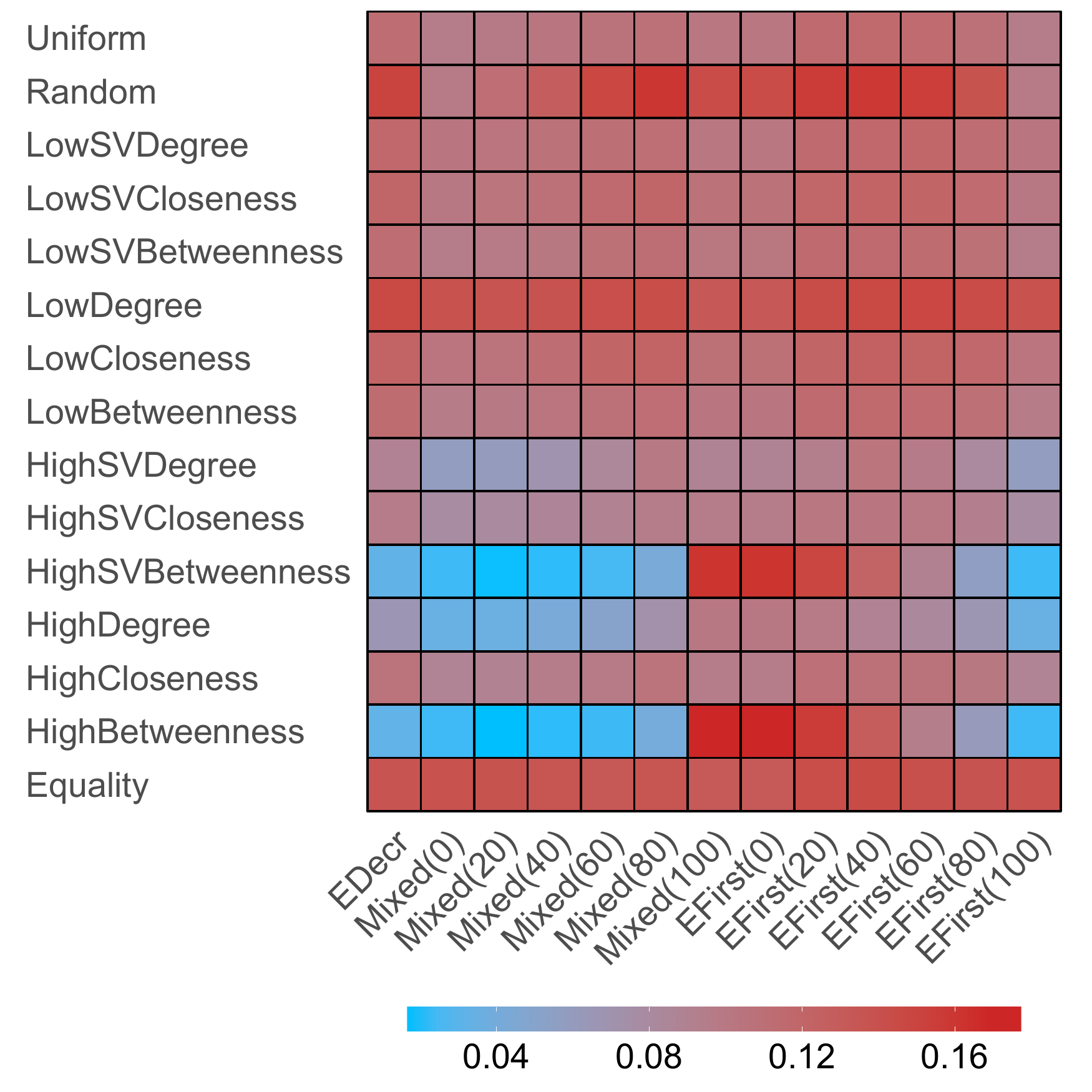} &
\includegraphics[width=\linewidth]{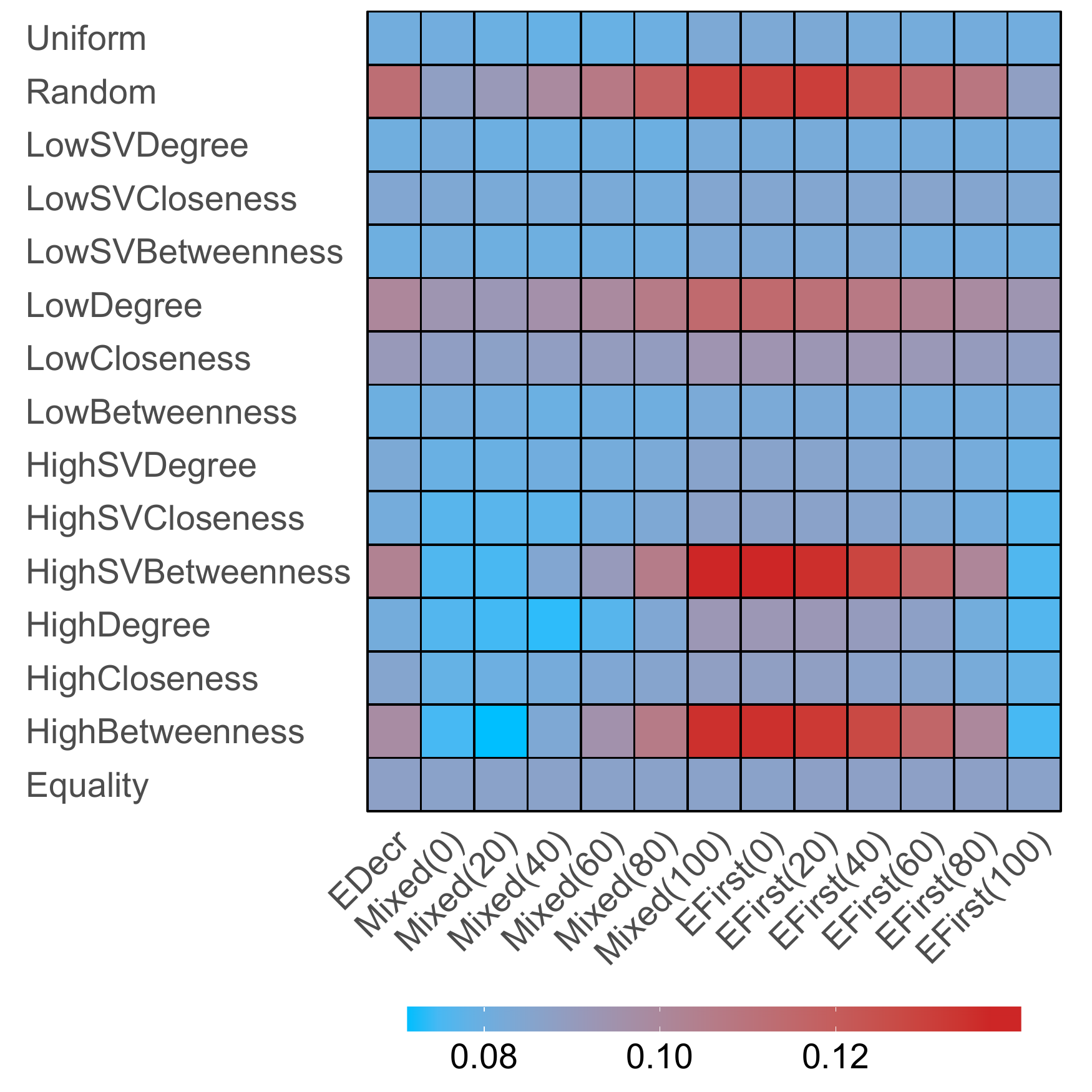} &
\includegraphics[width=\linewidth]{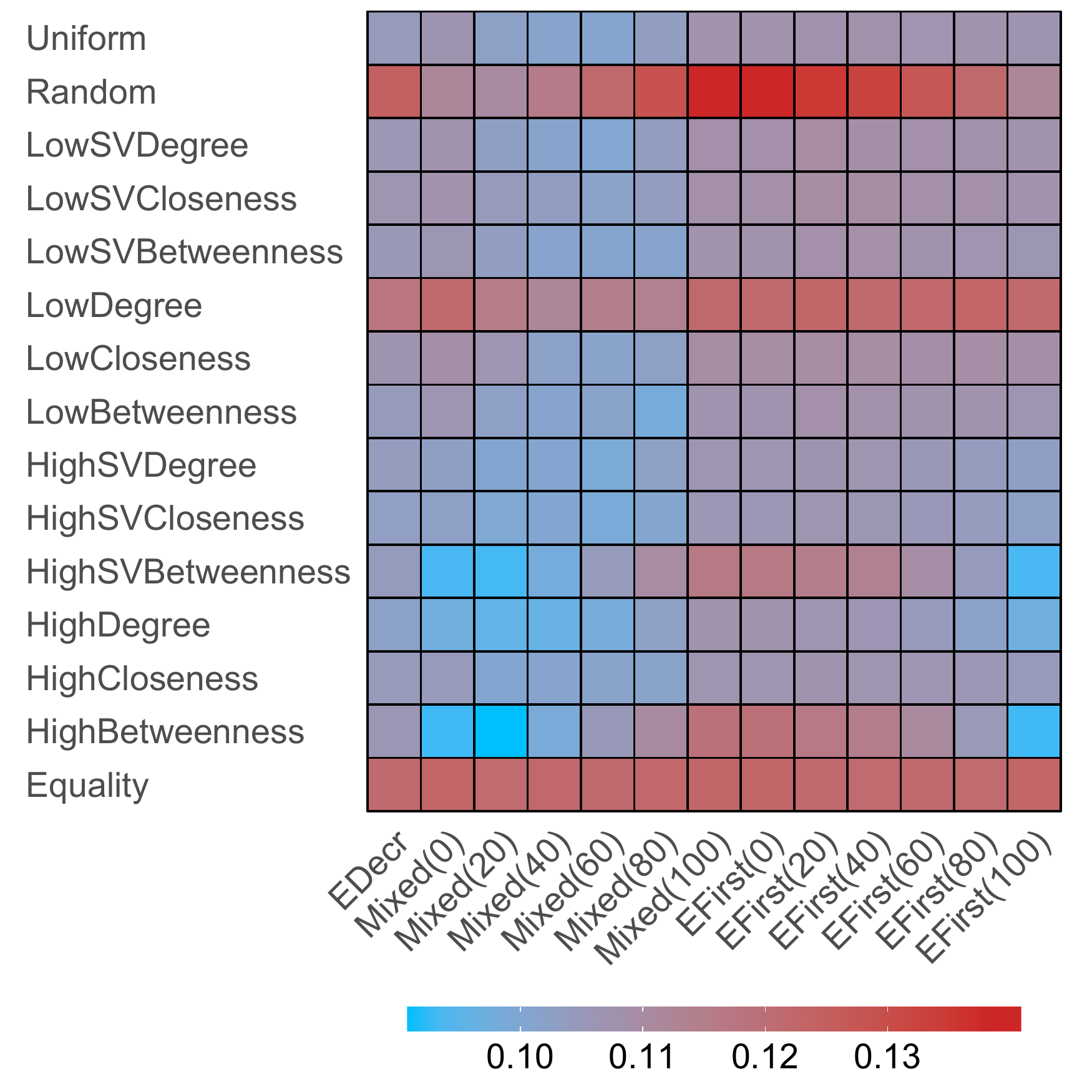} \\
\end{tabular}
\begin{tabular}{m{.32\textwidth}m{.32\textwidth}}
\multicolumn{1}{c}{Small world networks} &
\multicolumn{1}{c}{Random trees} \\
\includegraphics[width=\linewidth]{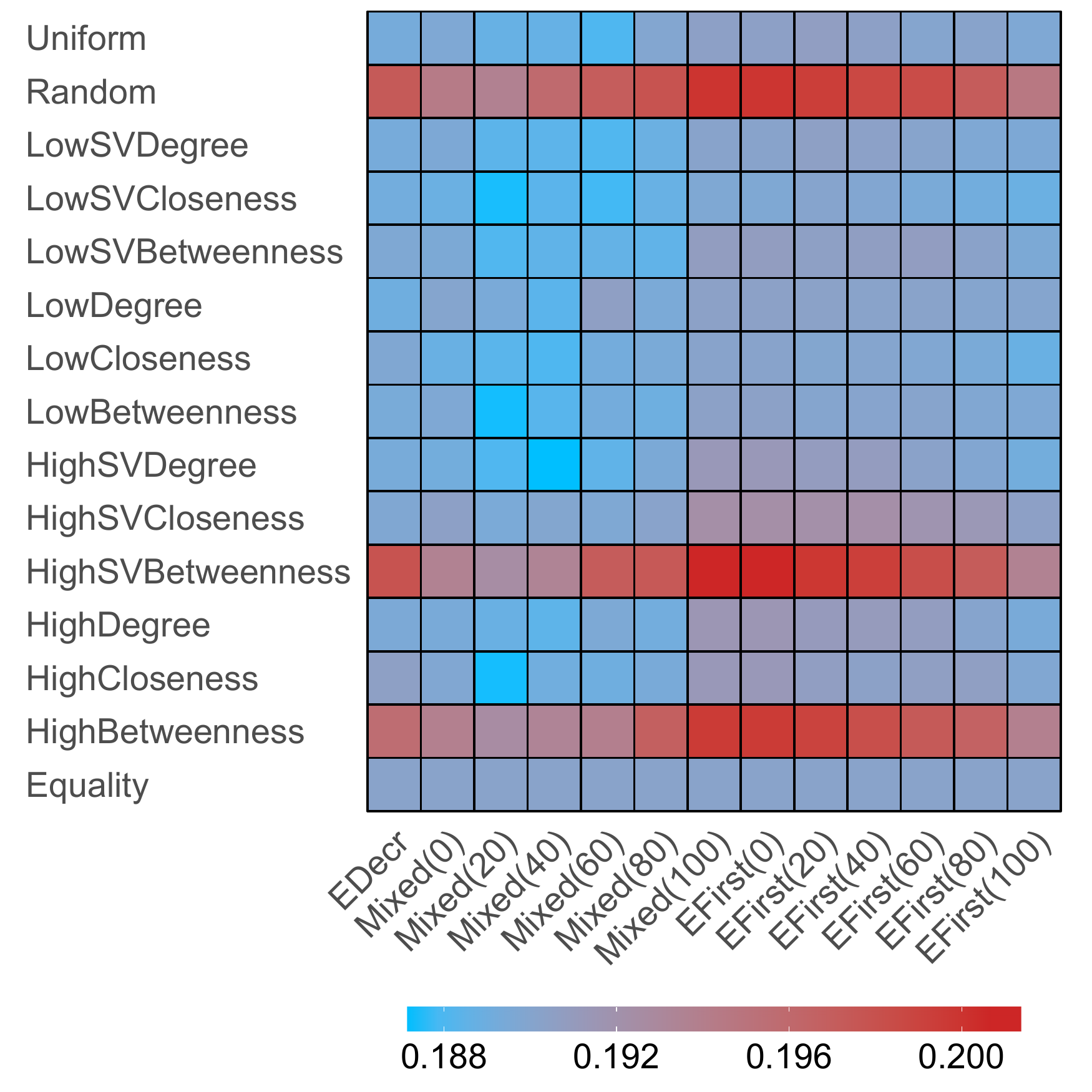} &
\includegraphics[width=\linewidth]{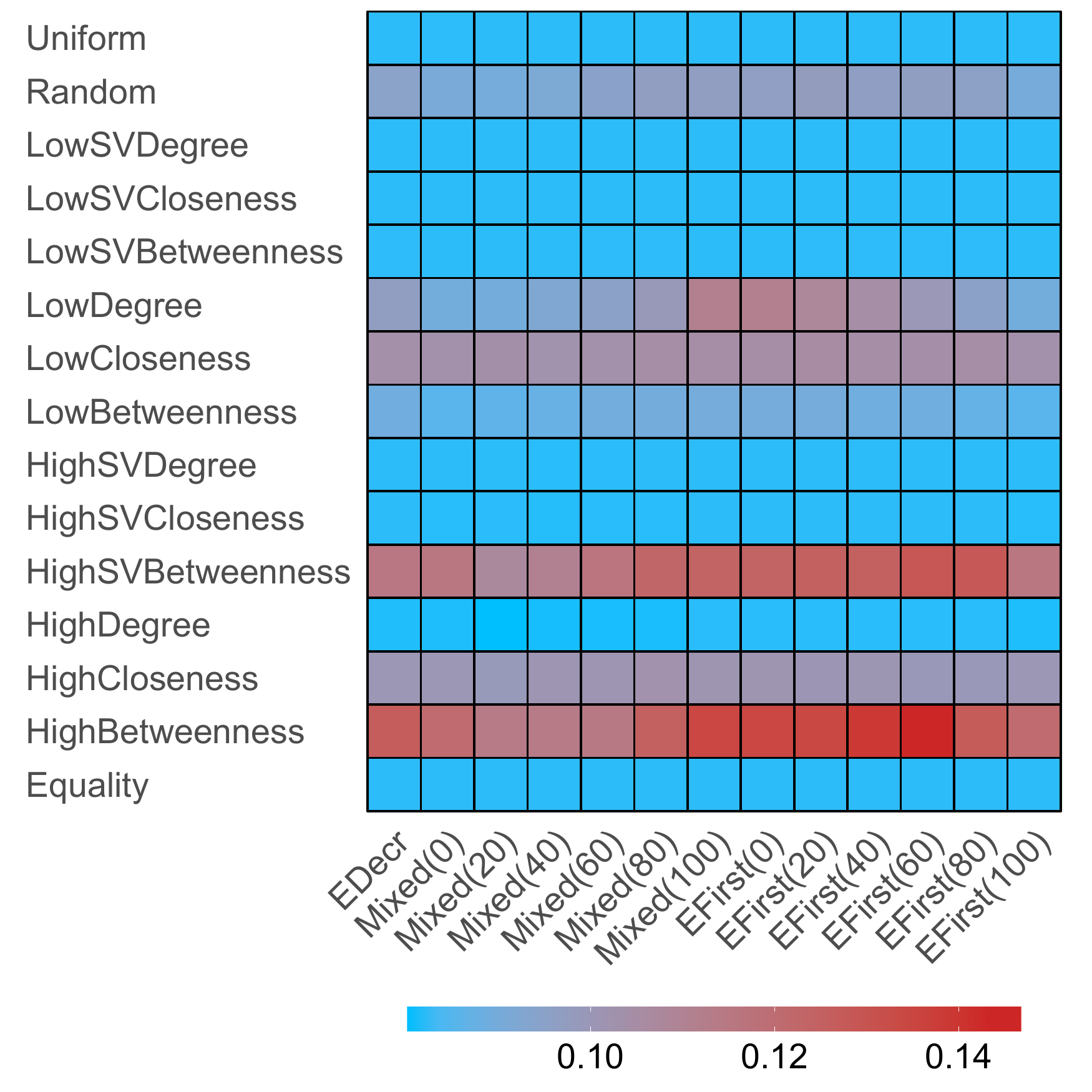} \\
\end{tabular}
\caption{
Comparison of effectiveness of defender's strategies for different attacker's strategies on networks with $1000$ nodes.
Color of each cell represents the expected percentage of nodes successfully activated by the attacker.
Results are taken as an average over $100$ simulations, with a new network generated for each simulation using one of the models.
}
\label{fig:heat-1000}
\end{figure}

\begin{figure}[t]
\centering
\setlength\tabcolsep{0pt}
\begin{tabular}{m{.5\textwidth}m{.5\textwidth}}
\multicolumn{1}{c}{Networks with $80$ nodes (against optimal)} &
\multicolumn{1}{c}{Networks with $80$ nodes (against best heuristic)} \\
\includegraphics[width=\linewidth]{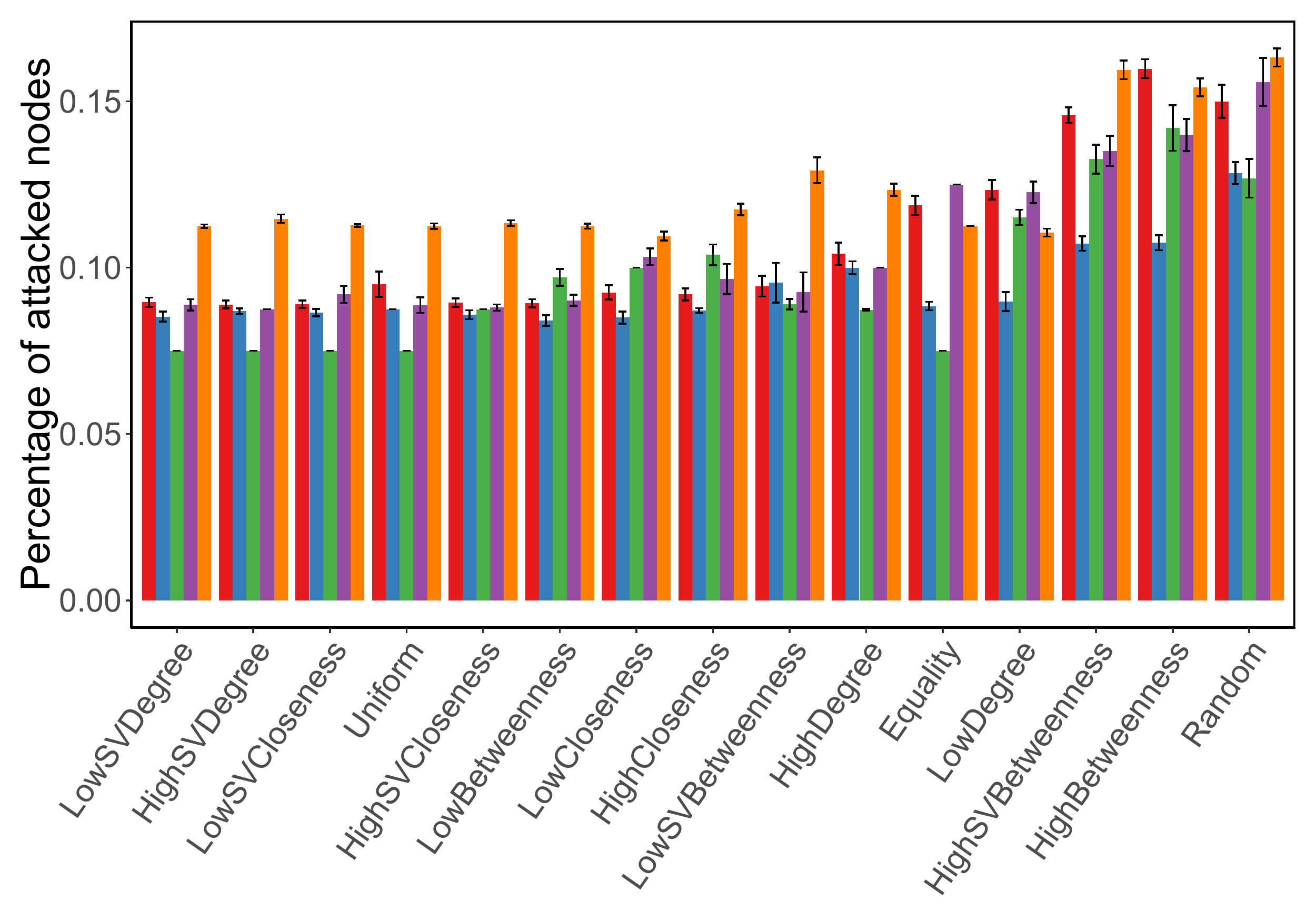} &
\includegraphics[width=\linewidth]{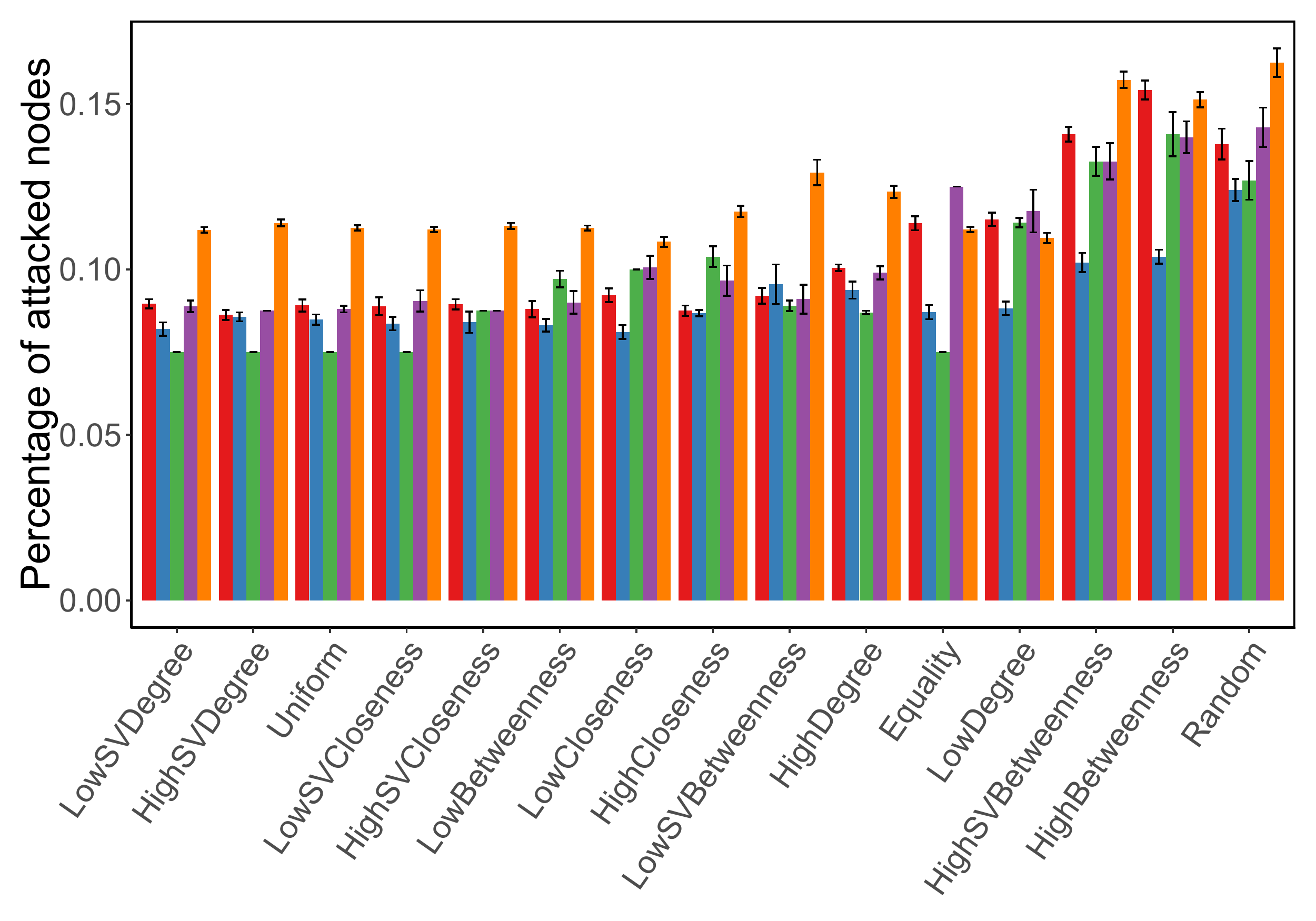} \\
\end{tabular}
\begin{tabular}{m{.5\textwidth}}
\multicolumn{1}{c}{Networks with $1000$ nodes (against best heuristic)} \\
\includegraphics[width=\linewidth]{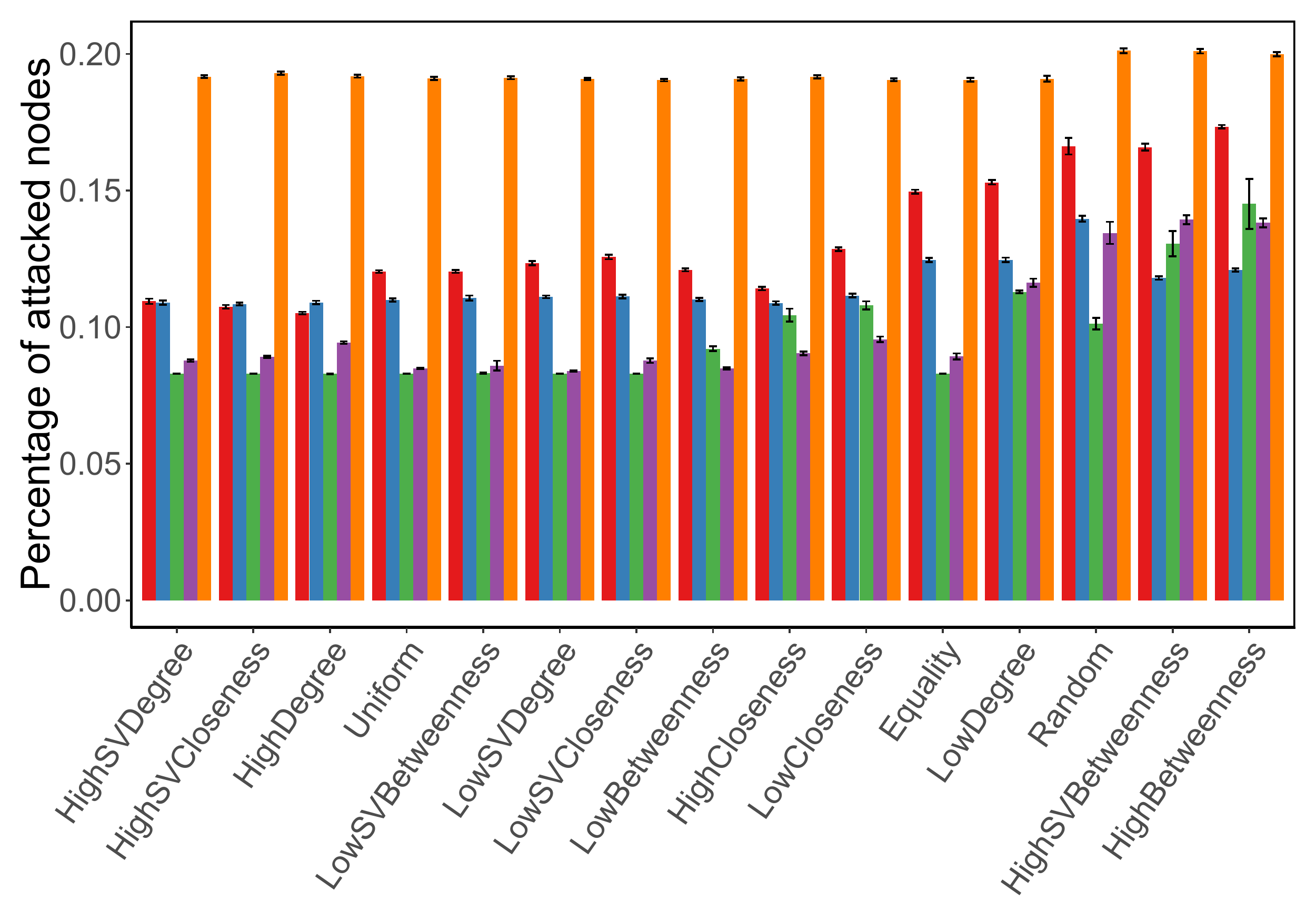} \\
\end{tabular}
\begin{tabular}{m{\textwidth}}
\multicolumn{1}{c}{\includegraphics[width=.7\linewidth]{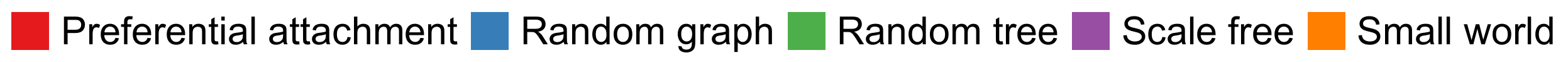}} \\
\end{tabular}
\caption{
Comparison of effectiveness of defender's strategies, under assumption that the attacker uses optimal strategy (for networks with $80$ nodes) or the best heuristic considered by us (for networks with $1000$ nodes).
Results are taken as an average over $100$ simulations, with a new network generated for each simulation using one of the models.
Defender's strategies are sorted from most to least effective on average.
Error bars represent $95\%$ confidence intervals.
}
\label{fig:best-bars}
\end{figure}

\noindent \paragraph{Small networks} We will first describe results for networks with $80$ nodes.
We generate networks of this size to be able to compare the performance of heuristic solutions against the optimal strategy of the attacker.

Interestingly, the Uniform strategy is a relatively effective way of defending the network.
At the same time, the Equality strategy usually gets considerably worse results than the Uniform strategy.
This seems to be the case because, while the Equality strategy aims to compensate differences in the time of activation between nodes with different degrees, the Uniform strategy keeps hubs harder to activate than other nodes, discouraging the attacker from targeting them.

Out of the betweenness centrality-based defender strategies, the Low Betweenness strategy is always more effective than its High Betweenness counterpart.
The efficiency of the Low Betweenness strategy might seem surprising at first.
However, low betweenness centrality values are characteristic of not only peripheral nodes in the network, but also of the members of dense network structures (\eg, in a clique the betweenness centrality of all nodes is zero) that are typically high return targets for the attacker, as the activation of each of them affects many peers.
In fact, the Low Betweenness strategy in many cases is one of the most effective ways of defending the network.

Out of the defender strategies based on closeness centrality, their relative effectiveness depends on the network model under consideration.
The High Closeness strategy is more effective in majority of the simulations.
Nevertheless, they usually achieve similar performance and choosing the right closeness-based strategy is a highly effective way of defending the network.
Interestingly, these results can be connected to the role of nodes with low betweenness centrality and high closeness centrality as the ``pulse-takers'' of their organizations~\cite{henderson2013department}.

The High Degree strategy, focusing its defense efforts on hubs, is also relatively effective, as it discourages the attacker from activating high-degree nodes giving access to large parts of the network.
Low Degree and Random strategies are less effective in most cases.

Out of the strategies based on the Shapley value, the High Shapley Value Betweenness is the least effective in all cases, similarly to its basic counterpart, \ie, the High Betweenness strategy.
The Low Shapley Value Betweenness achieves better results, although on average it is still inferior to the rest of the strategies based on the Shapley value.
Both High and Low Shapley Value Degree, as well as both High and Low Shapley Value Closeness are usually among the best defender's heuristic, with the exact order depending on the network structure and the strategy of the attacker.
Interestingly, they are consistently more efficient than the basic centrality measures, indicating the importance of the synergy effect (expressed by the Shapley value) in planning the defense of the network.

From the attacker's point of view, if she is unable to compute the optimal solution, in most cases using the Greedy strategy with a high probability is the best choice and close to the optimal one.
In most cases pure Majority and pure Greedy strategies give better results than adding a small percentage of the other strategy into the mix.
The Epsilon-First strategy generally achieves better results than simple mix, while the Epsilon-Decreasing strategy is in most cases worse than both pure Majority and pure Greedy strategies.

As for the comparison between different network generation models in terms of attack potential, the small-world networks are on average the easiest to attack, as short expected distance between nodes allows to quickly reach less protected parts of the network.
The networks with exponential degree distributions (\ie, preferential attachment and scale-free networks) are on average easier to attack than random graphs and trees.
This may be due to the fact that hubs in networks with exponential degree distributions are more connected, and activating them gives easy access to a large part of the network.
The random trees are the most difficult target for the attacker.
They offer very limited choice when guiding the attack, as any two nodes are connected with but a single path.

We will focus on the choice of the best defense heuristic if the attacker utilizes her optimal strategy, which is the worst possible situation for the defender.
Results for this scenario are presented in Figure~\ref{fig:best-bars}.
The best defender's choice depends on the generation model of the network.
In the case of the random graphs, the Low Closeness and Low Betweenness heuristics give the best results (even though Low Closeness is on average worse than its high closeness counterpart when considering all different attacker's strategies), with a number of Shapley value-based heuristics closely following.
In case of the scale-free networks, the heuristics based on high scores of Shapley value degree and Shapley value closeness centrality measures provide the best results for the defender.
When the preferential attachment networks are considered, the High Shapley Value Degree and the Low Shapley Value Closeness heuristics offer the best protection.
For the small-world networks Low Closeness and Low Degree heuristics are the most effective, closely followed by the Low Shapley Value Degree strategy.
Finally, the random trees can be defended most efficiently by the High Shapley-Value Based Degree strategy.
As it can be seen, in most cases the heuristics utilizing Shapley value-based degree and closeness centralities are valid choices, albeit not always the best ones.

It might be interesting to compare the effectiveness ranking for small networks when the attacker uses the optimal strategy (top-left plot in Figure~\ref{fig:best-bars}), and when she uses the best heuristic (top-right plot in Figure~\ref{fig:best-bars}).
As it can be seen, the ranking of the defender strategies remains almost identical, with the only difference being the Uniform strategy getting slight edge over Low Shapley Value Closeness strategy.
Other trends also remain largely consistent between the two settings.


\noindent \paragraph{Large networks} We now discuss results for larger networks, consisting of $1000$ nodes.
Because of time and memory complexities we are unable to compute the optimal solution for networks of this size.
However, we present results of the heuristic algorithms to observe how trends in them change with the size of the network.

Most of our observations made for the smaller networks remain valid.
One notable difference in terms of defender strategies is an increase of the effectiveness of High Degree strategy.
In many cases it now outperforms most of the other strategies and the difference is especially striking in the case of the preferential attachment networks.
In this type of networks the expected degree of hubs greatly increases with size of the network, hence the attacker's advantage from activating a hub and gaining access to a large part of the network is also much higher.
Because of this, focusing defender's efforts on protecting high degree nodes significantly improves effectiveness of the defense.

Another notable difference is how the effectiveness of High Betweenness strategy greatly depends on the strategy of the attacker.
While it has considerable variance in effectiveness even for smaller networks, for larger networks it achieves both the best and the worst results throughout all defender strategies.
It is particularly efficient to use the High Betweenness strategy against attacker strategies utilizing Majority with high probability, while we observe much worse performance against Greedy-intensive attacker strategies.
A possible explanation is that in larger networks, especially these with scale-free properties, the distribution of betweenness centrality values among nodes is more heterogenous.
Majority strategy, ignoring actual time of activation of the nodes, is unable to avoid attacking high betweenness, strongly defended nodes, while Greedy strategy is able to do so.

Another difference is that the Epsilon-Decreasing strategy, usually achieving poor performance for small networks, in large networks is doing better in comparison to Mixed and Epsilon-First strategies.
In many cases it has better results than pure Majority or pure Greedy strategies, depending on the defender's strategy.
Similarly, for Mixed and Epsilon-First strategies a particular parametrization often achieves better performance than pure Greedy and Majority strategies.

While the small-world structure is the most difficult to defend for both small and large networks, the difference between it and the other models of generation is especially striking for the networks with $1000$ nodes.
In fact, the minimal number of successfully attacked nodes for small-world networks is greater than the maximal number of successfully attacked nodes for all other types of network structure.

Since we are unable to compute the optimal strategy of the attacker for large networks, we will now analyze the choice of the defender's heuristic under assumption that the attacker will use the best response out of her strategies considered in the paper (in most cases a particular parameterization of the Epsilon-First strategy).
The results for this scenario are presented in Figure~\ref{fig:best-bars}.
In the case of the random graphs, a few heuristics have a very similar effectiveness, namely High Shapley Value Closeness, High Closeness, High Shapley Value Degree and High Degree (the difference between the best and the worst of them is smaller than a single active node).
For the scale-free network, the Low Shapley Value Degree heuristic offers the best performance, with the Low Betweenness heuristic being the second.
In the case of both the preferential attachment networks and the random trees, the best choice of the defender is to use the High Degree heuristic.
Finally, for the small-world networks the best performance is provided by the Low Shapley Value Closeness heuristic.

\begin{figure}[tbh]
\centering
\setlength\tabcolsep{0pt}
\begin{tabular}{m{.25\textwidth}m{.25\textwidth}m{.25\textwidth}m{.25\textwidth}}
\multicolumn{1}{c}{Karate Club} &
\multicolumn{1}{c}{Facebook (small)} &
\multicolumn{1}{c}{Facebook (medium)} &
\multicolumn{1}{c}{Facebook (large)} \\
\includegraphics[width=\linewidth]{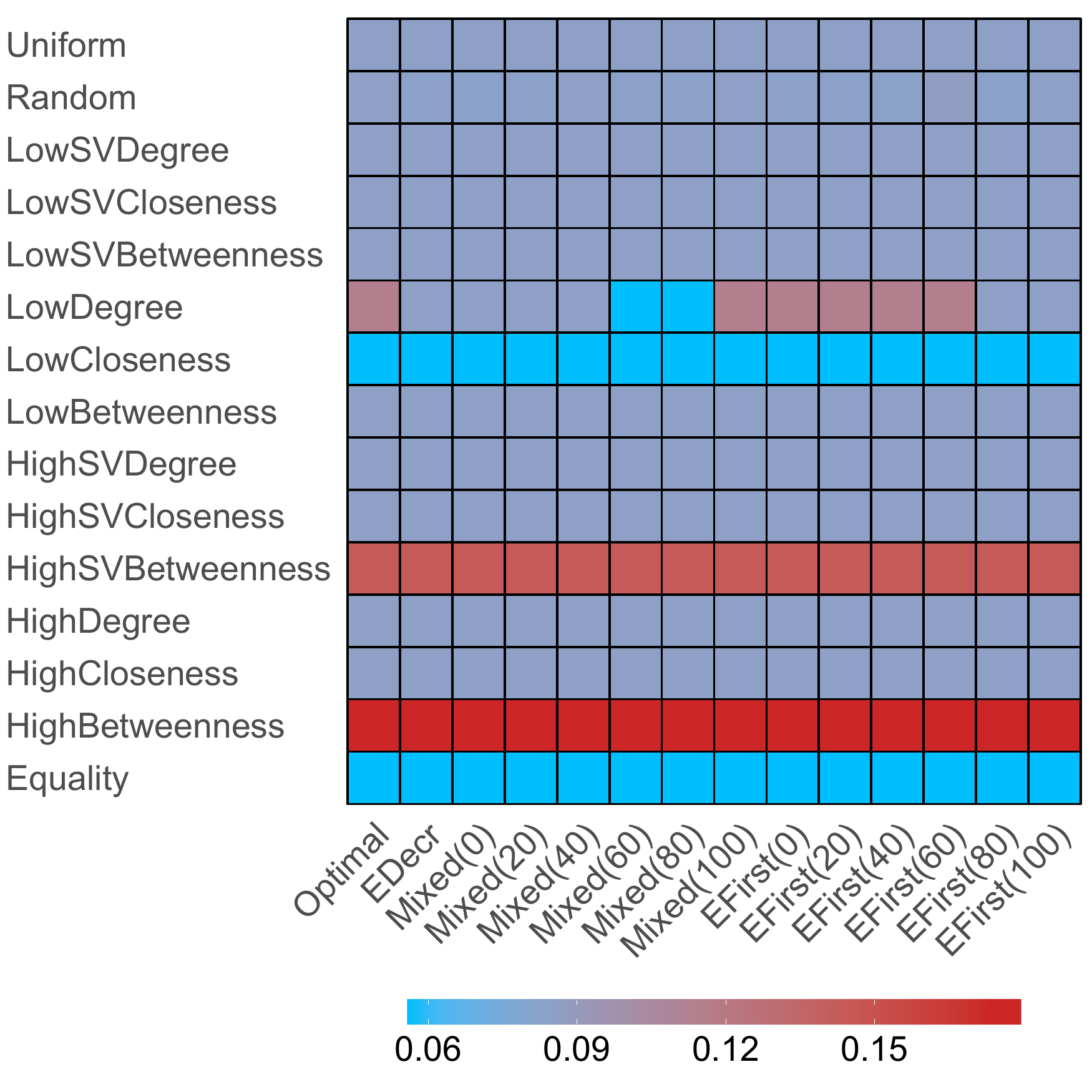} &
\includegraphics[width=\linewidth]{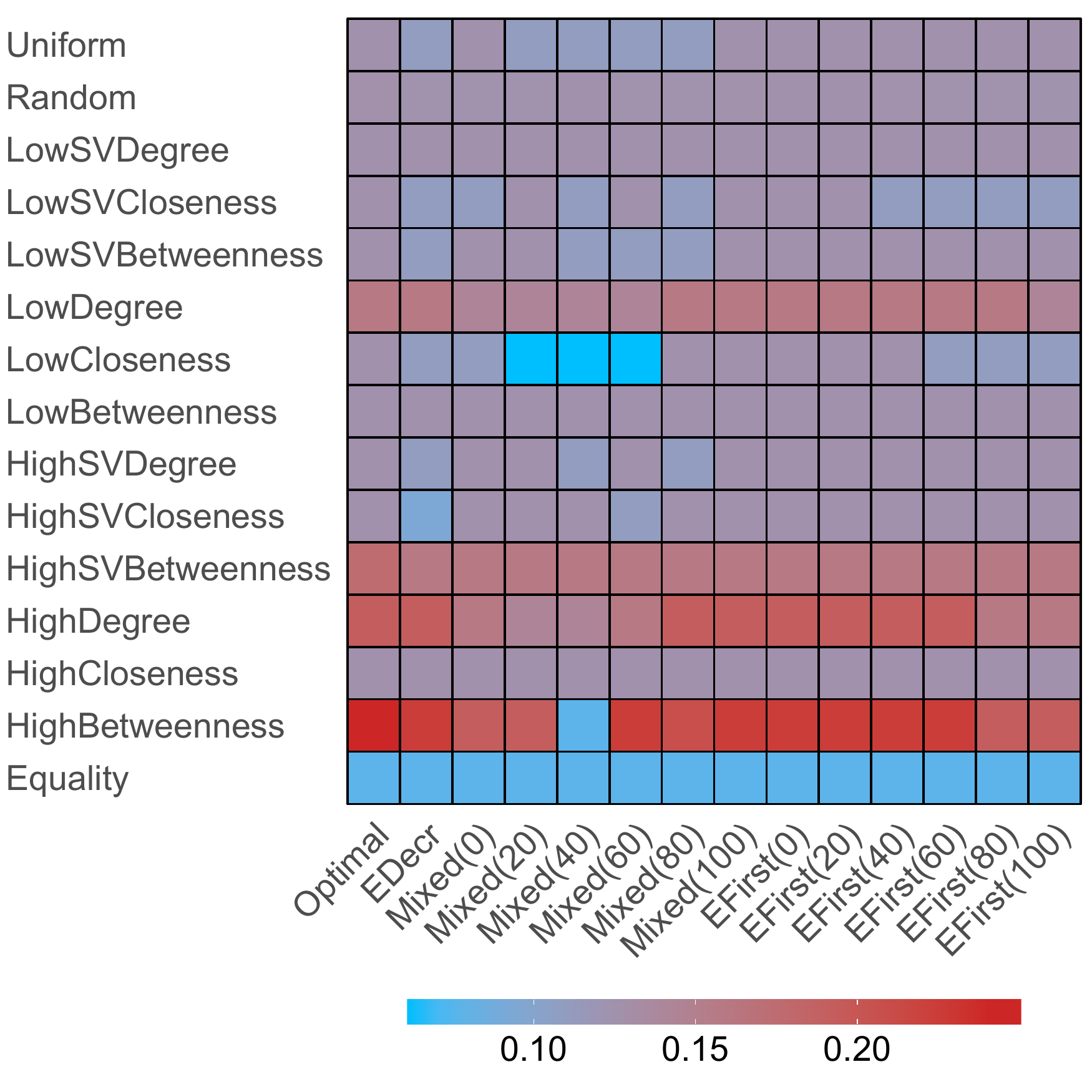} &
\includegraphics[width=\linewidth]{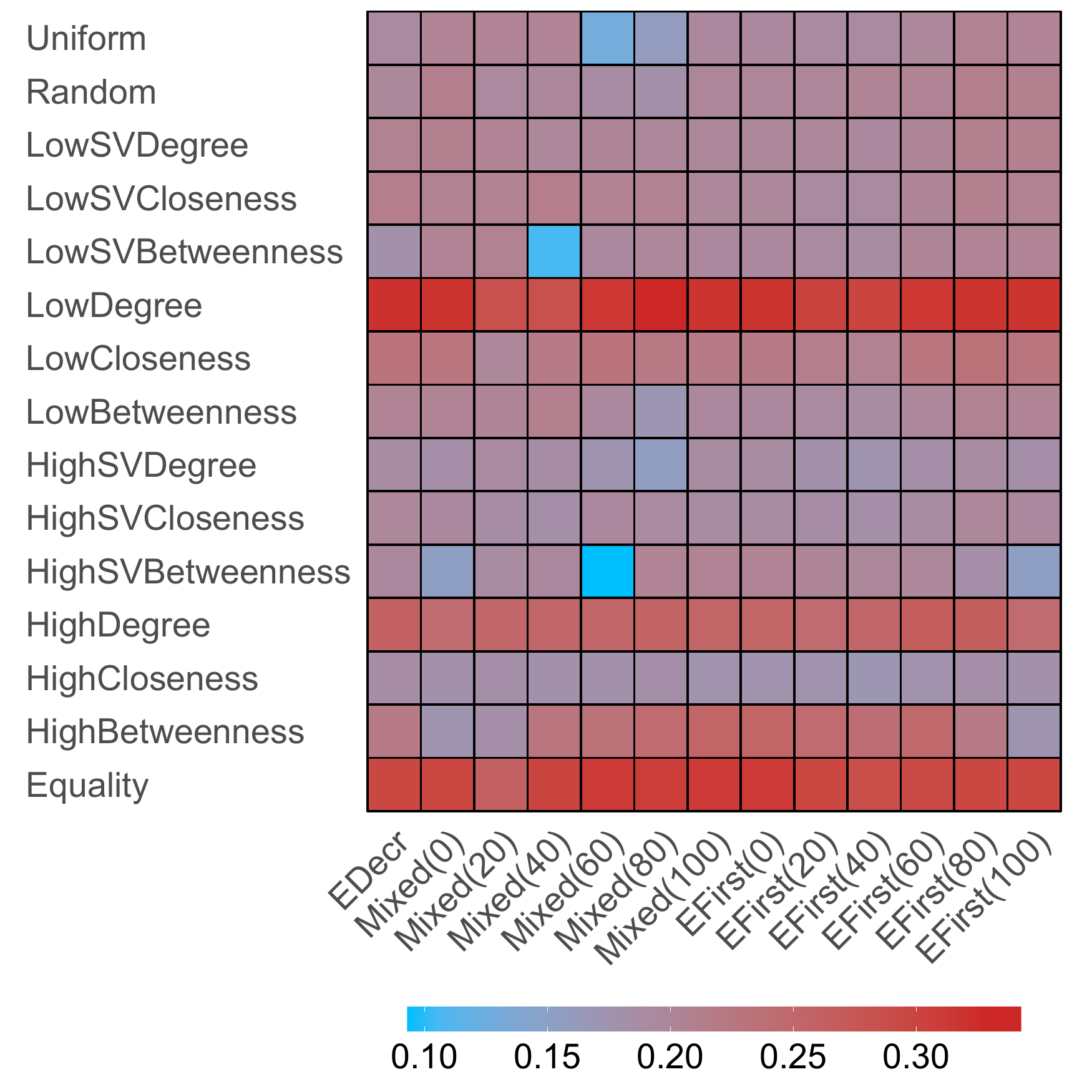} &
\includegraphics[width=\linewidth]{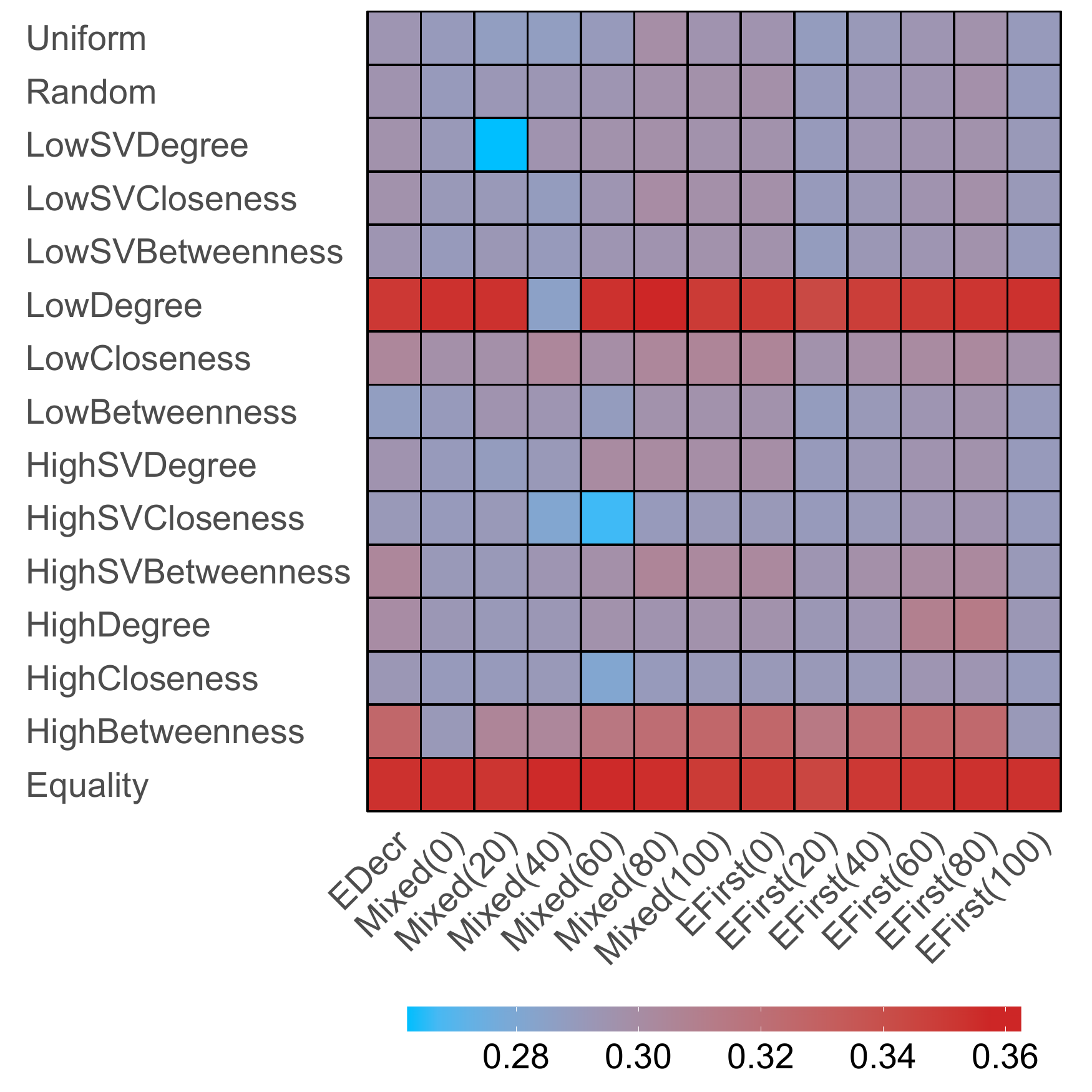} \\
\end{tabular}
\caption{
Comparison of effectiveness of defender's strategies for different attacker's strategies.
Color of each cell represents the expected percentage of nodes successfully activated by the attacker.
}
\label{fig:heat-reallife}
\end{figure}

\noindent \paragraph{Real-life networks}
Results of the simulations on real-life networks are presented in Figures~\ref{fig:heat-reallife} and~\ref{fig:best-bars-realife-soc}. 
Notice that we computed the Optimal attack strategy only for networks with at most $80$ nodes.

In general, the results for real-life networks show similar trends to those for the randomly generated datasets.
In a number of situations we observe a sharp increase or a sharp decrease in the effectiveness of a given attacker strategy for a particular parametrisations (see, e.g., the Mixed(20) attacker strategy against the Low Shapley Value Degree defense in the large Facebook fragment), which is caused by the fact that, unlike for the randomly generated network, the results are not averaged over multiple runs, but rather computed for a single structure.

The comparison of results for different sizes of the networks indicates that, in general, smaller networks are easier to defend.
For almost all strategies of the defender their relative effectiveness decreases with the number of nodes (\ie, the percentage of successfully attacked nodes grows with the number of nodes).

As for the effectiveness of the defense heuristics, Random, Low Degree, High Betweenness, and High Shapley Value Betweenness are among the worst methods of defending the network, just as for the random networks.

In a situation where the total number of nodes that can be attacked is small, maximizing the cost of attacking the first node (which is the main principle behind the Equality strategy) shows its strength.

\begin{figure}[tbh]
\centering
\setlength\tabcolsep{0pt}
\begin{tabular}{m{.5\textwidth}m{.5\textwidth}}
\multicolumn{1}{c}{Against the optimal strategy} &
\multicolumn{1}{c}{Against the best heuristic} \\
\includegraphics[width=\linewidth]{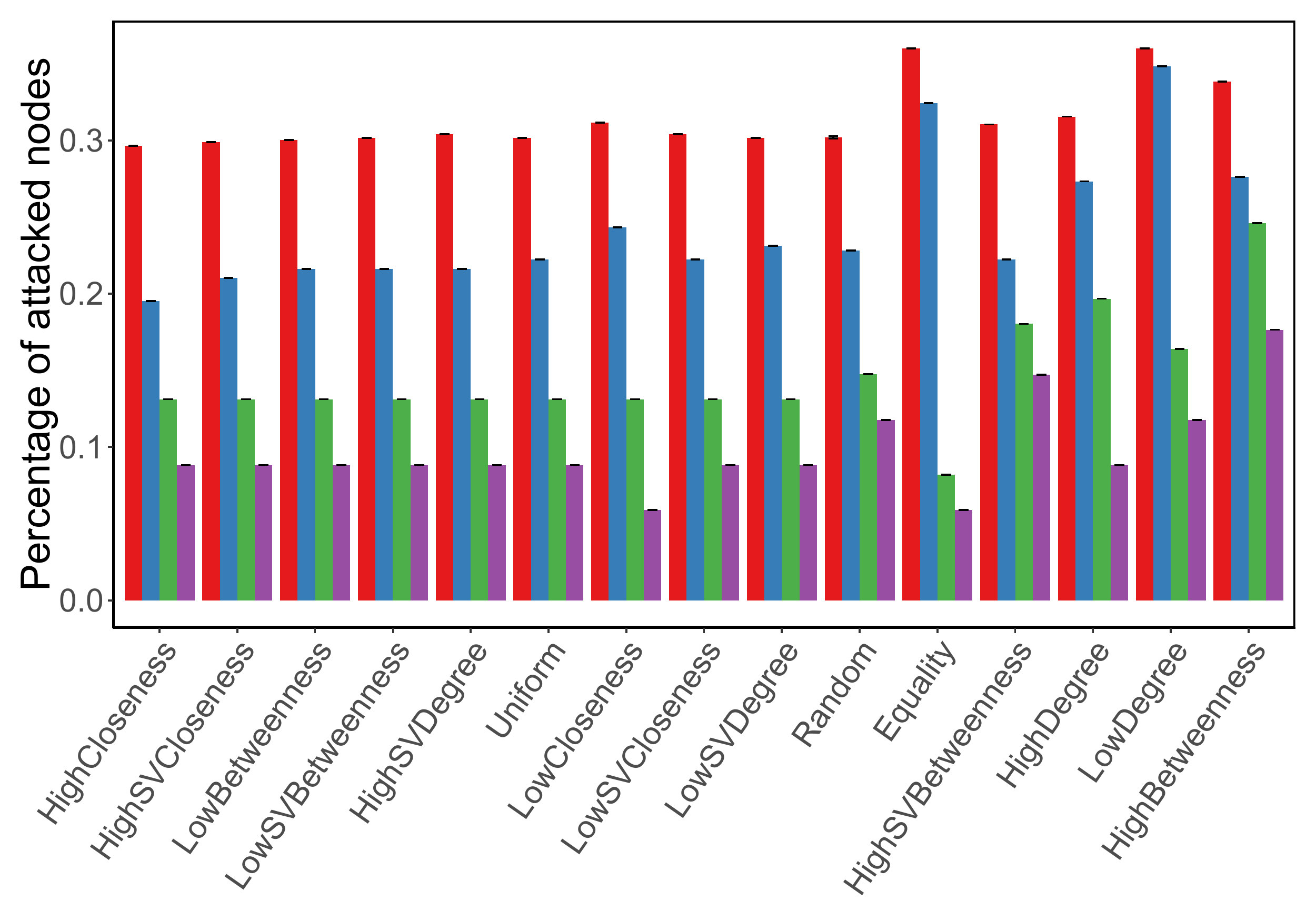} &
\includegraphics[width=\linewidth]{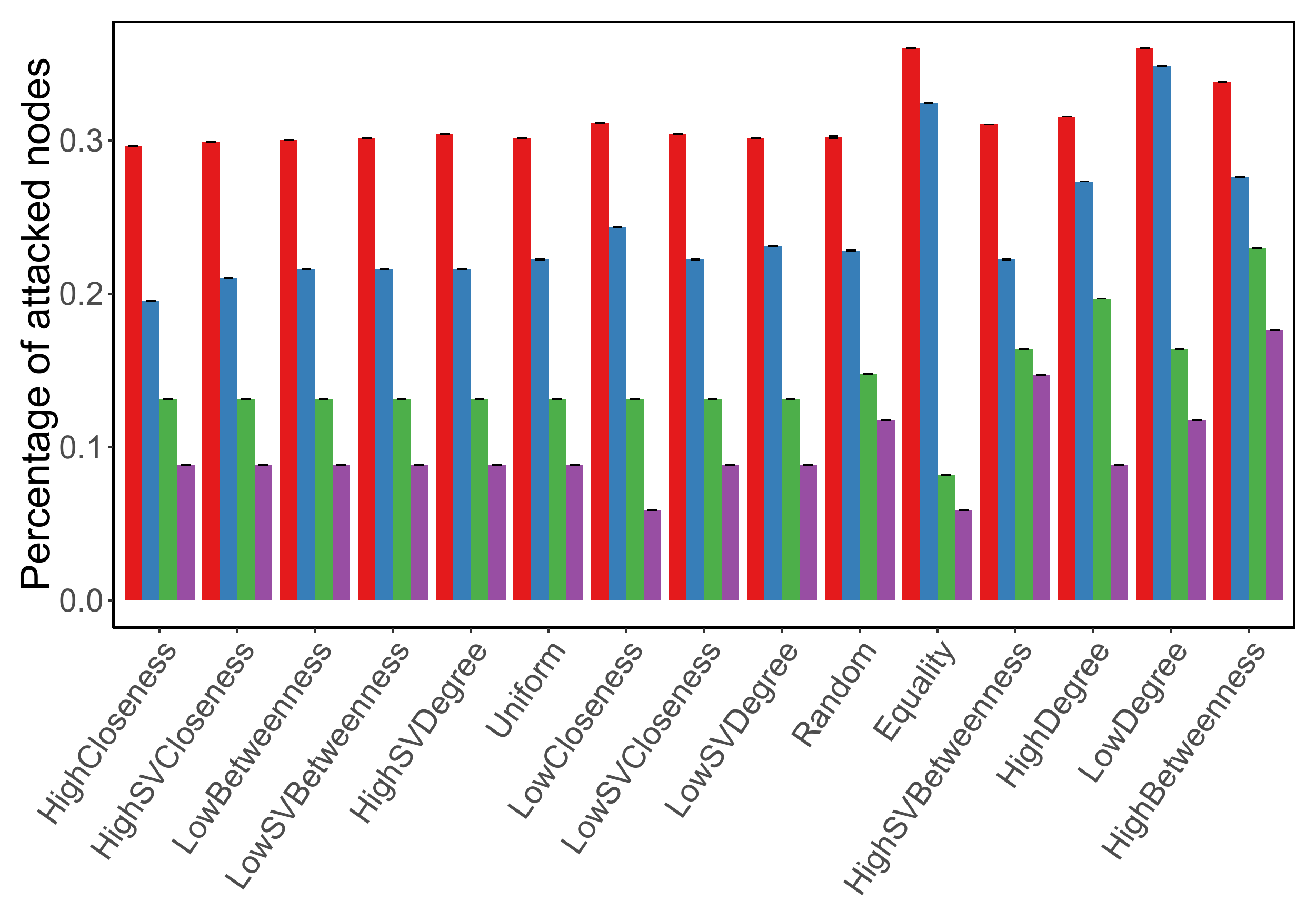} \\
\multicolumn{2}{c}{\includegraphics[width=.6\linewidth]{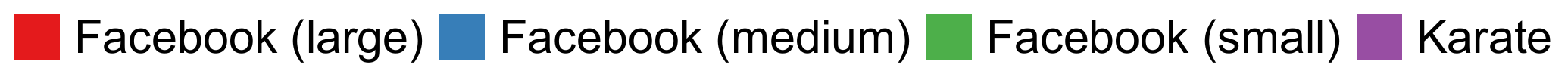}} \\
\end{tabular}
\caption{
Comparison of effectiveness of defender's strategies for real-life social networks, under assumption that the attacker uses optimal strategy or the best heuristic considered by us.
Defender's strategies are sorted from most to least effective on average.
}
\label{fig:best-bars-realife-soc}
\end{figure}


\subsection{Alternative Experimental Settings}

We also conducted a number of other experimental studies, the results of which are presented in the appendices.
Appendix~\ref{app:tiny-networks} presents the results of experiments with the mixed-integer linear programming strategy of the defender described in Section~\ref{app:milp}.
In Appendix~\ref{app:multiple-seeds} we consider a setting in which the attacker can choose multiple seeds at the beginning of the process.
Appendix~\ref{app:repeated-setting} presents a series of security diffusion games in which the defender updates her distribution of security resources based on the previous choices of the attacker.
Finally, in Appendix~\ref{app:alternative-models} we consider alternative forms of the function used to compute the time of activation of the network nodes.

\section{Discussion and Conclusions}

\noindent
We considered the problem of defending the network from a sequential attack that is strategically planned and controlled by the attacker, rather than one that spreads due to a stochastic process.
We proved that finding an optimal way to spread the attack is an NP-complete task in a general case, which suggests that the attacker may not always utilize the best way of diffusion.
We found optimal strategies for simple network structures and formulated the problem of finding an optimal way of defending the network as a mixed-integer linear programming.
For large and complex networks, we experimentally analyzed ways of defending the network from various types of attacks.
We found that the highest utility of the defender is achieved by utilizing one of Shapley value-based network centrality measures.
In most cases the attacker can achieve high utility by using the greedy strategy with a high probability.

One of the fundamental assumptions underlying our study is that of centralised defense, i.e., of a single authority who decides on the security level of each node. An alternative assumption that has been often considered in the network literature is that the level of defense can be chosen individually by each node in a network \cite{goyal2015interaction, kunreuther2003interdependent}.
In most of these settings, each target is autonomous and can decide whether or how to defend itself~\cite{goyal2015interaction}.
Therefore, this body of literature typically uses game-theoretic methods to identify equilibria of the game between attackers and defenders when it is possible to efficiently identify best strategies of all players. Such decentralized defense in interdependent security games fits scenarios, where individuals deliberate about the incentive to invest in protection such as vaccination, airline security or fire protection and act upon that \cite{kunreuther2003interdependent}. Our study, in turn, focuses the  scenarios with a central defender who is in charge of protecting all connected targets. This one defender can therefore build a resilient network or centrally assign defense resources across nodes~\cite{goyal2014attack, dziubinski2017you}. Arguably, centralized defenders can perform better in finding global defense solutions and hence can be best adopted in high-stake problems whereas decentralized defense gives nodes the option of deciding whether to protect themselves or not and therefore unincentivized nodes might put a whole network at a risk.  

In a broader context, our work can be applicable and extended to other types of sequential attacks on a network~\cite{nguyen2009stochastic}, that are used to model, for instance, competitive adoption of a product~\cite{broecheler2010scalable} and protecting computer networks~\cite{durkota2014computing}.

As for the ideas for future work, one possible venue is more in-depth analysis of the exploration-exploitation techniques.
In this work we considered only two simple techniques.
Using more intricate learning techniques~\cite{hu1998multiagent} might yield better results in adjusting a strategy of the attacker.
Another interesting problem would be to consider a setting where, instead of playing just one game, players can face each other off multiple times on the same network and they modify their strategies as the process goes on.
Finding an equilibrium of this repeated game may lead to a better understanding of the dynamics of defending a network against a strategic continuous attack.

\section*{Acknowledgments}

\noindent This work was funded by the Cooperative Agreement between the Masdar Institute of Science and Technology (Masdar Institute), Abu Dhabi, UAE and the Massachusetts Institute of Technology (MIT), Cambridge, MA, USA-Reference 02/MI/MIT/CP/11/07633/GEN/G/00.
Also, this article was supported by the Polish National Science Centre grants no. 2015/17/N/ST6/03686 and 2016/23/B/ST6/03599.

\clearpage

\bibliographystyle{abbrv}
\bibliography{bibliography}

\appendix

\clearpage

\section{Optimal Strategies for Cliques}
\label{app:cliques}

\begin{theorem}
\label{thrm:clique}
Let network $G=(V,E)$ be a clique, \ie, $\forall_{v,w\in V}(v,w)\in E$. Given a particular distribution of security resources $\sr$, an optimal attack sequence is $\seq^*$, where nodes are ordered non-decreasingly according to assigned security resources, \ie, $\forall_{i<j}\sr(\seq^*_i) \leq \sr(\seq^*_j)$.

An optimal defense strategy against an optimal attack strategy is to spread security resources uniformly, \ie, $\forall_{v\in V}\sr^*(v)=\frac{\SR}{n}$.
\end{theorem}

\begin{proof}
Since the network is a clique, the activation time of $i$-th node in sequence $\seq$ is $\frac{n-1+\sr(\seq_i)}{\max(1,i-1)}$. This is because $\seq_i$ has $n-1$ neighbors, $i-1$ of which have been already activated).

First we prove our claim about the strategy of the attacker.

We prove it by contradiction.
Assume to the contrary, that in an optimal attack sequence $\seq$ we have $\sr(\seq_i) > \sr(\seq_j)$ for $i < j$ such that $i > 1$ and $j > 2$ (notice that we can always swap first and second element of the sequence as the sum of their activation times remains the same).
However, such a sequence can be improved by swapping elements $i$ and $j$.
Let $\et_{ij}$ be the time of activation of $\seq_i$ and $\seq_j$ before such a swap, and $\et_{ji}$ after the swap (note that the activation times of all other nodes remain the same after the swap). We have that (notice that $\max(1,i-1)=i-1$ and $\max(1,j-1)=j-1$ for $i > 1$ and $j > 2$):
$$
\begin{aligned}
\et_{ji}-\et_{ij} =& \left(\frac{n-1+\sr(\seq_j)}{i-1}+\frac{n-1+\sr(\seq_i)}{j-1}\right)- \left(\frac{n-1+\sr(\seq_i)}{i-1}+\frac{n-1+\sr(\seq_j)}{j-1}\right) \\
=& \frac{\sr(\seq_i)-\sr(\seq_j)}{j-1} - \frac{\sr(\seq_i)-\sr(\seq_j)}{i-1} < 0.
\end{aligned}
$$

Therefore, for $i < j$ such that $i > 1$ and $j > 2$ we always have $\sr(\seq_i) \leq \sr(\seq_j)$.

We will now prove our claim about the strategy of the defender. The proof is by contradiction.
Assume to the contrary, that in an optimal distribution of security resources $\sr$ we have $i>1$ such that $\sr(\seq^*_i) \neq \sr(\seq^*_{i+1})$, where $\seq^*$ is an optimal attack sequence (again, notice that we can always swap the first and the second elements in the attack sequence as the sum of their activation times remains the same).
Since $\seq^*$ is an optimal attack sequence we have that $\sr(\seq^*_i) < \sr(\seq^*_{i+1})$.
However, such $\sr$ can be improved by setting $\sr(\seq^*_i)$ and $\sr(\seq^*_{i+1})$ to their average.
Let $\et_{1}$ be time of activation of $\seq^*_i$ and $\seq^*_{i+1}$ before the change, and $\et_{2}$ after the change (the activation times of all other nodes remain the same after the change).
We have:
$$
\begin{aligned}
\et_{2}-\et_{1} =& \left(\frac{n-1+\frac{\sr(\seq^*_i)+\sr(\seq^*_{i+1})}{2}}{i-1}+\frac{n-1+\frac{\sr(\seq^*_i)+\sr(\seq^*_{i+1})}{2}}{i}\right)- \left(\frac{n-1+\sr(\seq^*_i)}{i-1}+\frac{n-1+\sr(\seq^*_{i+1})}{i}\right) \\
=& \frac{\sr(\seq^*_{i+1})-\sr(\seq^*_i)}{2(i-1)}-\frac{\sr(\seq^*_{i+1})-\sr(\seq^*_i)}{2i} = \frac{\sr(\seq^*_{i+1}) - \sr(\seq^*_i)}{2i(i+1)} > 0.
\end{aligned}
$$

Therefore, in an optimal distribution of security resources we have $\sr(\seq^*_i) = \sr(\seq^*_{i+1})$ for $i>1$.
\end{proof}

\section{The Optimal Attack for Trees}
\label{app:trees}

\begin{theorem}
\label{thrm:tree-attacker}
Let network $G$ be a tree. The optimal strategy of the attacker can be computed in polynomial time.
\end{theorem}

\begin{proof}

In Algorithm~\ref{alg:tree-attack}, we present the pseudocode of the algorithm for computing the optimal attacker's strategy on a tree starting from an \textit{a priori} given node. It performs the exhaustive search of all possible strategies of the attacker by traversing the tree in a bottom-up fashion.
In the pseudocode, $V_{botup}$ denotes the sequence of nodes in $V$ in a bottom-up order (when $\vs$ is the root), $c(v)$ denotes the sequence of the children of $v$, $c_i(v)$ denotes the $i$-th node in this sequence, and $t(v)$ denotes the size of the subtree with $v$ as the root.
We denote the concatenation of sequences using the $\oplus$ symbol.

\begin{algorithm}[tbh!]
\setstretch{1.1}
\LinesNumbered
\DontPrintSemicolon
\SetAlgoNoEnd
\SetAlgoNoLine
\KwIn{Tree $G=(V,E)$, distribution of security resources $\sr$, time limit $T$, starting node $\vs$.}
\KwOut{Optimal attack sequence starting with node $\vs$}
\For{$v \in V$}{
	\For{$k \in \langle 1, \ldots, n\rangle$}{
		$x[v,k] \gets \infty$\;
	}
}
\For(\tcp*[f]{loop over all nodes in bottom-up order})
{$v \in V_{botup}$}{
	$\tau[v,0] \gets 0$\;
	$\gamma[v,0] \gets \emptyset$\;
	$\tau[v,1] \gets d(v)+\sr(v)$\;
	$\gamma[v,1] \gets \langle v \rangle$\;
	\If(\tcp*[f]{if $v$ is not a leaf aggregate results from children})
	{$d(v) > 1$}{
		\For(\tcp*[f]{choose targets only from the subtree of the first child})
		{$k \in \langle 1, \ldots, t(c_1(v))\rangle$}{
			$y[1,k] \gets d(v)+\sr(v) + \tau[c_1(v),k-1]$\;
			$q[1,k] \gets \langle v \rangle \oplus \gamma[c_1(v),k-1]$\;
		}
		$l \gets t(c_1(v)) + 1$\;
		\For(\tcp*[f]{choose targets from the subtrees of the first $m$ children})
		{$m \in \langle 2, \ldots, |c(v)|\rangle$}{
			$l \gets l + t(c_m(v))$\;
			\For(\tcp*[f]{loop over the number of attacked nodes})
			{$k \in \langle 1, \ldots, l\rangle$}{
				$y[m,k] \gets \infty$\;
				\For(\tcp*[f]{loop over the number of attacked nodes we select from the subtrees of the first $m-1$ children})
				{$i \in \langle 1, \ldots, k\rangle$}{
					$\tilde{y} \gets y[m-1,i]+\tau[c_m(v),k-i]$\;
					\If{$\tilde{y} < y[m,k]$}{
						$y[m,k] \gets \tilde{y}$\;
						$q[m,k] \gets q[m-1,i] \oplus \gamma[c_m(v),k-i]$\;
					}
				}
			}
		}
		\For(\tcp*[f]{transfer results to tables $\tau$ and $\gamma$})
		{$k \in \langle 2, \ldots, t(v)\rangle$}{
			$\tau[v,k] \gets y[|c(v)|,k]$\;
			$\gamma[v,k] \gets q[|c(v)|,k]$\;
		}
	}
}
\Return $\max_{k:\tau[\vs,k] < T}\gamma[\vs,k]$
\caption{Finding the optimal attack sequence for a tree.}
\label{alg:tree-attack}
\end{algorithm}

One can notice that, in the case of a tree, every node other than the seed node has exactly one active neighbor at the moment of activation.
Hence, its time of activation is always $d(v)+\sr(v)$.
In what follows, we will call this value the \textit{weight} of node $v$.
The task of finding the optimal strategy of the attacker boils down to finding the largest subtree with the sum of weights lower than $T$. 
Algorithm~\ref{alg:tree-attack} achieves this goal by filling tables $\tau$ and $\gamma$.

Entry $\tau[v,k]$ contains the weight of the lightest subtree with $k$ nodes rooted at $v$, while $\gamma[v,k]$ contains the nodes in this subtree (in the order of activation).
Computing values of $\tau[v,0]$ and $\tau[v,1]$ (as well as $\gamma[v,0]$ and $\gamma[v,1]$) is trivial.
If node $v$ is not a leaf, we compute $\tau[v,k]$ and $\gamma[v,k]$ for $k>1$ by filling tables $y$ and $q$.

Entry $y[m,k]$ contains the weight of the lightest subtree with $k$ nodes rooted at $v$, constructed by using only the descendants of the first $m$ children of $v$.
Entry $q[m,k]$ contains nodes in this subtree (in the order of activation).
In lines~18-22, we check all possible values of $y[m,k]$ by taking $i$ nodes from the first $m-1$ children of $v$ (and $v$ itself) and adding to them $k-i$ nodes from the subtree of the $m$-th child of $v$.

Instructions in lines~11-12 as well as instructions in lines~24-25 are run $\mathcal{O}(n^2)$ times, since the loop in line~4 is run exactly $n$ times and, for any $v$, we have $t(v) \leq n$.
One can notice that the body of the loop in line~14 is run exactly $n-1$ times in total (because it is a loop over children of node $v$ and every node other than $\vs$ has exactly one parent).
Since both $l \leq n$ and $k \leq n$, instructions in lines~19-22, driving the time complexity of the algorithm, are run $\mathcal{O}(n^3)$ times. 
Hence, the time complexity of Algorithm~\ref{alg:tree-attack} is $\mathcal{O}(n^3)$.
The memory complexity of Algorithm~\ref{alg:tree-attack} is $\mathcal{O}(n^2)$ if we represent sequences in tables $\gamma$ and $q$ as trees visited in a prefix order.
\end{proof}

\section{Dynamic Programming Algorithm for Finding Optimal Attack}
\label{app:dynamic}

In this section, we present a dynamic-programming algorithm to find the optimal strategy of the attacker.
Its pseudocode can be found in Algorithm~\ref{alg:dynamic-programming}.
The algorithm computes the strategy that yields the highest possible utility of the attacker against a given distribution of security resources.
The $\oplus$ symbol denotes the concatenation of sequences.

\begin{algorithm}[tbh]
\setstretch{1.1}
\LinesNumbered
\DontPrintSemicolon
\SetAlgoNoEnd
\SetAlgoNoLine
\KwIn{The network $(V,E)$, the defender's distribution of security resources $\sr$, the seed node $\vs \in V$, and the attacker's time limit $T$.}
\KwOut{Attacker's sequence of activation starting with $\vs$ of maximal length given the time limit.}
\For{$C \subseteq V$}{
	$\et^*[C] \gets \infty$\;
}
$\et^*[\{\vs\}] \gets 0$\;
$\seq^*[\{\vs\}] \gets \langle\vs\rangle$\;
\For{$k=1,\ldots,n-1$}{
	\For{$C \subset V : (|C|=k) \land (\et^*[C] < \infty)$}{
		\For{$v \in V : (v \notin C) \land (N(v) \cap C \neq \emptyset)$}{
			$\Delta \et \gets \frac{d(v)+\sr(v)}{|N(v) \cap C|}$\;
			\If{$(\et^*[C]+\Delta \et < \et^*[C\cup\{i\}]) \land (\et^*[C]+\Delta \et < T)$}{
				$\et^*[C\cup\{i\}] \gets \et^*[C]+ \Delta \et$\;
				$\seq^*[C\cup\{i\}] \gets \seq^*[C] \oplus \langle i\rangle$\;
			}
		}
	}
	\If{$\lnot \exists_{C \subseteq V}|C|=k+1 \land \et^*[C] < \infty$}{
		\Return $\seq^*[\argmin_{C \subseteq V : |C|=k}\et^*[C]]$\;
	}
}
\Return $\seq^*[V]$\;
\caption{A dynamic-programming algorithm for the optimal strategy of the attacker.}
\label{alg:dynamic-programming}
\end{algorithm}

In lines~1-4 we initialize data structures, \ie, array $\et^*$ holding the minimal time necessary to activate a given set of nodes and array $\seq^*$ holding the sequence allowing to activate a given set of nodes in this minimal time.
In $k$-th run of loop in line~5 we fill both arrays for these subsets of size $k+1$ that can be activated within time limit.
We do it by iterating over the subsets of size $k$ that can be activated within time limit in loop in line~6 and trying to activate next node in loop in line~7.
In line~8 we compute the time necessary to active node $v$ when previously we activated nodes in $C$ in the best possible way.
In lines~9-11 we update the way of activating nodes $C \cup \{v\}$ if activating nodes in $C$ first and node $v$ after that is better then currently known best solution.
We end the algorithm either when it is no longer possible to activate any set of the size $k+1$ (in lines~12-13) or when we find a way of activating the entire network (in line~14).

Time complexity of the algorithm is $\mathcal{O}(2^n n^2)$, as for every of the $\mathcal{O}(2^n)$ subsets of the set of nodes we might need to iterate over $\mathcal{O}(n)$ of its neighbors and for each of them compute their time of activation in time $\mathcal{O}(n)$.
Notice that arrays can be implemented using hash tables.
Given the exponential time complexity of the algorithm, it is only suitable for small networks.

\section{Experimental Results with Mixed-Integer Linear Programming}
\label{app:tiny-networks}

\noindent
To test the effectiveness of mixed-integer linear programming in finding the optimal strategy of the defender we perform experiments on randomly generated networks.
In order to avoid selecting arbitrary set of strategies $\avs$ available to the attacker, we allow the attacker to choose any strategy.
However, under these assumptions, we are able to solve mixed-integer linear programming instance efficiently only for very small networks.
A potential idea for future work is studying whether there exists a way of finding a set $\avs$ of limited size that is still guaranteed to always contain the optimal attacker strategy. 

Figures~\ref{fig:best-bars-milp} and~\ref{fig:heat-milp} present our results for networks with $6$ nodes and average degree $2$, taken as an average over $100$ simulations.
Mixed-integer linear programming were solved using \textit{lp\_solve} solver version $5.5.2.5$.

As it can be seen, distribution of security resources computed using mixed-integer linear programming is considerably more successful in mitigating the attack than most of the other defense strategy, with centrality-based heuristics closely following.
This suggests that, at least in the case of simple structures of very small networks, centrality measures may be close to capturing the complexity of the problem of finding the optimal defense.

Another difference in comparison to results for larger networks is that all considered strategies of the attacker achieve very similar results.
Very small number of nodes results in fairly limited variety of valid strategies available to the attacker.

\begin{figure}[ht]
\centering
\setlength\tabcolsep{0pt}
\begin{tabular}{m{.5\textwidth}m{.5\textwidth}}
\multicolumn{1}{c}{Against optimal strategy} &
\multicolumn{1}{c}{Against best heuristic} \\
\includegraphics[width=\linewidth]{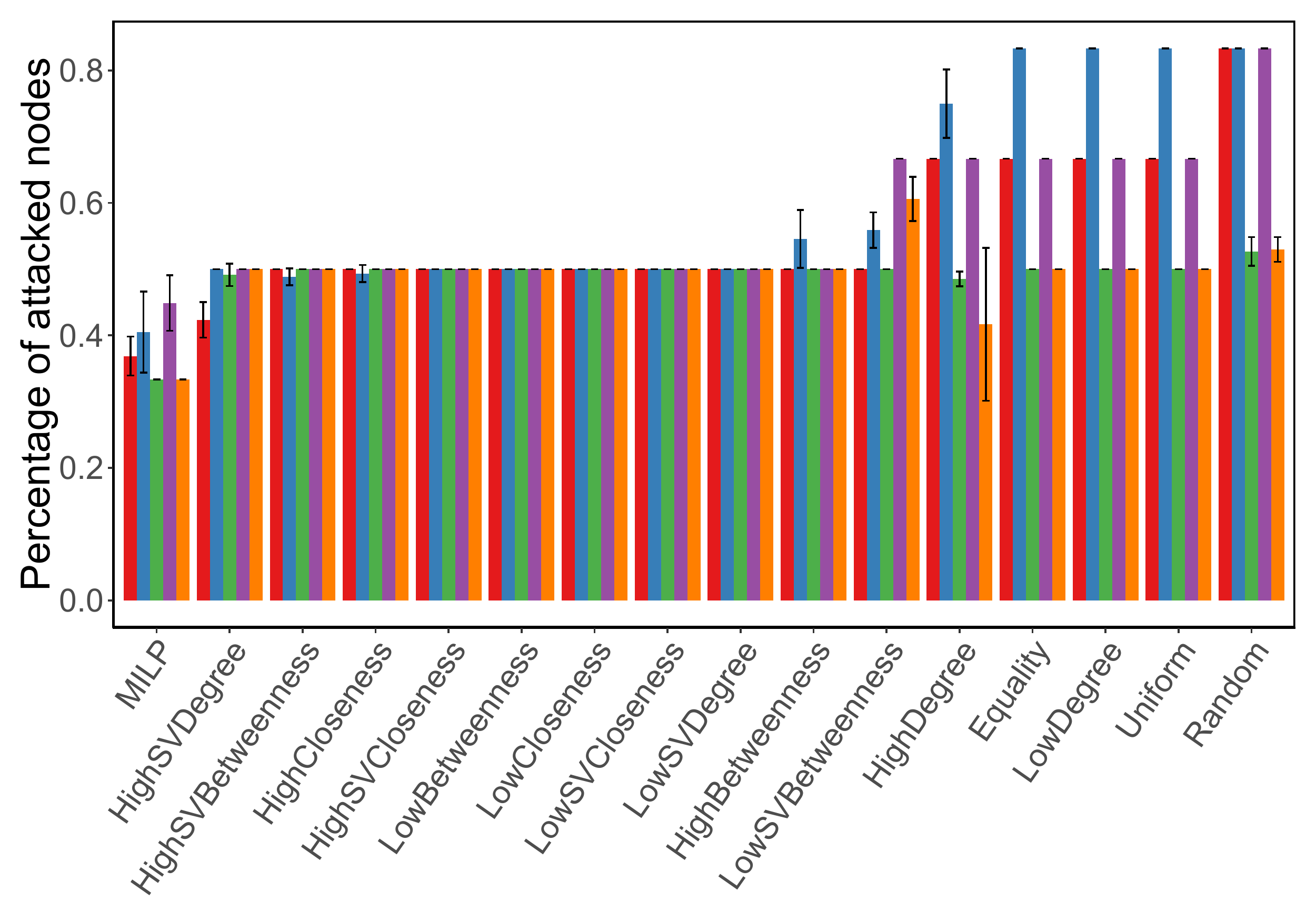} &
\includegraphics[width=\linewidth]{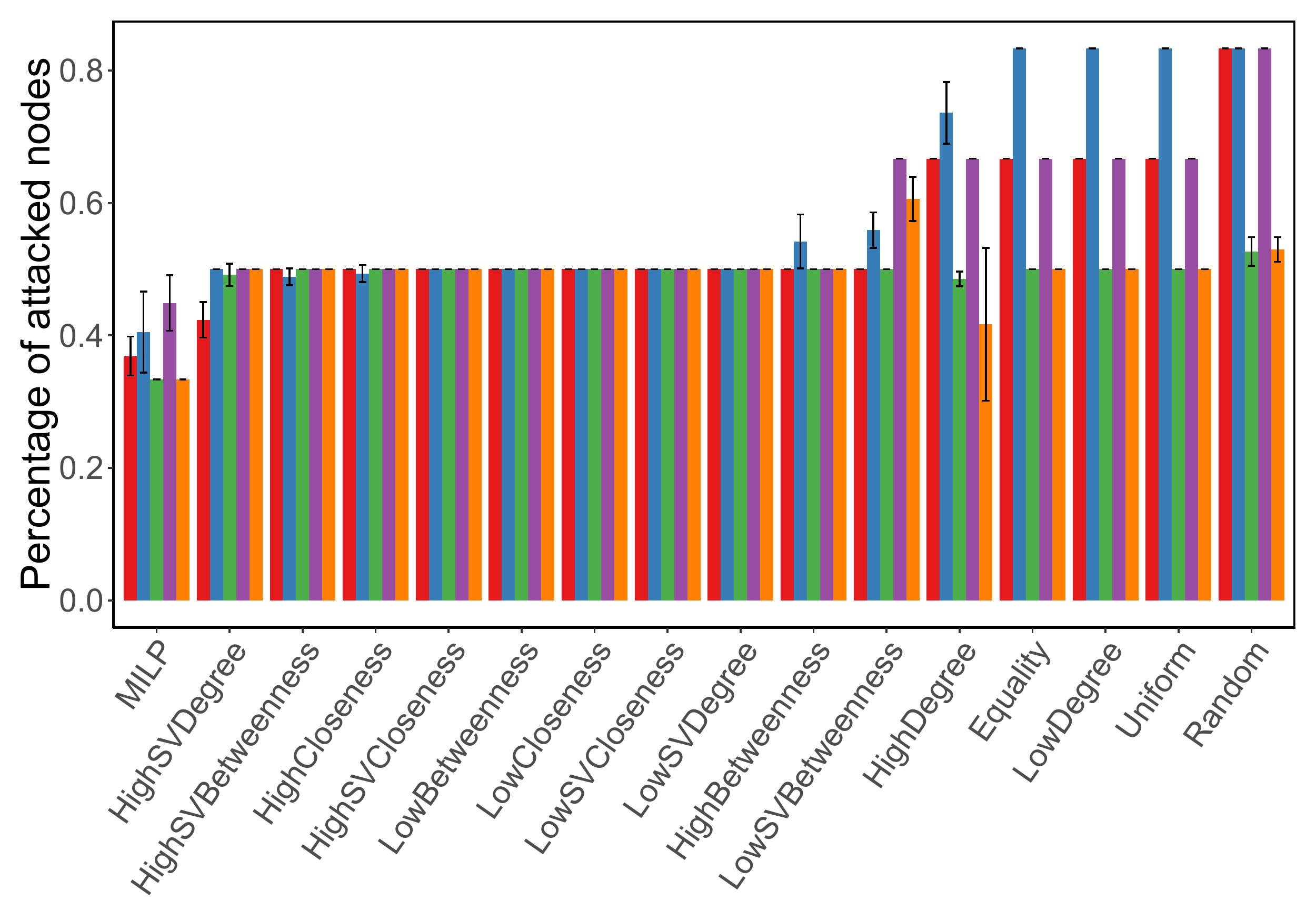} \\
\multicolumn{2}{c}{\includegraphics[width=.7\linewidth]{figures/plots/best-bars/random-legend}} \\
\end{tabular}
\caption{
Comparison of effectiveness of defender's strategies, under assumption that the attacker uses optimal strategy or the best heuristic considered by us.
Results are taken as an average over $100$ simulations, with a new network generated for each simulation using one of the models.
Defender's strategies are sorted from most to least effective on average.
Error bars represent $95\%$ confidence intervals.
}
\label{fig:best-bars-milp}
\end{figure}

\begin{figure}[t]
\centering
\setlength\tabcolsep{0pt}
\begin{tabular}{m{.32\textwidth}m{.32\textwidth}m{.32\textwidth}}
\multicolumn{1}{c}{Preferential attachment networks} &
\multicolumn{1}{c}{Scale free networks} &
\multicolumn{1}{c}{Random graph networks} \\
\includegraphics[width=\linewidth]{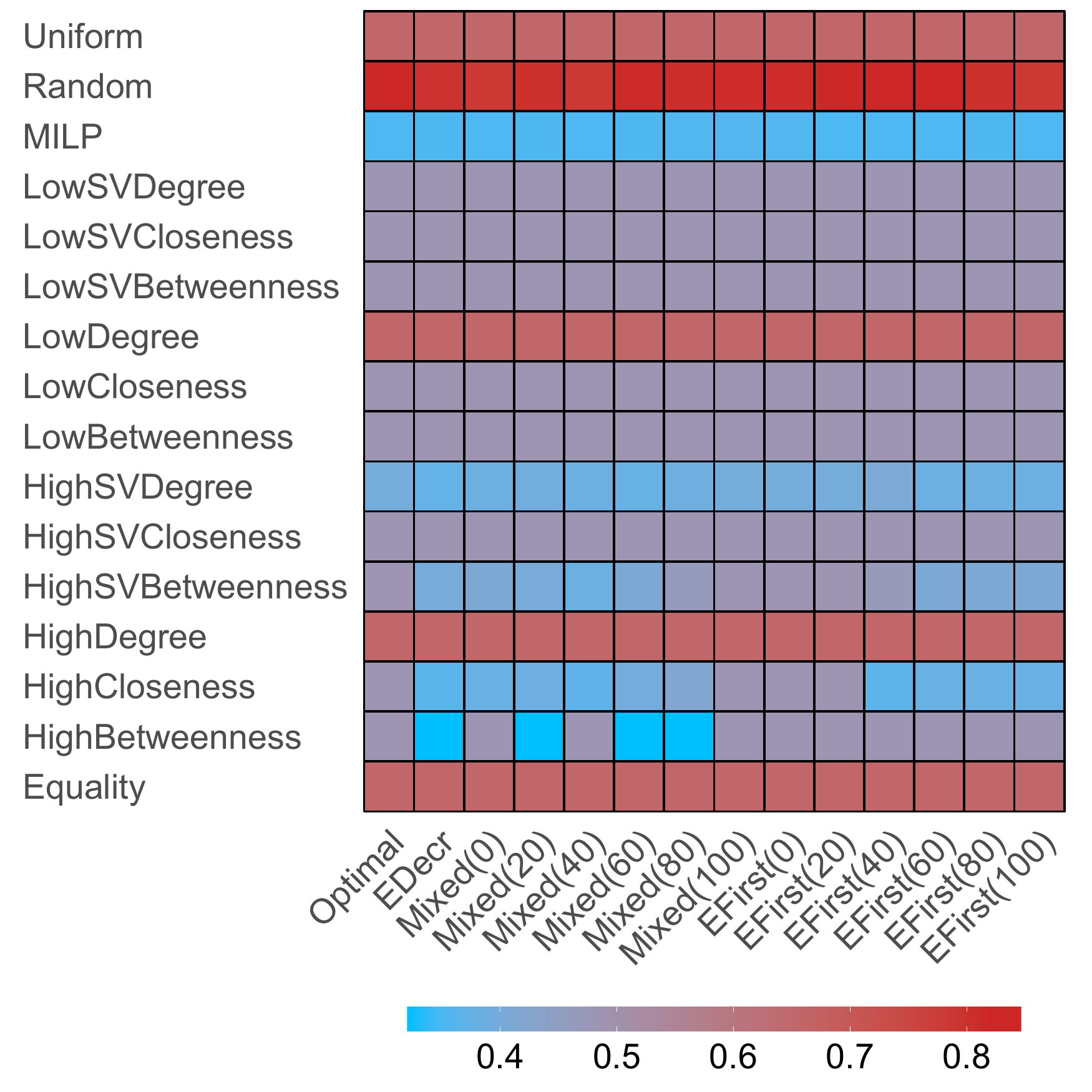} &
\includegraphics[width=\linewidth]{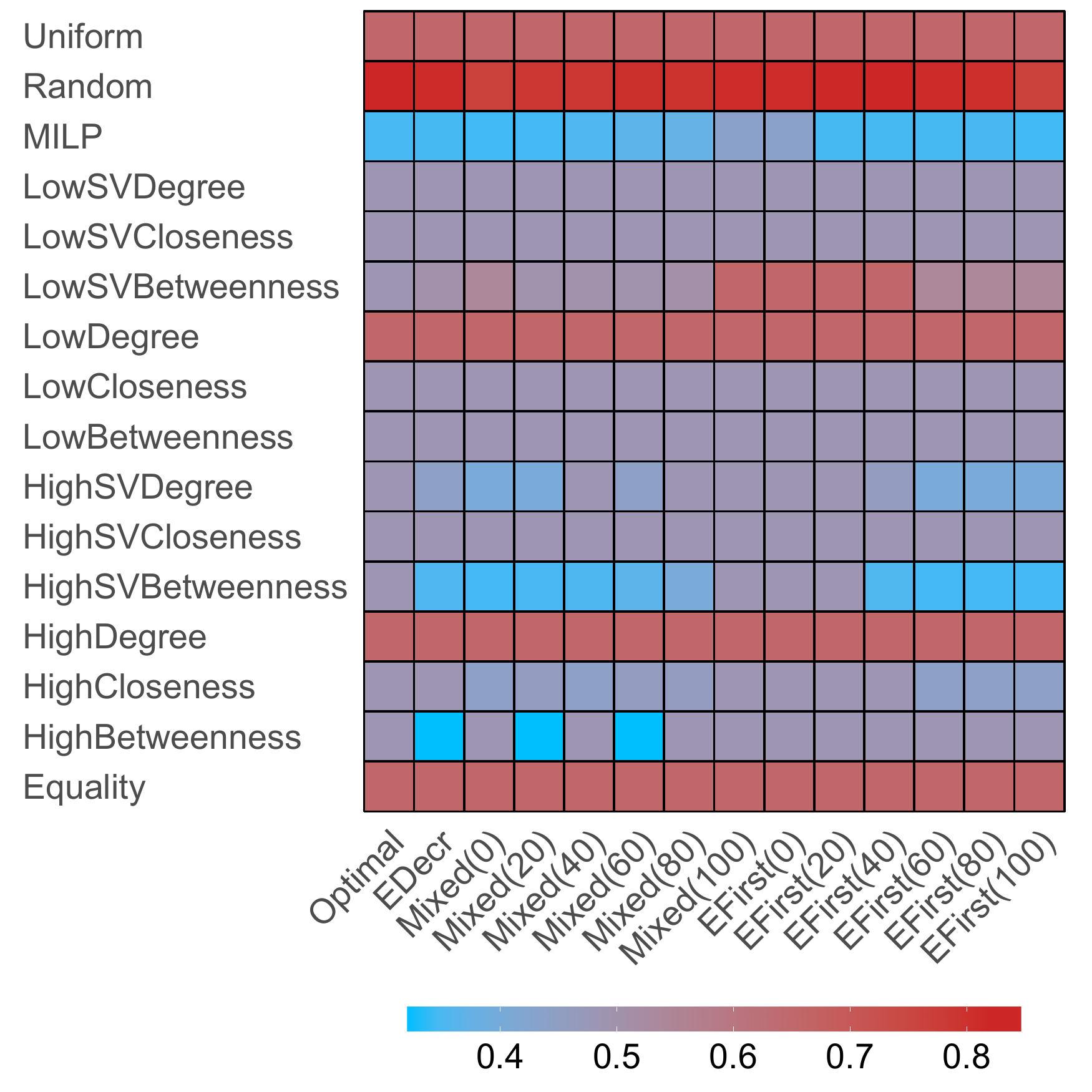} &
\includegraphics[width=\linewidth]{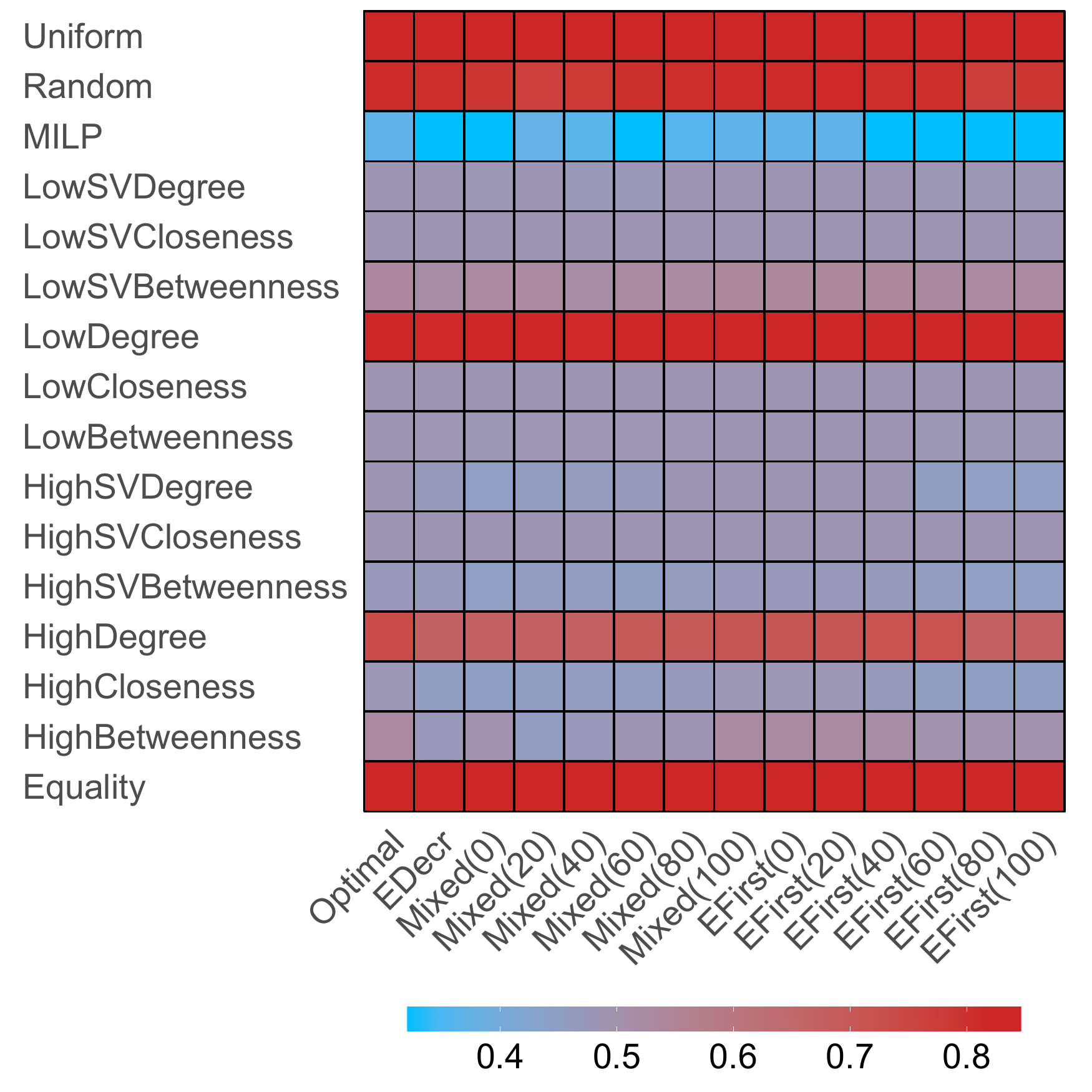} \\
\end{tabular}
\begin{tabular}{m{.32\textwidth}m{.32\textwidth}}
\multicolumn{1}{c}{Small world networks} &
\multicolumn{1}{c}{Random trees} \\
\includegraphics[width=\linewidth]{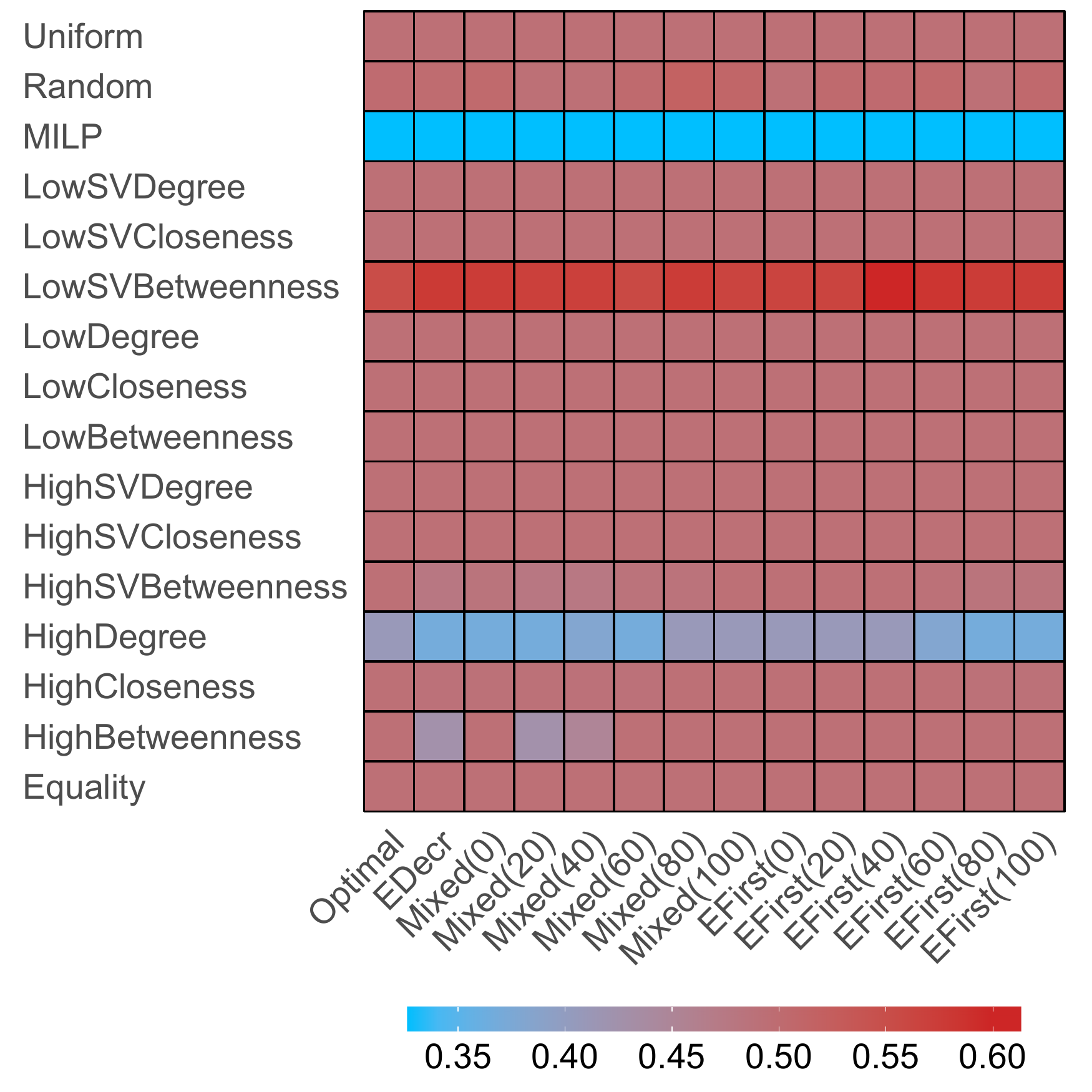} &
\includegraphics[width=\linewidth]{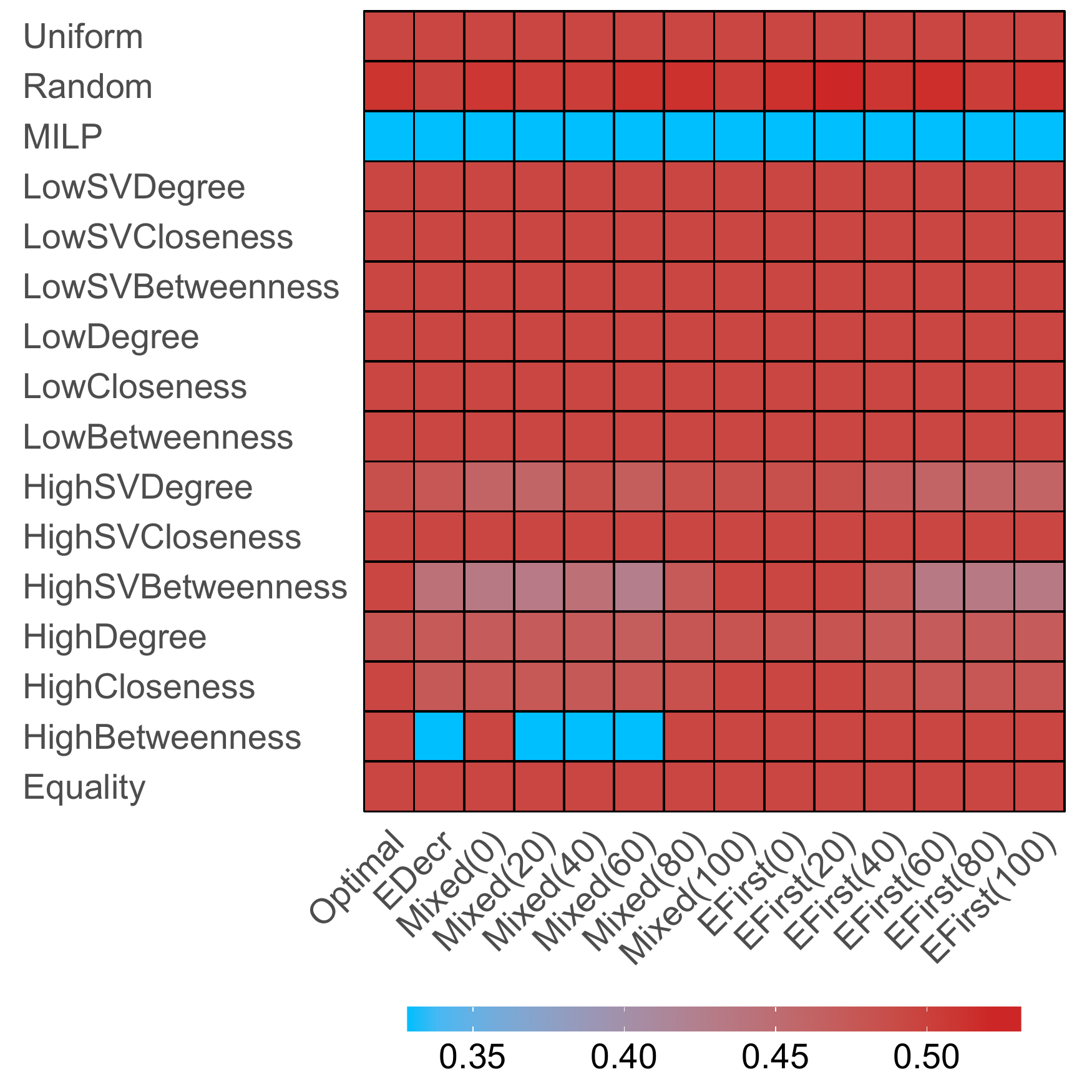} \\
\end{tabular}
\caption{
Comparison of effectiveness of defender's strategies for different attacker's strategies.
Color of each cell represents the expected percentage of nodes successfully activated by the attacker.
Results are taken as an average over $100$ simulations, with a new network generated for each simulation using one of the models.
}
\label{fig:heat-milp}
\end{figure}

\section{Experimental Results with Multiple Seeds}
\label{app:multiple-seeds}

The model of strategic diffusion~\cite{alshamsi2018optimal} allows only one seed node.
We now investigate the potential advantages of activating multiple seeds at the beginning of the process.
We assume the time of activating each subsequent seed to be computed using the same formula as for the first seed in the basic version of our model, \ie, $\et(v) = d(v) + \sr(v)$ for $v$ that is a seed node.
We choose the subsequent seeds using the same method as when choosing the first seed in the basic version of the model, \ie, selecting those that provide the highest number of activated nodes (breaking ties uniformly at random).
In the remainder of the process we follow the same procedure as in the basic version of the model, \ie, we choose the subsequent targets from the nodes with at least one active neighbour.

Figures~\ref{fig:heat-multi-80-1} and~\ref{fig:heat-multi-80-2} present the results of our experiments.
As it can be seen, forcing the attacker to pick multiple seeds increases the effectiveness of the attack in most cases, as the expected percentage of successfully attacked nodes in higher in columns with two or three seeds activated (colour of the hetmap cells is closer to red).
An interesting exception from this rule are the random graph networks generated using Erd{\H{o}}s-R{\'e}nyi model.
The uniform structure of these networks does not offer much advantage in terms of activating additional seeds.

It is also worth noting that the general trends in data are preserved when changing the number of seeds, \eg, strategies of the defender that are effective in experiments with a single seeds tend to also be effective in a setting with multiple seeds.
Finally, although the difference in results when using multiple seeds is notable, the magnitude of the change is not particularly substantial, \eg, the difference between the percentage of successfully attacked nodes with only one seed and with two seeds is on average lesser than $0.01$.

\begin{figure}[b]
\centering
\setlength\tabcolsep{0pt}
\begin{tabular}{m{.03\textwidth}m{.32\textwidth}m{.32\textwidth}m{.32\textwidth}}
&
\multicolumn{1}{c}{One seed} &
\multicolumn{1}{c}{Two seeds} &
\multicolumn{1}{c}{Three seeds} \\
\rotatebox{90}{Preferential attachment networks} &
\includegraphics[width=\linewidth]{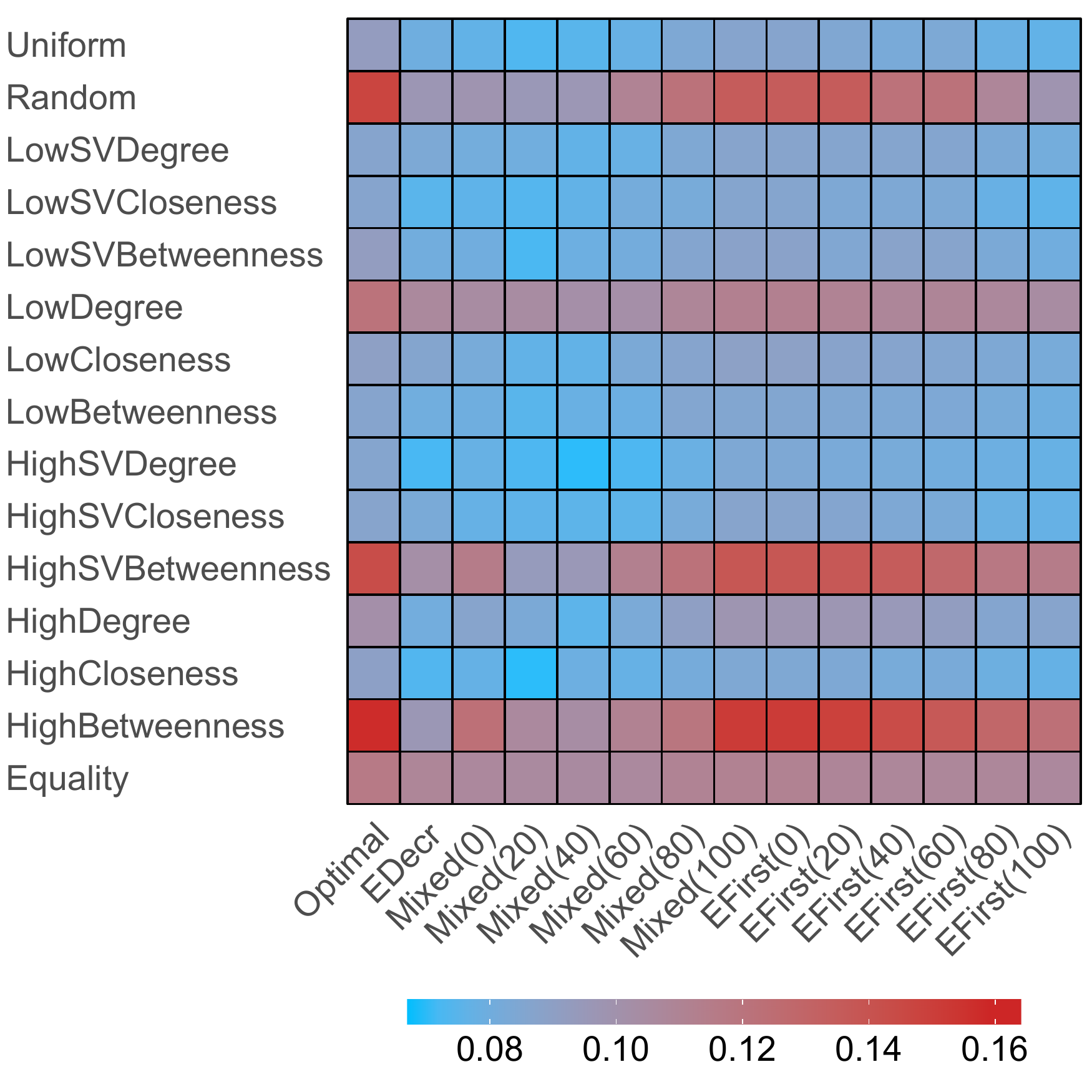} &
\includegraphics[width=\linewidth]{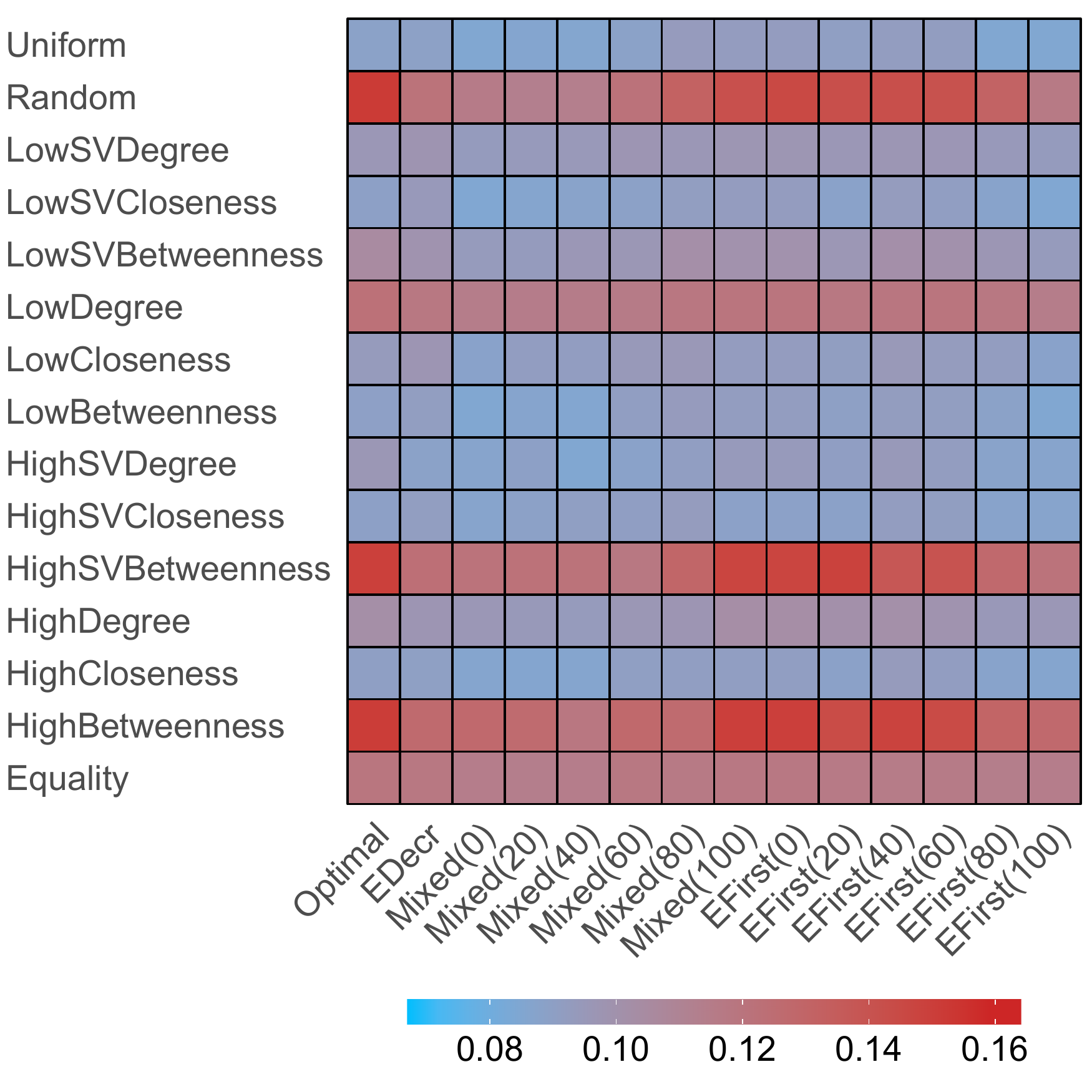} &
\includegraphics[width=\linewidth]{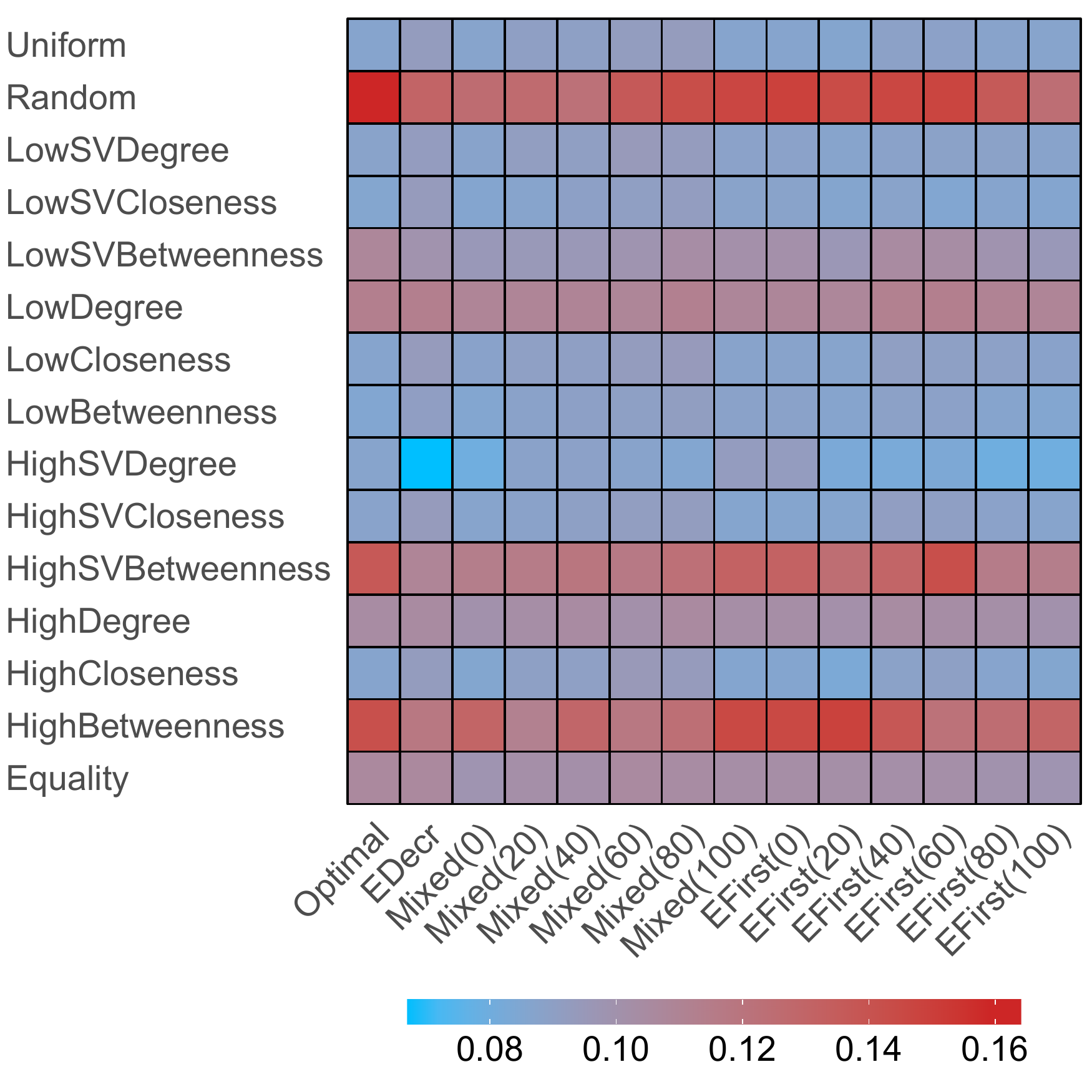} \\
\rotatebox{90}{Scale free networks} &
\includegraphics[width=\linewidth]{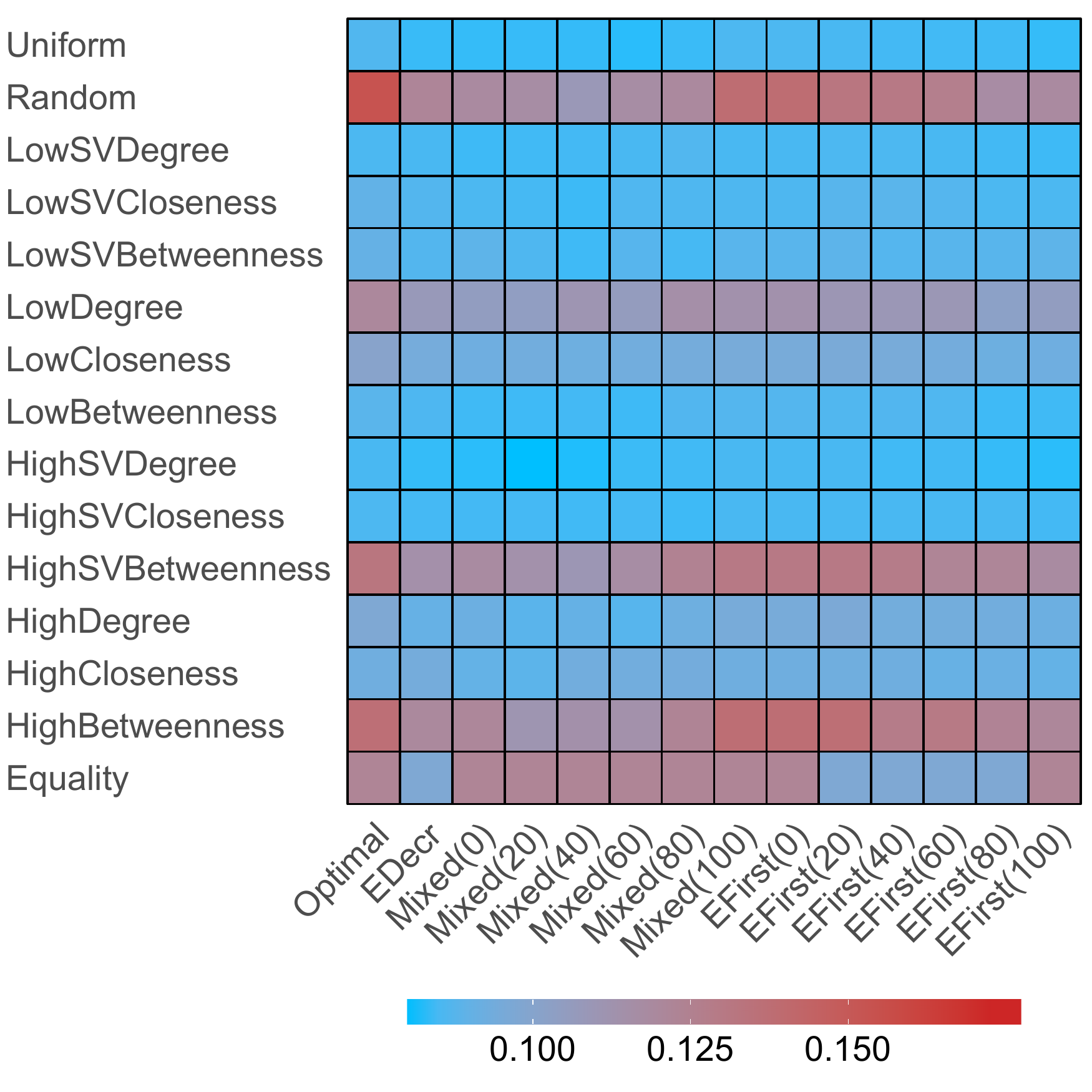} &
\includegraphics[width=\linewidth]{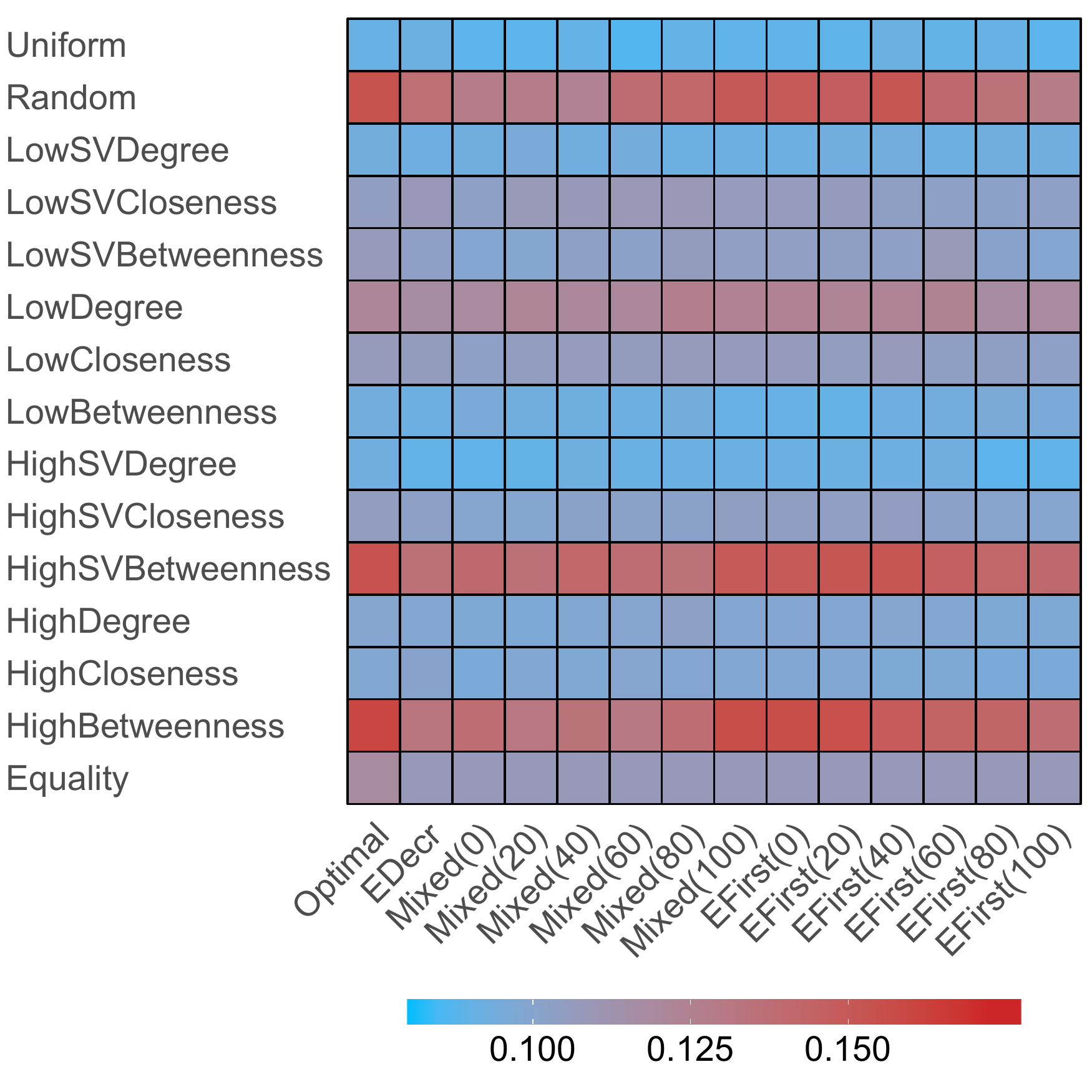} &
\includegraphics[width=\linewidth]{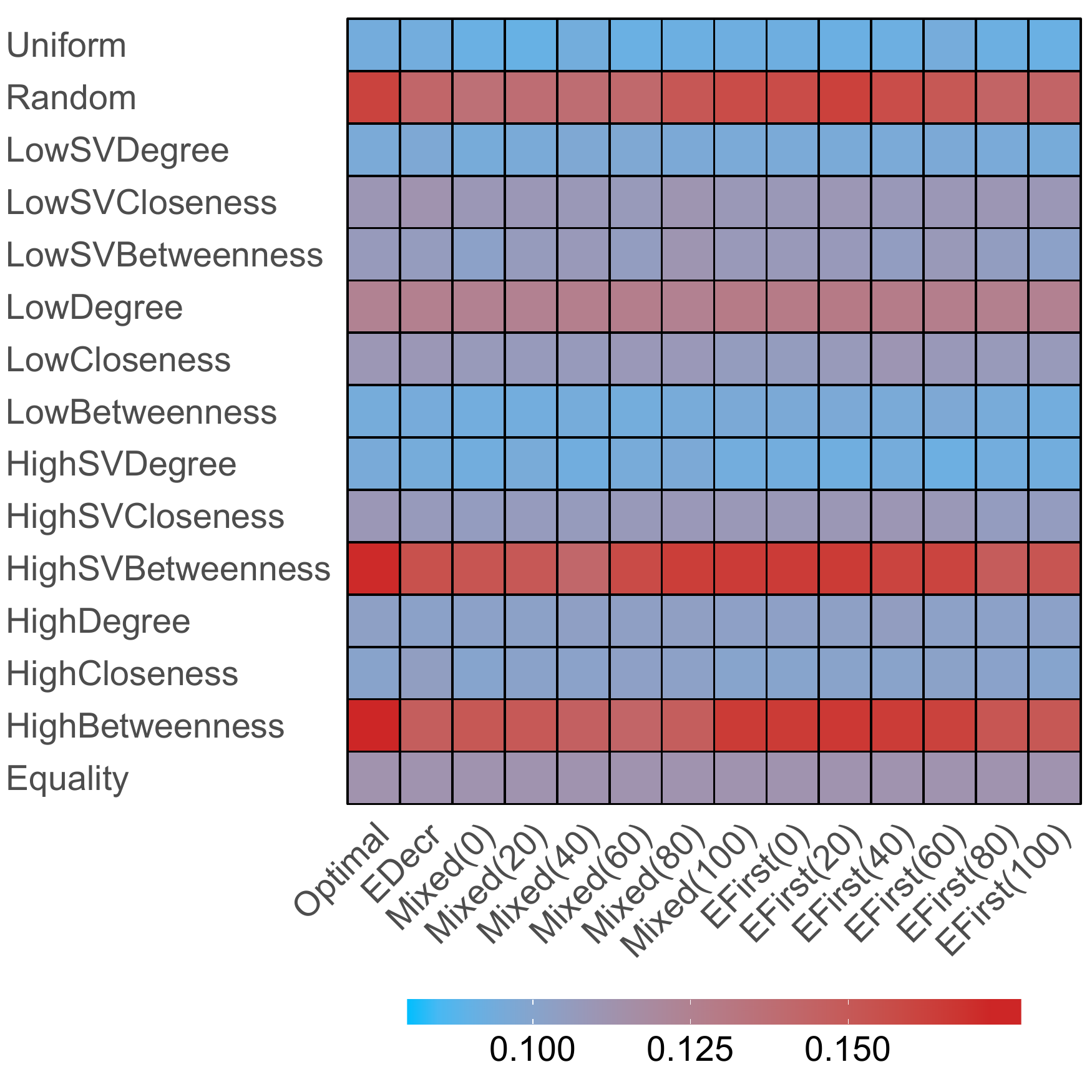} \\
\rotatebox{90}{Random graph networks} &
\includegraphics[width=\linewidth]{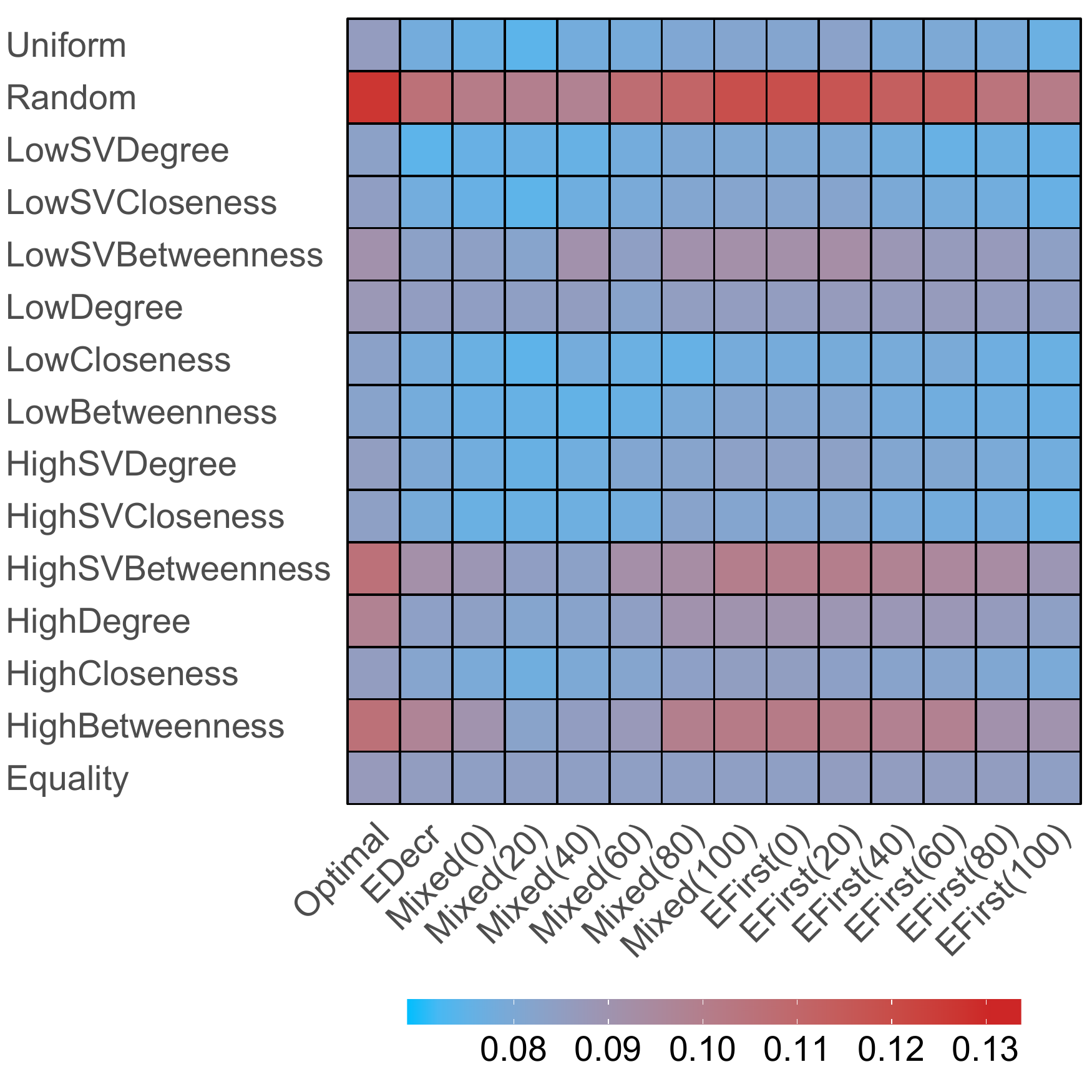} &
\includegraphics[width=\linewidth]{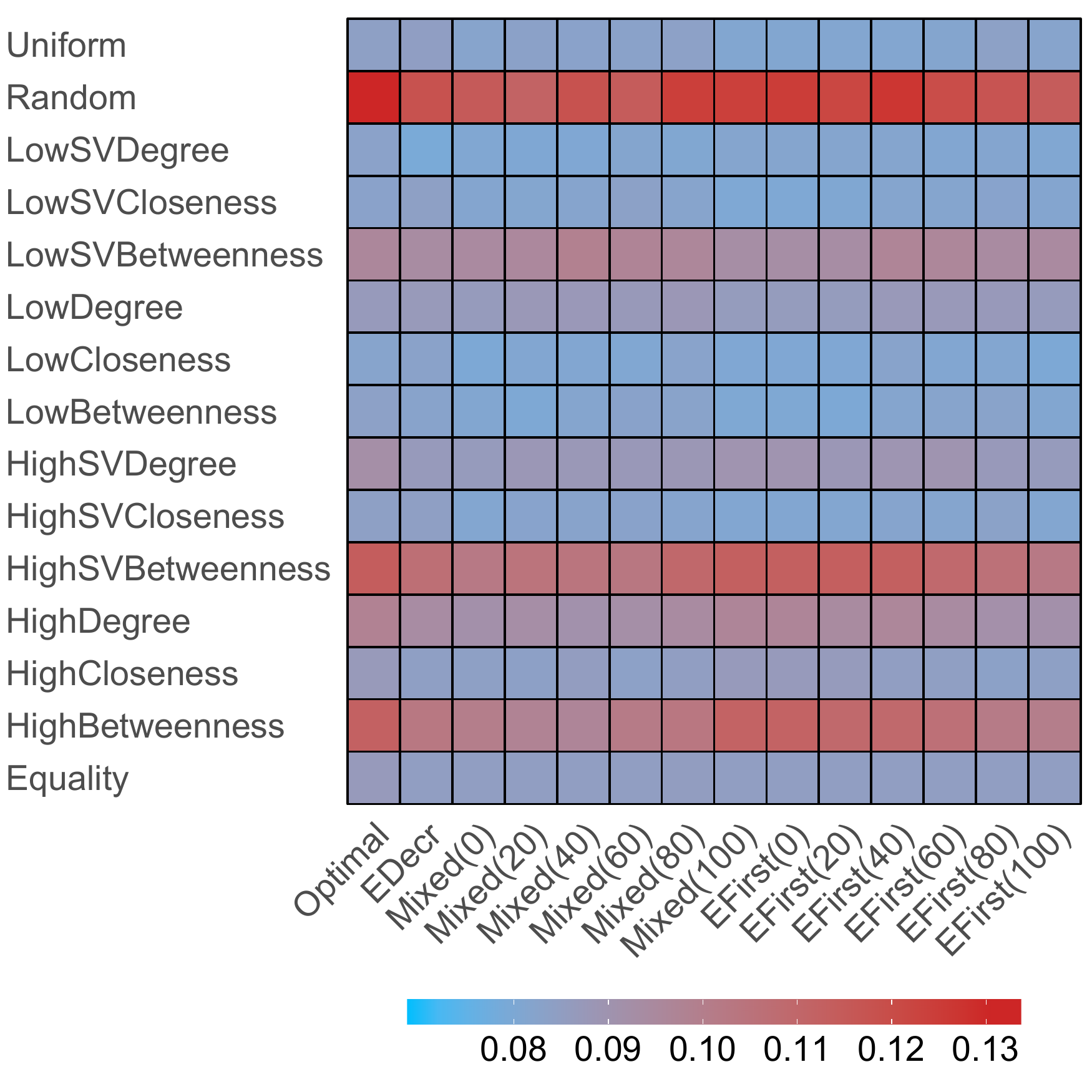} &
\includegraphics[width=\linewidth]{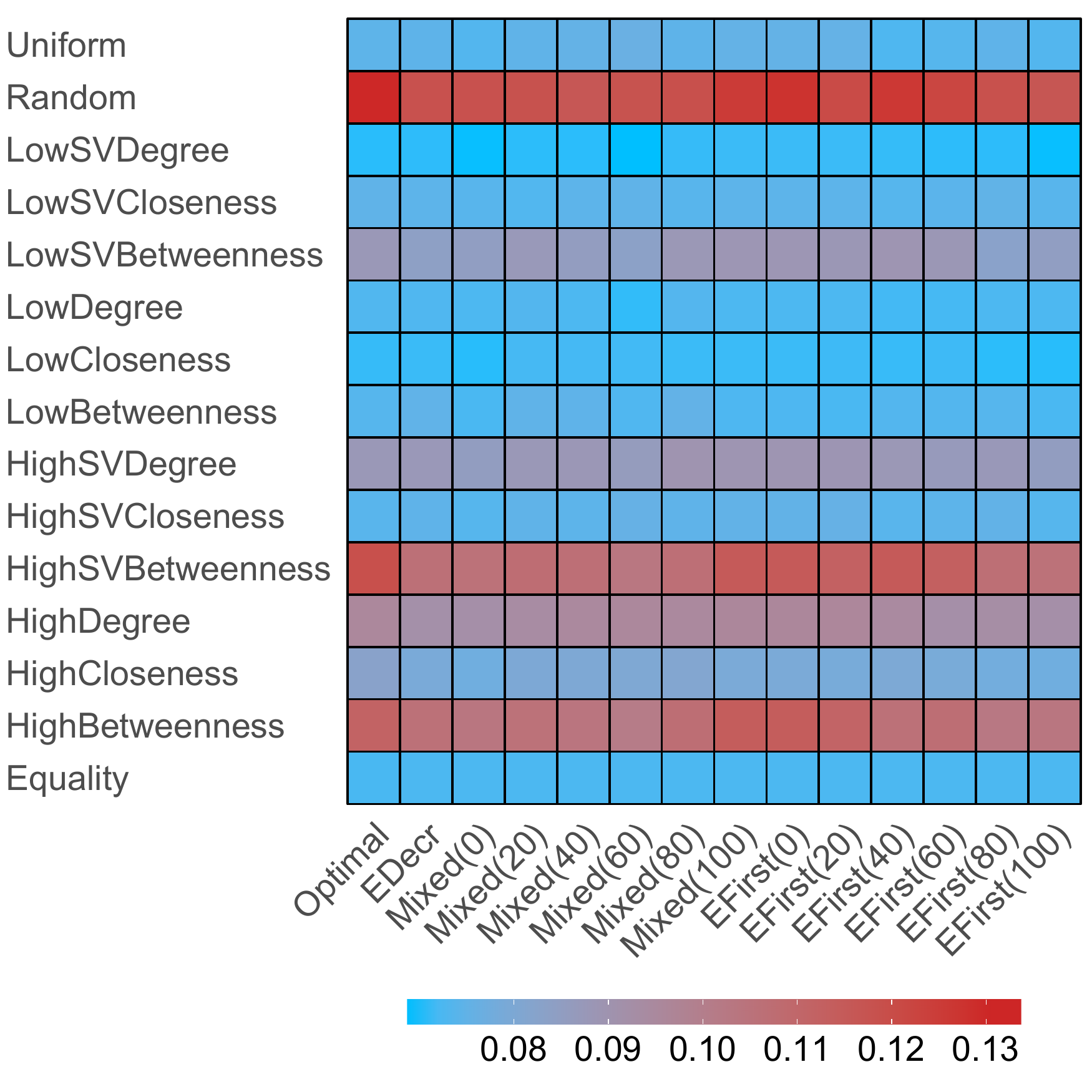} \\
\end{tabular}
\caption{
Comparison of effectiveness of defender's strategies for different attacker's strategies on networks with $80$ nodes, when attacker uses one, two, and three seeds.
Color of each cell represents the expected percentage of nodes successfully activated by the attacker.
Results are taken as an average over $100$ simulations, with a new network generated for each simulation using one of the models.
Scales are fixed for each network type to allow comparison between the number of seeds.
}
\label{fig:heat-multi-80-1}
\end{figure}

\begin{figure}[b]
\centering
\setlength\tabcolsep{0pt}
\begin{tabular}{m{.03\textwidth}m{.32\textwidth}m{.32\textwidth}m{.32\textwidth}}
&
\multicolumn{1}{c}{One seed} &
\multicolumn{1}{c}{Two seeds} &
\multicolumn{1}{c}{Three seeds} \\
\rotatebox{90}{Small world networks} &
\includegraphics[width=\linewidth]{figures/plots/multiHeat/ba-80-3-1-heatMulti} &
\includegraphics[width=\linewidth]{figures/plots/multiHeat/ba-80-3-2-heatMulti} &
\includegraphics[width=\linewidth]{figures/plots/multiHeat/ba-80-3-3-heatMulti} \\
\rotatebox{90}{Random trees} &
\includegraphics[width=\linewidth]{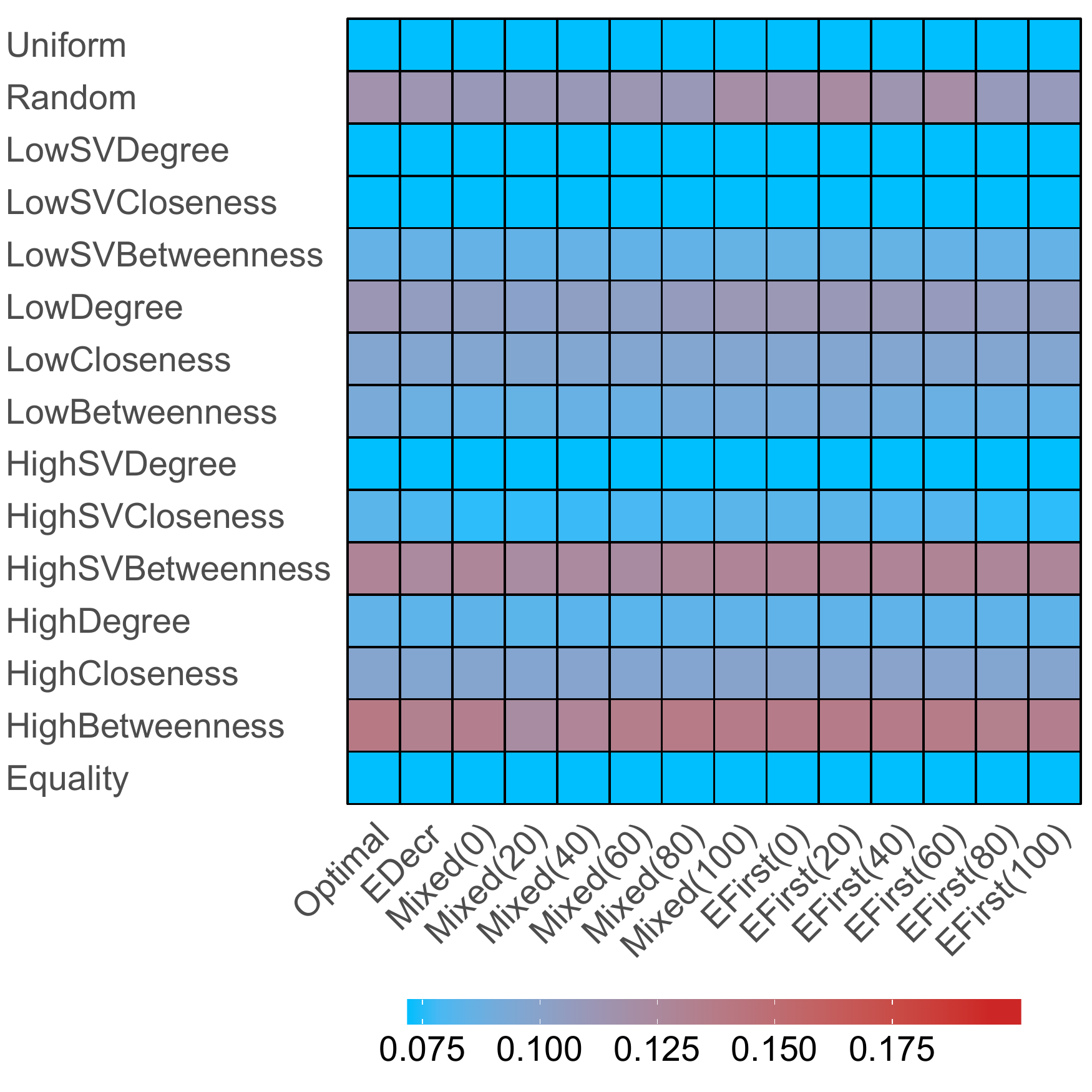} &
\includegraphics[width=\linewidth]{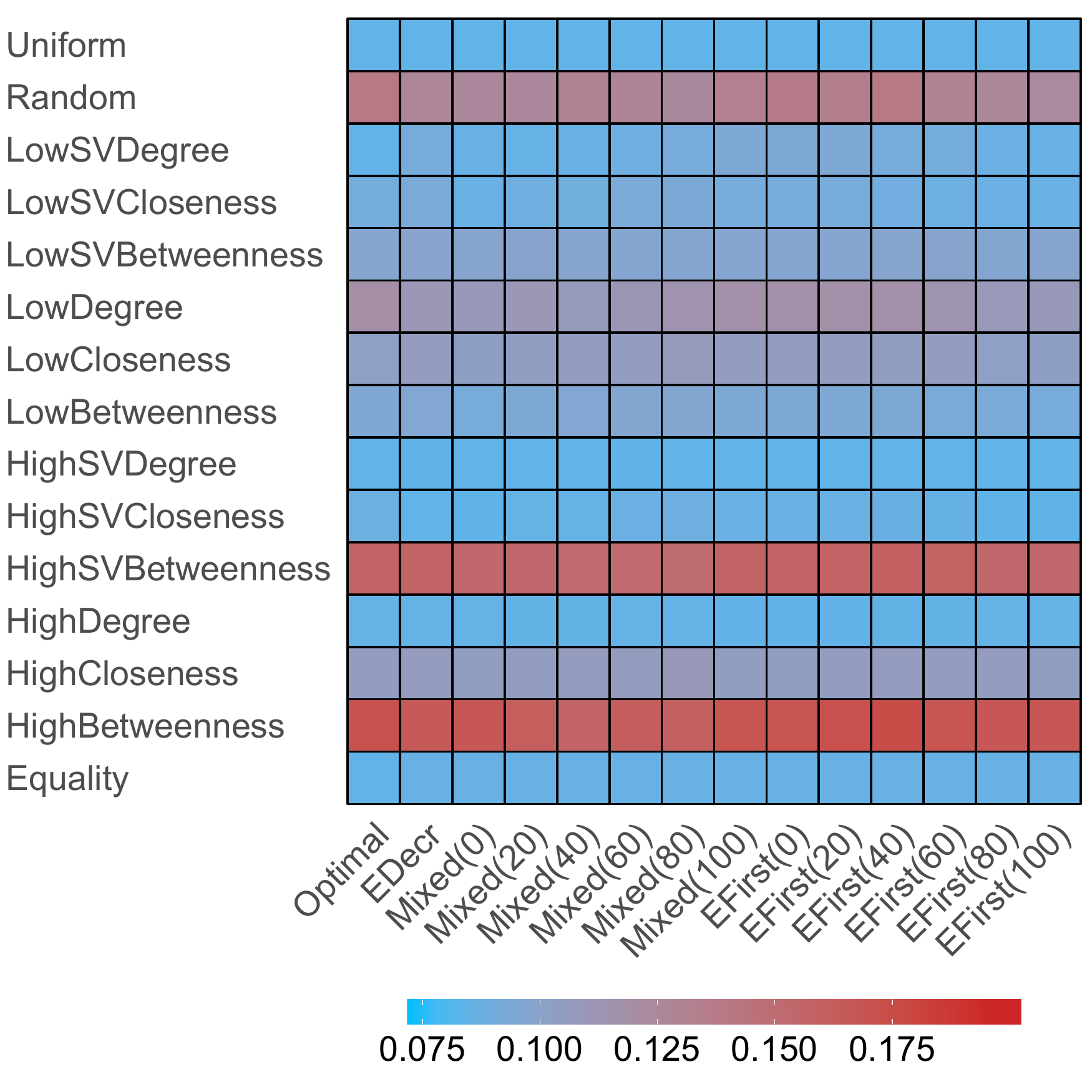} &
\includegraphics[width=\linewidth]{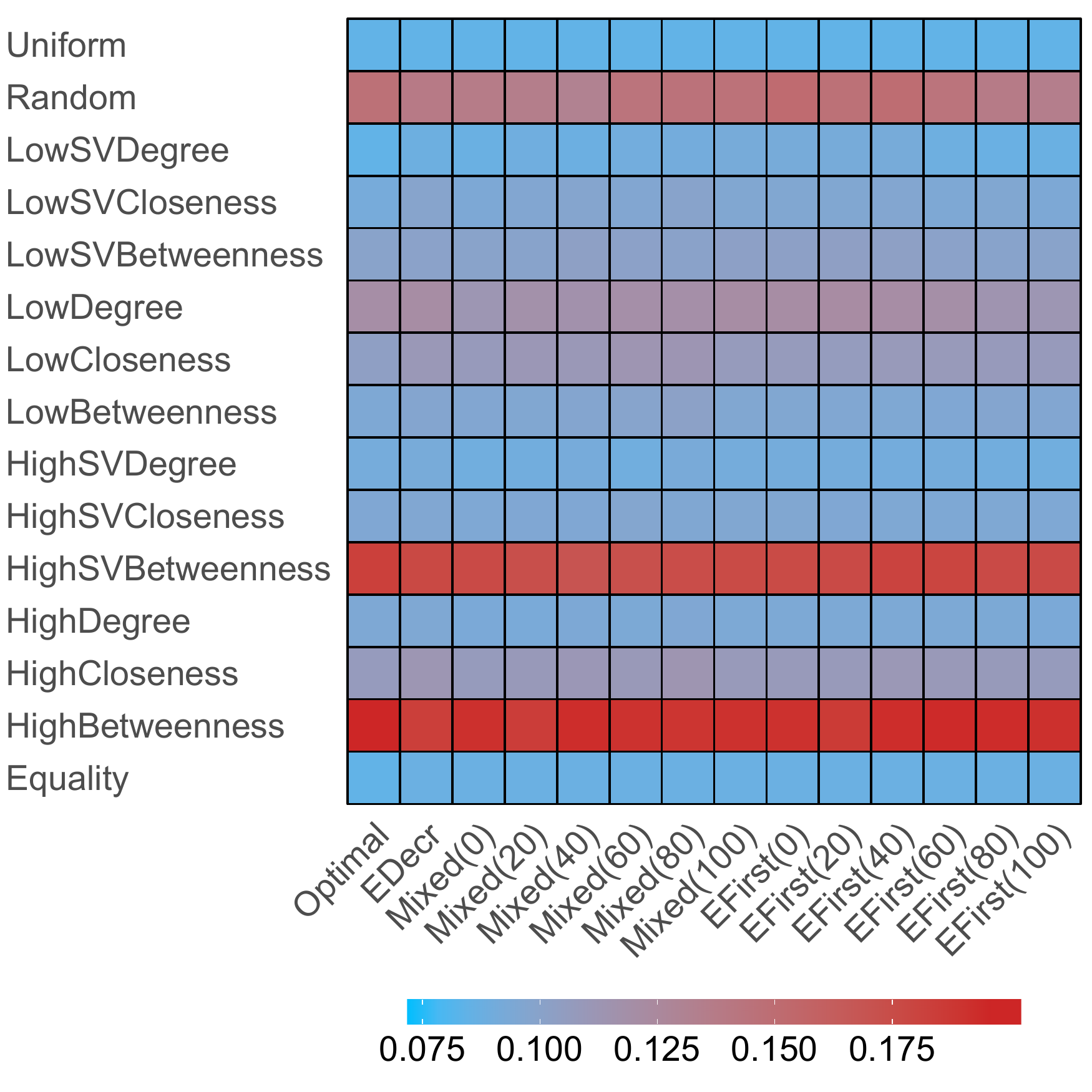} \\
\end{tabular}
\caption{
Comparison of effectiveness of defender's strategies for different attacker's strategies on networks with $80$ nodes, when attacker uses one, two, and three seeds.
Color of each cell represents the expected percentage of nodes successfully activated by the attacker.
Results are taken as an average over $100$ simulations, with a new network generated for each simulation using one of the models.
Scales are fixed for each network type to allow comparison between the number of seeds.
}
\label{fig:heat-multi-80-2}
\end{figure}

\section{Repeated Game Setting}
\label{app:repeated-setting}

We now consider a process in which the security diffusion game is played multiple times between the same two players.
We investigate whether the knowledge about the attacker's activities can help the defender achieve better performance.

We consider a process consisting of rounds.
In each round, the defender and the attacker play a single instance of the security diffusion game on the same network.
We assume that the attacker always utilizes the Optimal strategy, while the defender updates her distribution of security resources according to results of the previous round.
The total amount of security resources remains the same.

In what follows let $\sr_j(v_i)$ denote the amount of security resources assigned to node $v_i$ in round $j$, and let $I_j$ denote the set of nodes activated in round $j$.

The value of $\sr_1$ is given by the strategy used by the defender in the first round, which is one of the strategies described for the single game setting.
For every subsequent round $j > 1$ the defender updates the security resources assignment as follows:
$$
\sr_j(v_i) = \frac{\omega_j(v_i)}{\sum_{i=1}^{n} \omega_j(v_i)}\SR
$$
\noindent where:
$$
\omega_j(v_i) =
\begin{cases}
r_{i,j} \sr_{j-1}(v_i) & \text{ if } v_i \in I_{j-1} \\
\sr_{j-1}(v_i)  & \text{otherwise} \\
\end{cases}
$$
with $r_{i,j}$ being drawn uniformly at random from interval $[1,2]$.
In other words, the defender increases the amount of security resources assigned to nodes that were activated in the previous round.

Figure~\ref{fig:repeated} presents results for our repeated game setting.
Each data point is taken as an average over $100$ simulations.
Colored areas represent $95\%$ confidence interval.

As it can be seen, updating the strategy that initially is not particularly successful leads to reducing the number of nodes captured by the attacker.
However, in case of the defender strategies that are the most effective in the first round, moving the security resources to the nodes that got successfully attacked not only does not help, but it actually makes the network more vulnerable, \ie, in the following rounds the average number of captured nodes increases.
By moving the security resources to nodes that were attacked in the previous round, the defender makes other nodes more vulnerable, which is then exploited by the attacker.

Interestingly, in most cases almost all initial strategies converge to similar efficiency.
It suggests that after a high enough number of rounds an equilibrium state is reached.

\begin{figure*}[tbh]
\centering
\setlength\tabcolsep{0pt}
\begin{tabular}{m{.32\textwidth}m{.32\textwidth}m{.32\textwidth}}
\multicolumn{1}{c}{Preferential attachment networks} &
\multicolumn{1}{c}{Scale free networks} &
\multicolumn{1}{c}{Random graph networks} \\
\includegraphics[width=\linewidth]{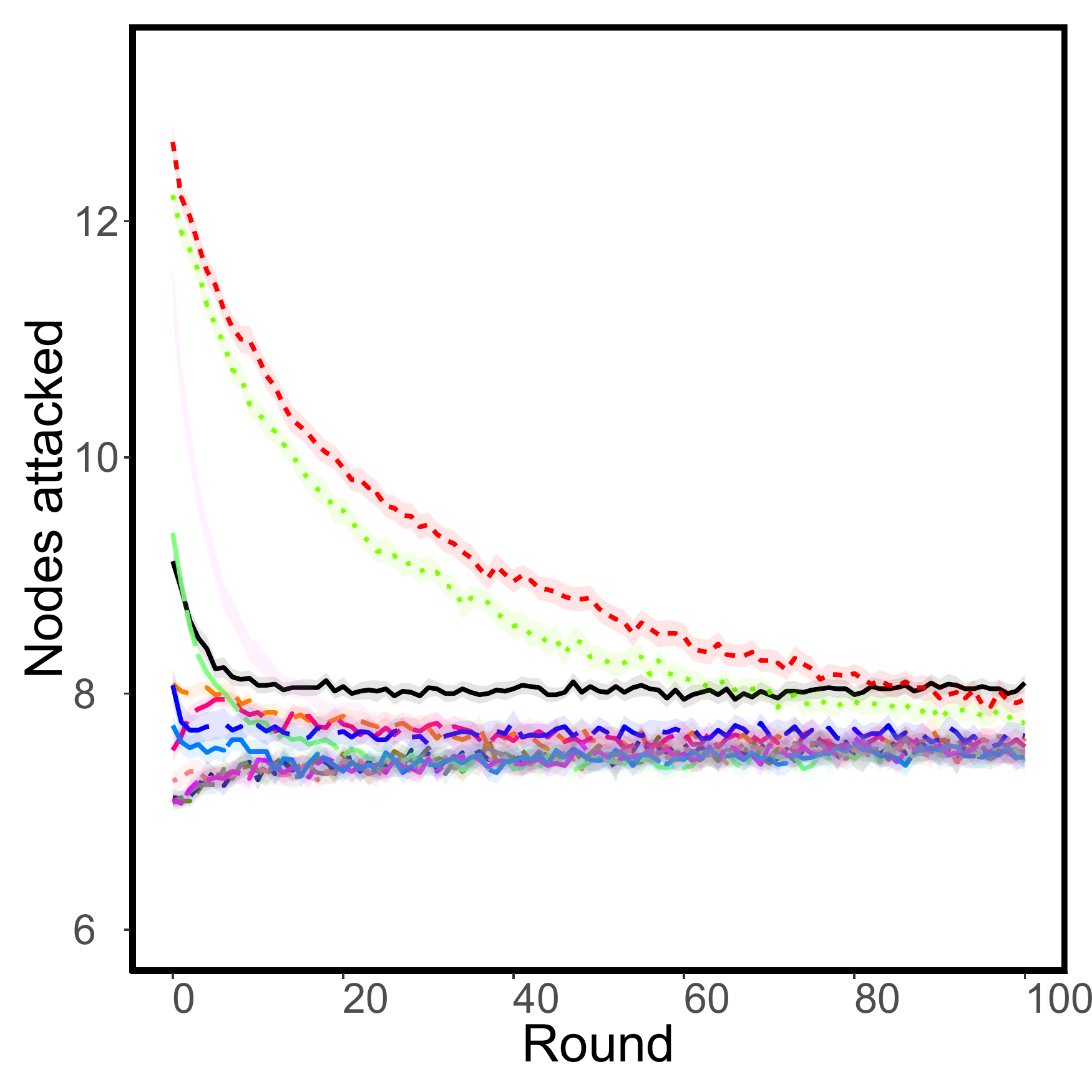} &
\includegraphics[width=\linewidth]{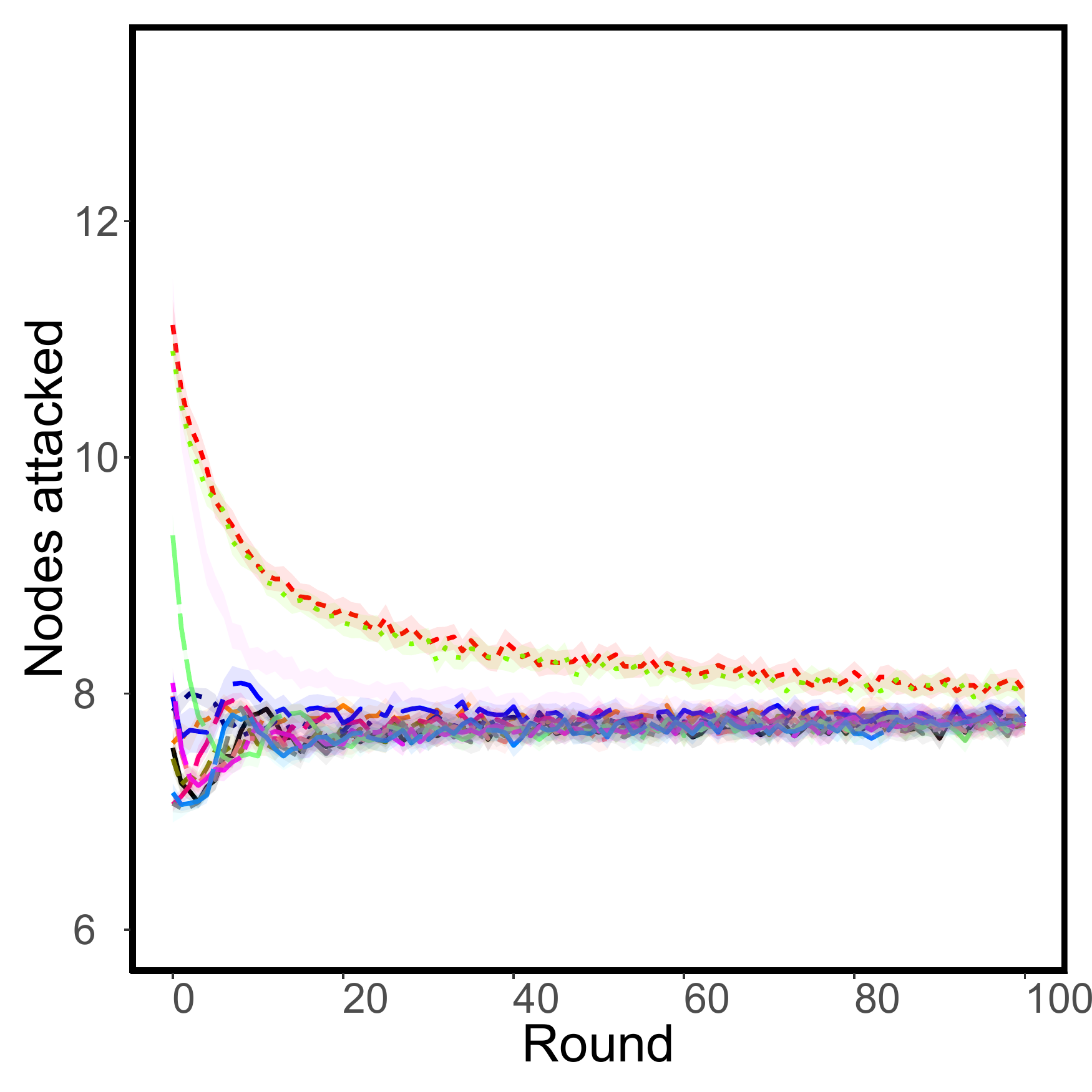} &
\includegraphics[width=\linewidth]{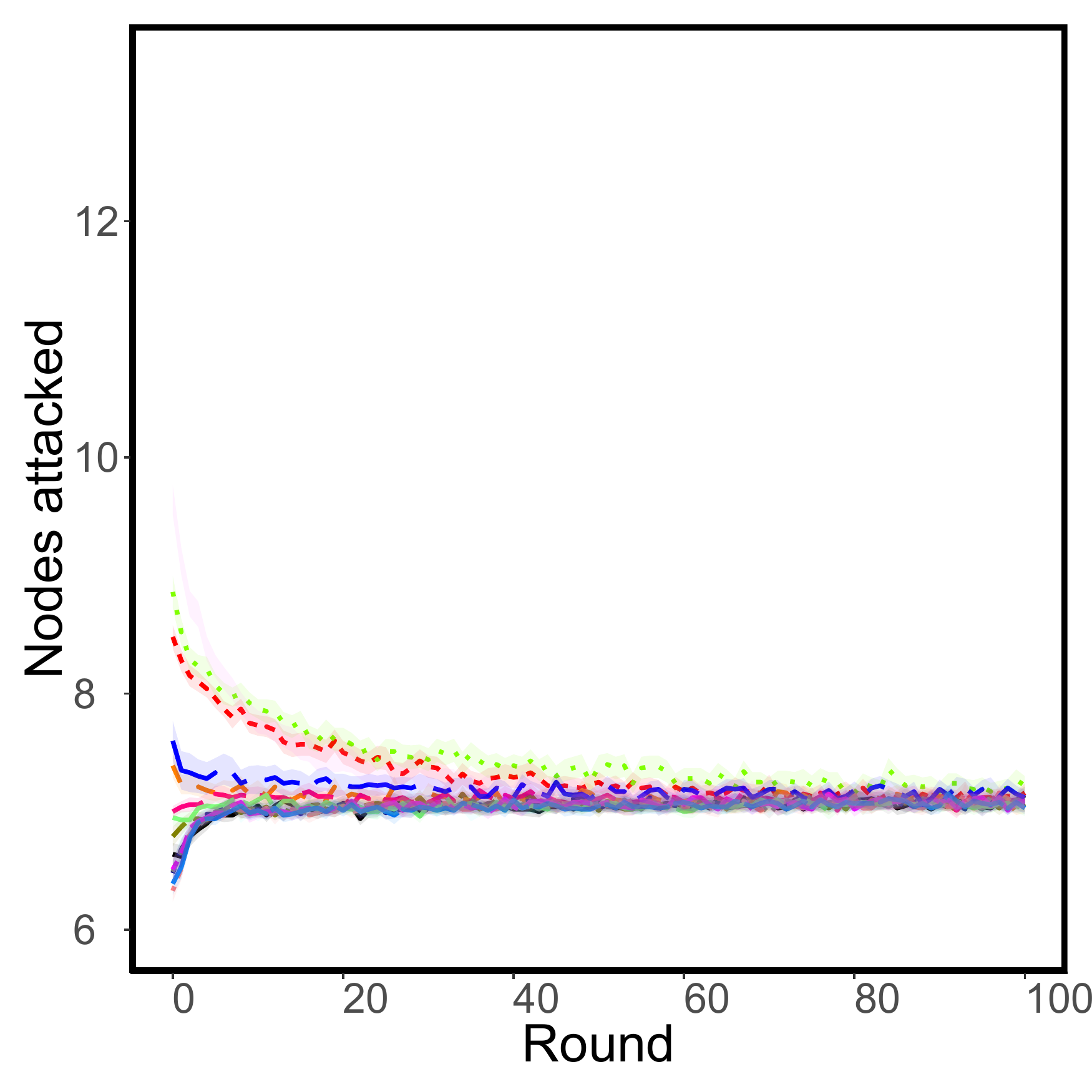} \\
\end{tabular}
\begin{tabular}{m{.32\textwidth}m{.32\textwidth}}
\multicolumn{1}{c}{Small world networks} &
\multicolumn{1}{c}{Random trees} \\
\includegraphics[width=\linewidth]{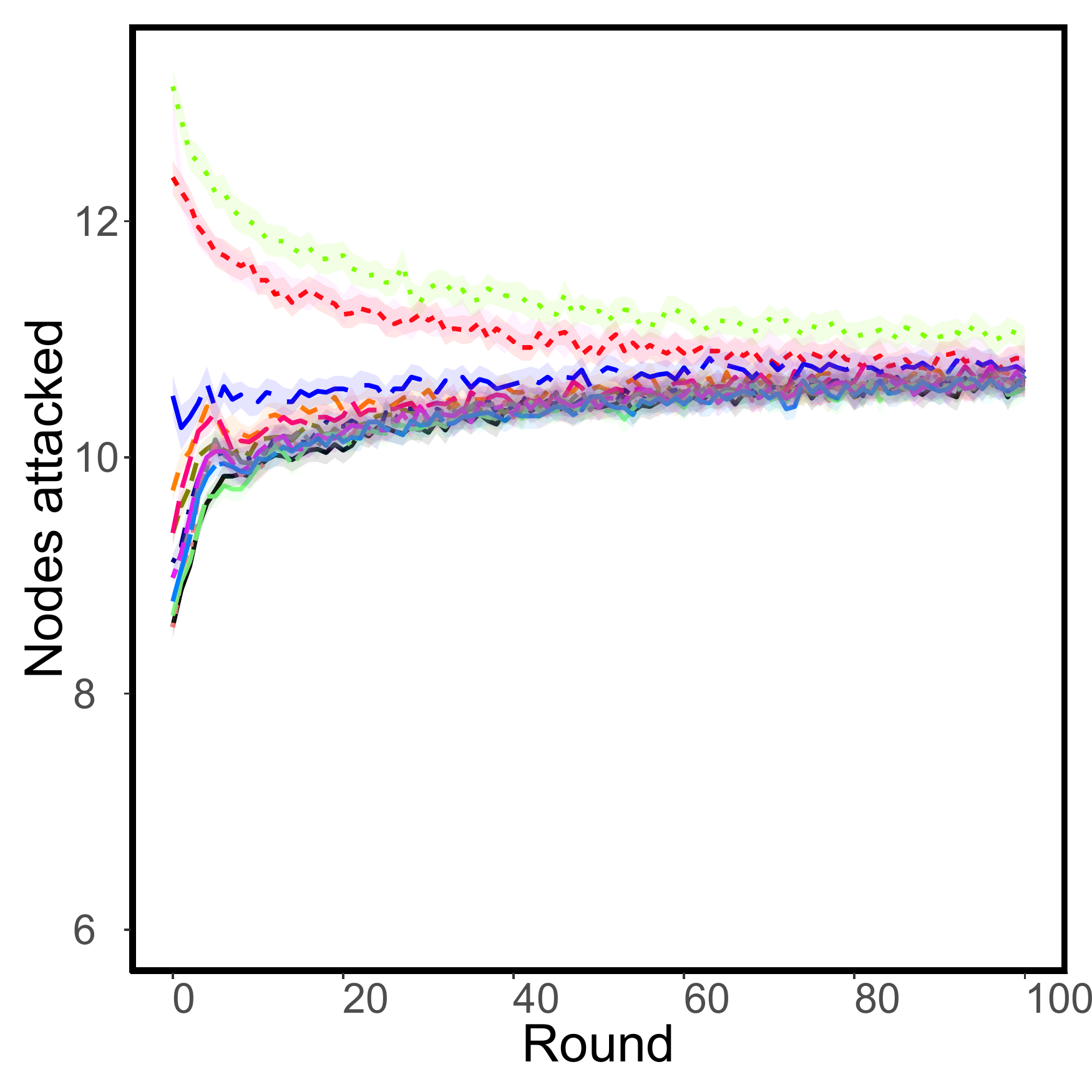} &
\includegraphics[width=\linewidth]{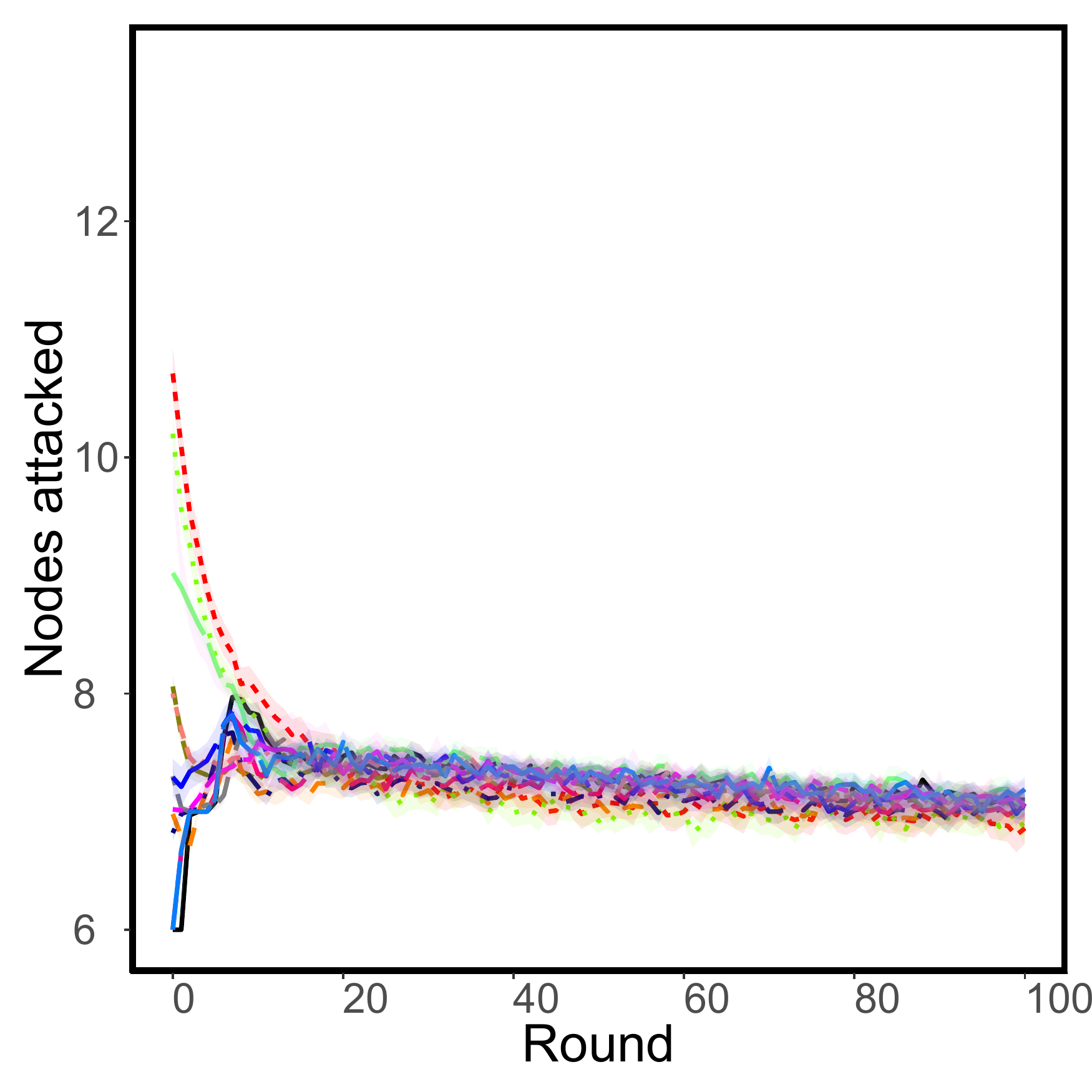} \\
\end{tabular}
\begin{tabular}{m{\textwidth}}
\includegraphics[width=\linewidth]{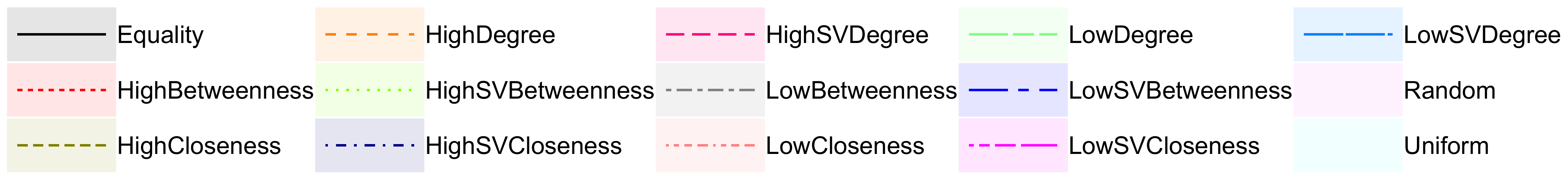} \\
\end{tabular}
\caption{
Results of our experiments for repeated game setting.
The x-axis represents the round, \ie, the  consecutive instance of the game.
The y-axis represents the number of activated nodes after the process finished.
Every line represents different initial strategy of the defender.
The scale is fixed to make plots easier to compare.
}
\label{fig:repeated}
\end{figure*}

\section{Experimental Results with Alternative Diffusion Models}
\label{app:alternative-models}

We now consider alternative settings of our experiments, where we modify the formula of the activation time of a given node.
In particular, following the strategic diffusion model ~\cite{alshamsi2018optimal}, we introduce the parameter $\alpha \in \R^+$ that allows to adjust the strength of path dependencies.
In a setting with a given value of $\alpha$ we have that:

\begin{itemize}
\item The first node in the sequence chosen by the attacker (the seed node) is activated in time $\et(v) = \left(d(v) + \sr(v)\right)^\alpha$.
\item The activation time for all other nodes is $\et(v) = \left(\frac{d(v) + \sr(v)}{|N(v) \cap I|}\right)^\alpha$ where $I$ is the set of currently activated nodes.
\end{itemize}

Notice that for $\alpha=1$ this is equivalent to the basic model described in Section~\ref{sec:game-definition} of the main article.

The results of our experiments with $\alpha=0.8$ and with $\alpha=1.25$ are presented in Figures~\ref{fig:heat-80-alt-models-0.8}, \ref{fig:heat-80-alt-models-1.25}, and \ref{fig:best-bars-alt-models}.

As for the relative effectiveness of the different strategies of the attacker, the general trends are similar to those in experiments with the basic model with $\alpha = 1$.
The parameters of the attacker's strategies that proved to be successful in the basic version of the experiments are usually also effective when the function determining the time of activation changes.

As for the relative effectiveness of the defender's strategies, results for the setting with $\alpha=0.8$ show similar tendencies to those for $\alpha=1$.
However, increasing the value of $\alpha$ to $1.25$ brings some interesting insights.
The Equality strategy, in other cases not particularly successful, suddenly becomes one of the best ways of defending the network.
In situation where capturing the nodes becomes significantly harder (which is the case for $\alpha>1$), withholding the initial attack also becomes more crucial (which is the idea behind the Equality strategy).

As for the comparison between different network models, while results for $\alpha=0.8$ present similar trends to the results for the basic setting, again an interesting change can be noticed when $\alpha=1.25$.
While in other cases the small-world networks were the easiest to attack, they are more resilient when $\alpha > 1$.
Increasing the value of $\alpha$ puts more pressure on activating nodes with many active neighbors and reduces the advantage of being able to quickly reach all parts of the network.

\begin{figure}[t]
\centering
\setlength\tabcolsep{0pt}
\begin{tabular}{m{.32\textwidth}m{.32\textwidth}m{.32\textwidth}}
\multicolumn{1}{c}{Preferential attachment networks} &
\multicolumn{1}{c}{Scale free networks} &
\multicolumn{1}{c}{Random graph networks} \\
\includegraphics[width=\linewidth]{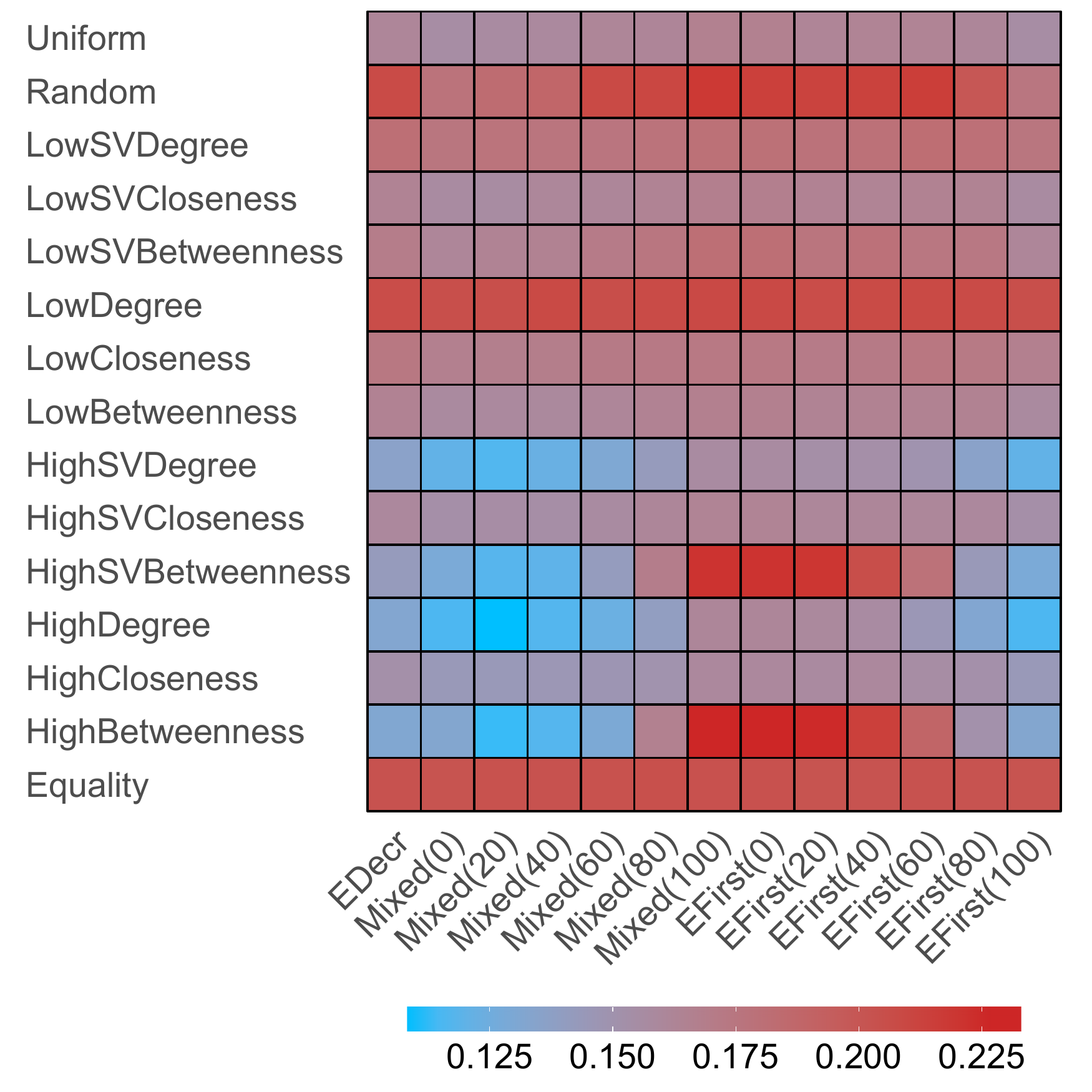} &
\includegraphics[width=\linewidth]{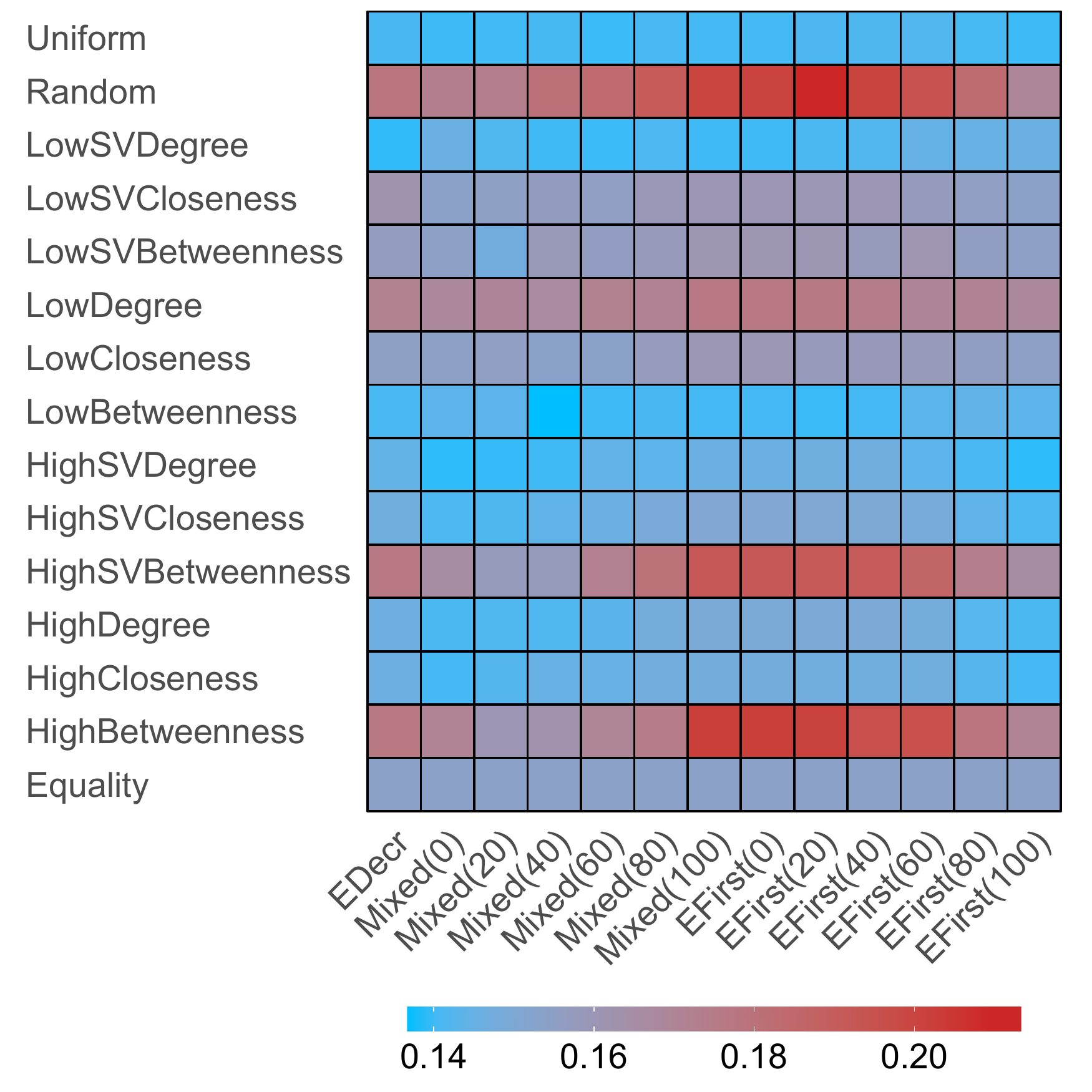} &
\includegraphics[width=\linewidth]{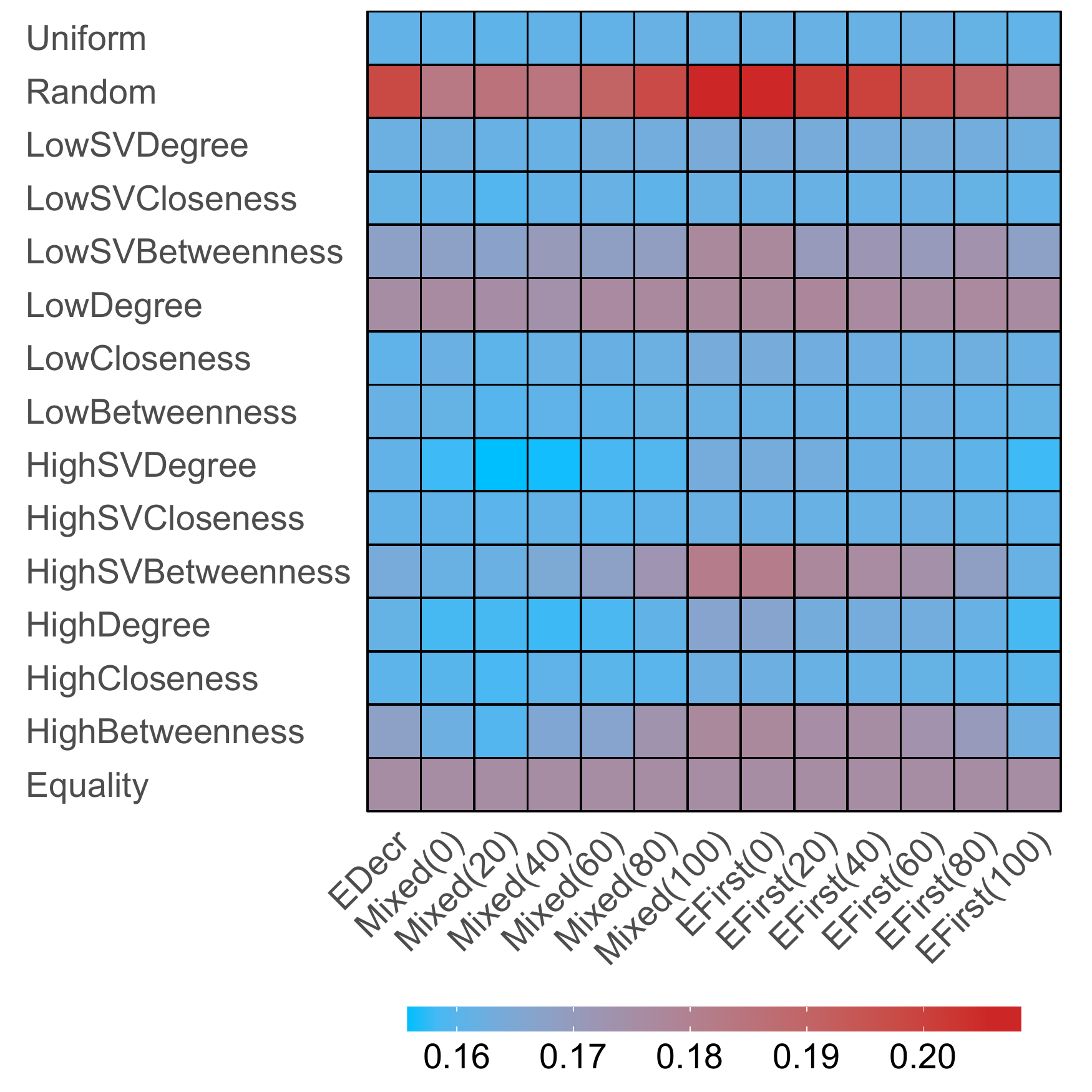} \\
\end{tabular}
\begin{tabular}{m{.32\textwidth}m{.32\textwidth}}
\multicolumn{1}{c}{Small world networks} &
\multicolumn{1}{c}{Random trees} \\
\includegraphics[width=\linewidth]{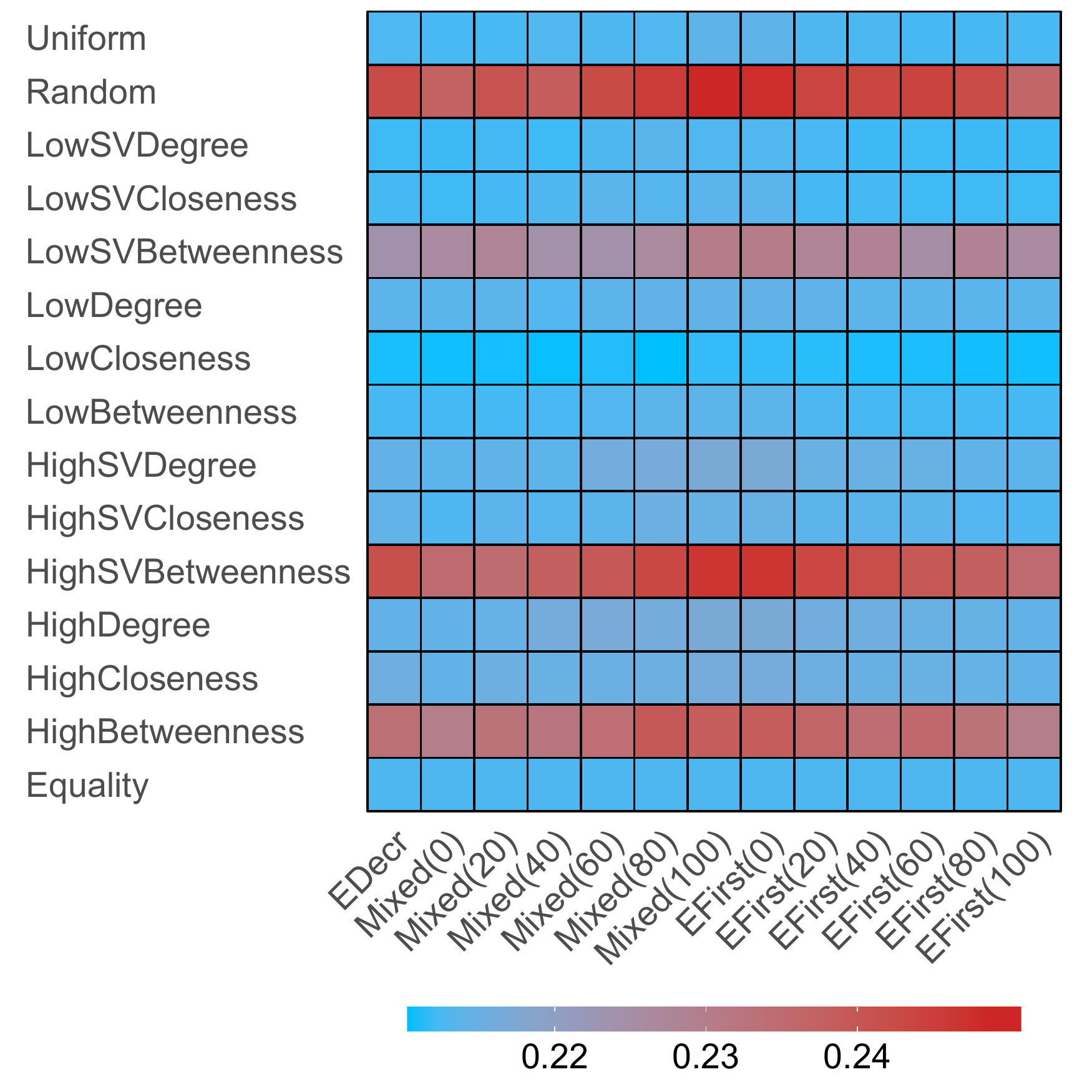} &
\includegraphics[width=\linewidth]{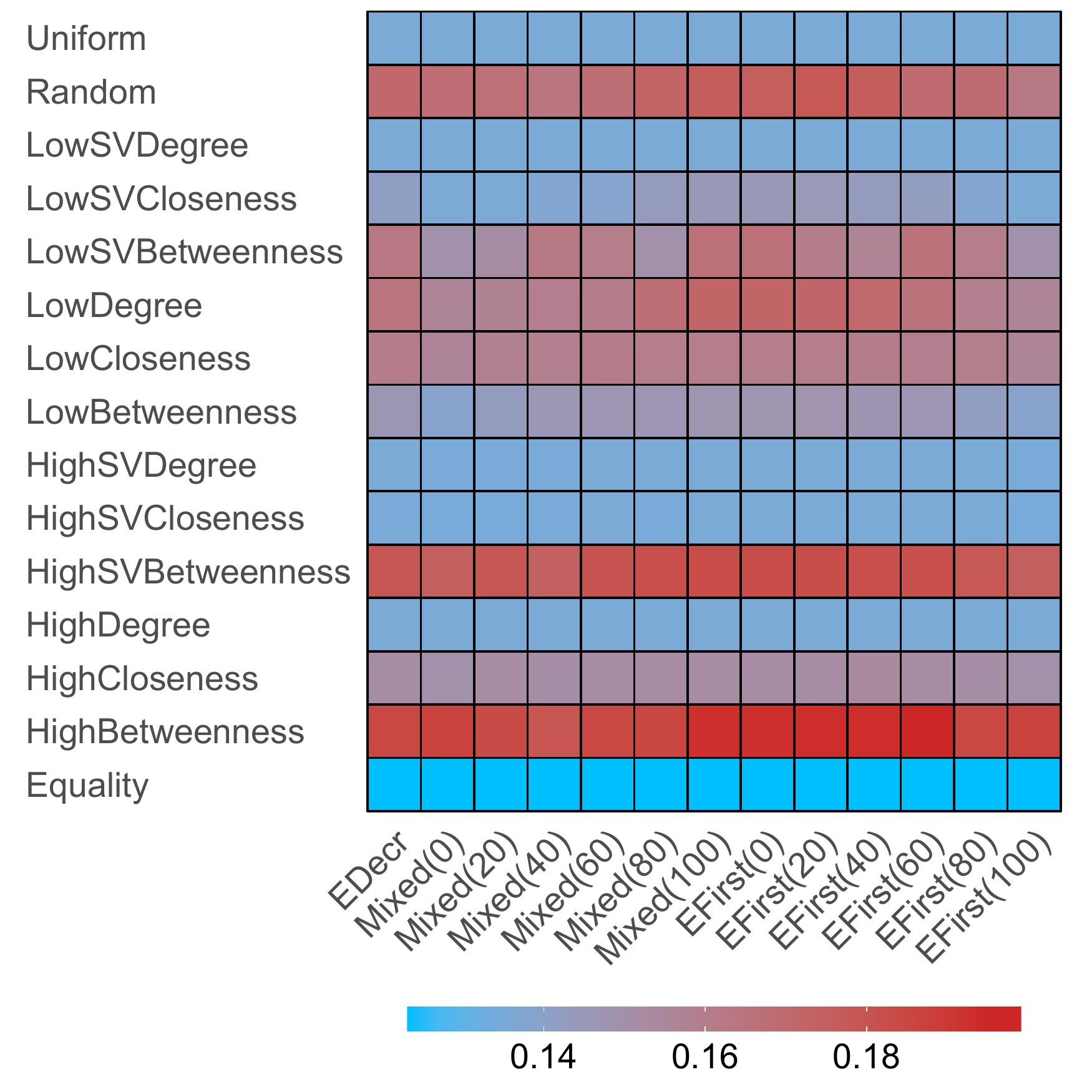} \\
\end{tabular}
\caption{
Comparison of effectiveness of defender's strategies for different attacker's strategies on networks with $80$ nodes in model with $\alpha=0.8$.
Color of each cell represents the expected percentage of nodes successfully activated by the attacker.
Results are taken as an average over $100$ simulations, with a new network generated for each simulation using one of the models.
}
\label{fig:heat-80-alt-models-0.8}
\end{figure}

\begin{figure}[t]
\centering
\setlength\tabcolsep{0pt}
\begin{tabular}{m{.32\textwidth}m{.32\textwidth}m{.32\textwidth}}
\multicolumn{1}{c}{Preferential attachment networks} &
\multicolumn{1}{c}{Scale free networks} &
\multicolumn{1}{c}{Random graph networks} \\
\includegraphics[width=\linewidth]{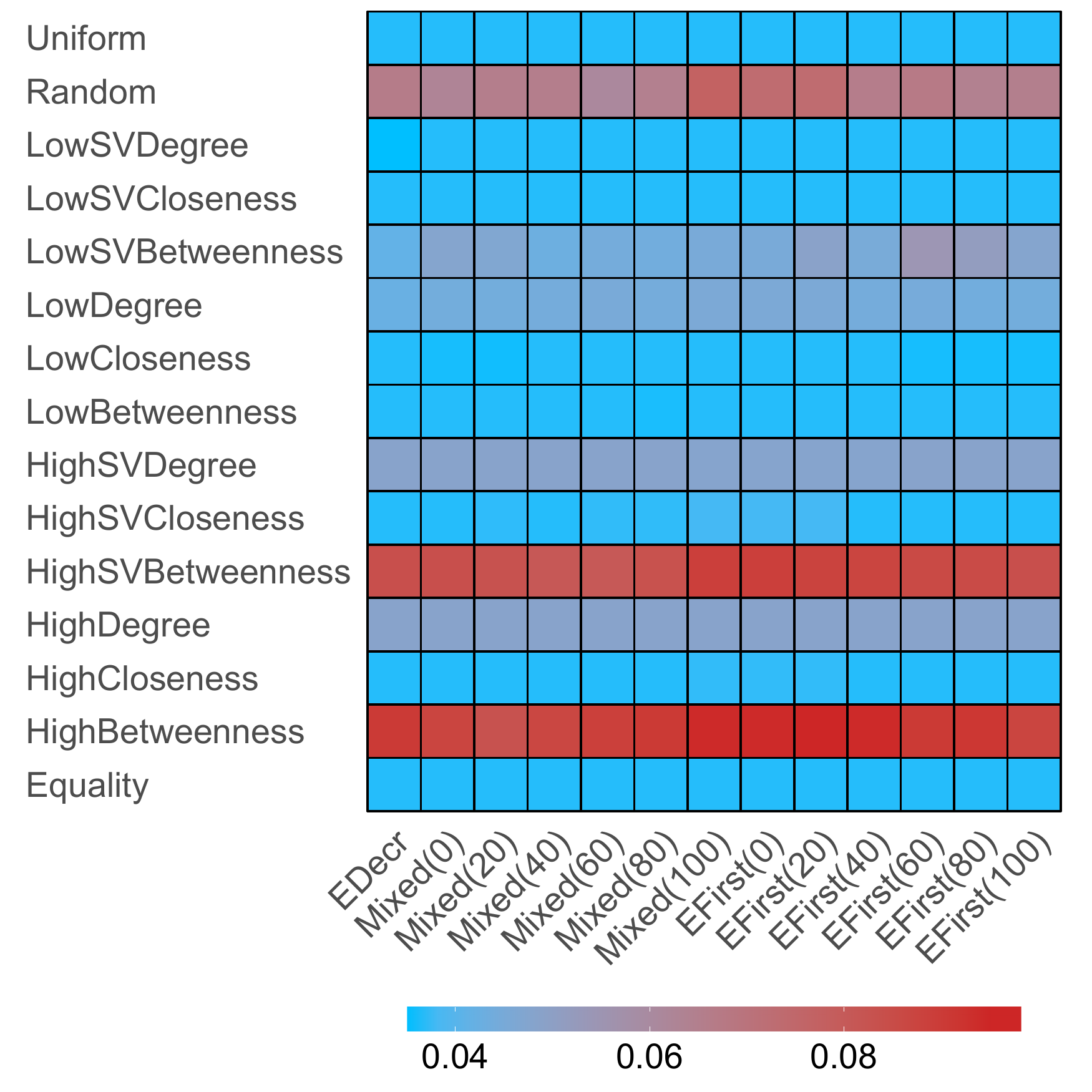} &
\includegraphics[width=\linewidth]{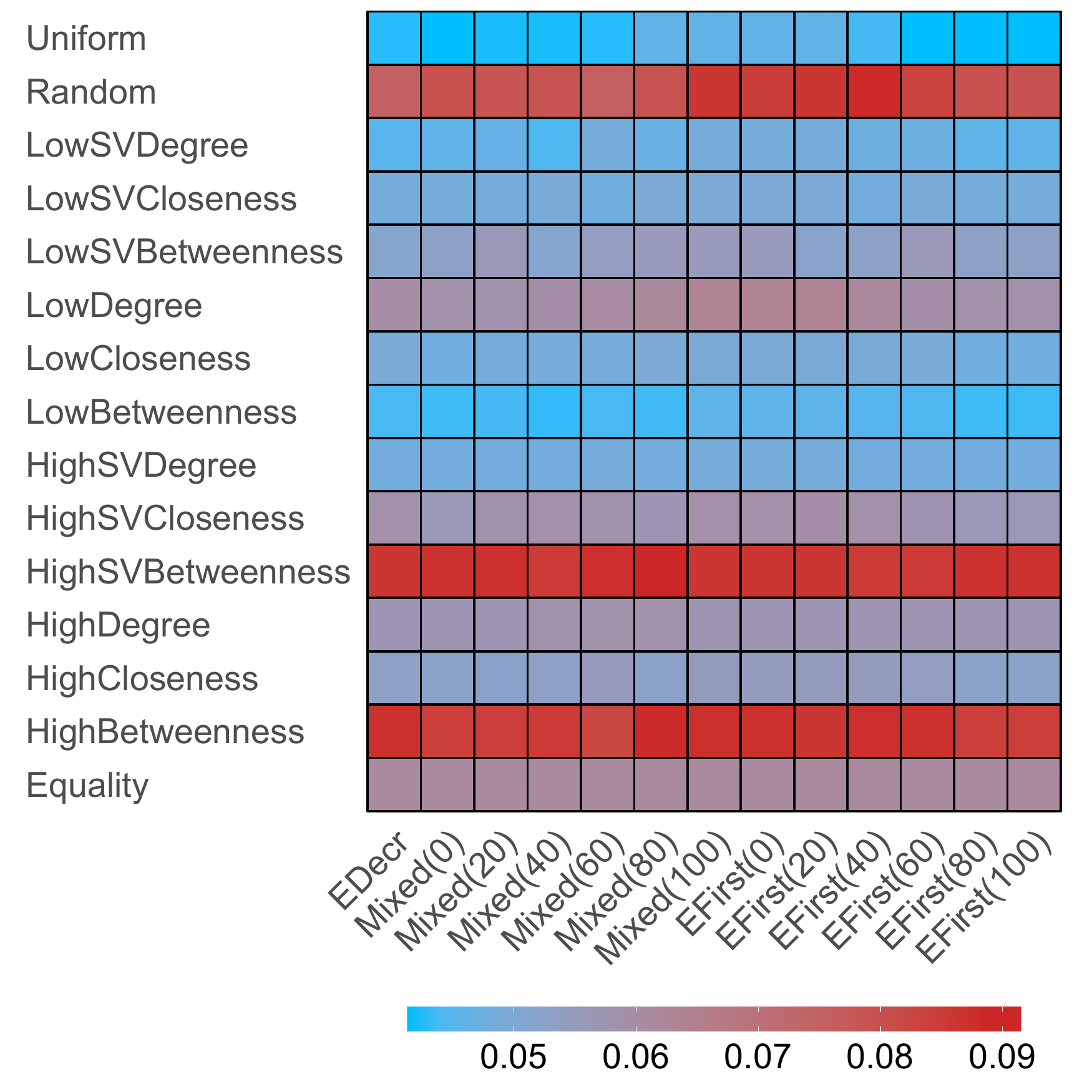} &
\includegraphics[width=\linewidth]{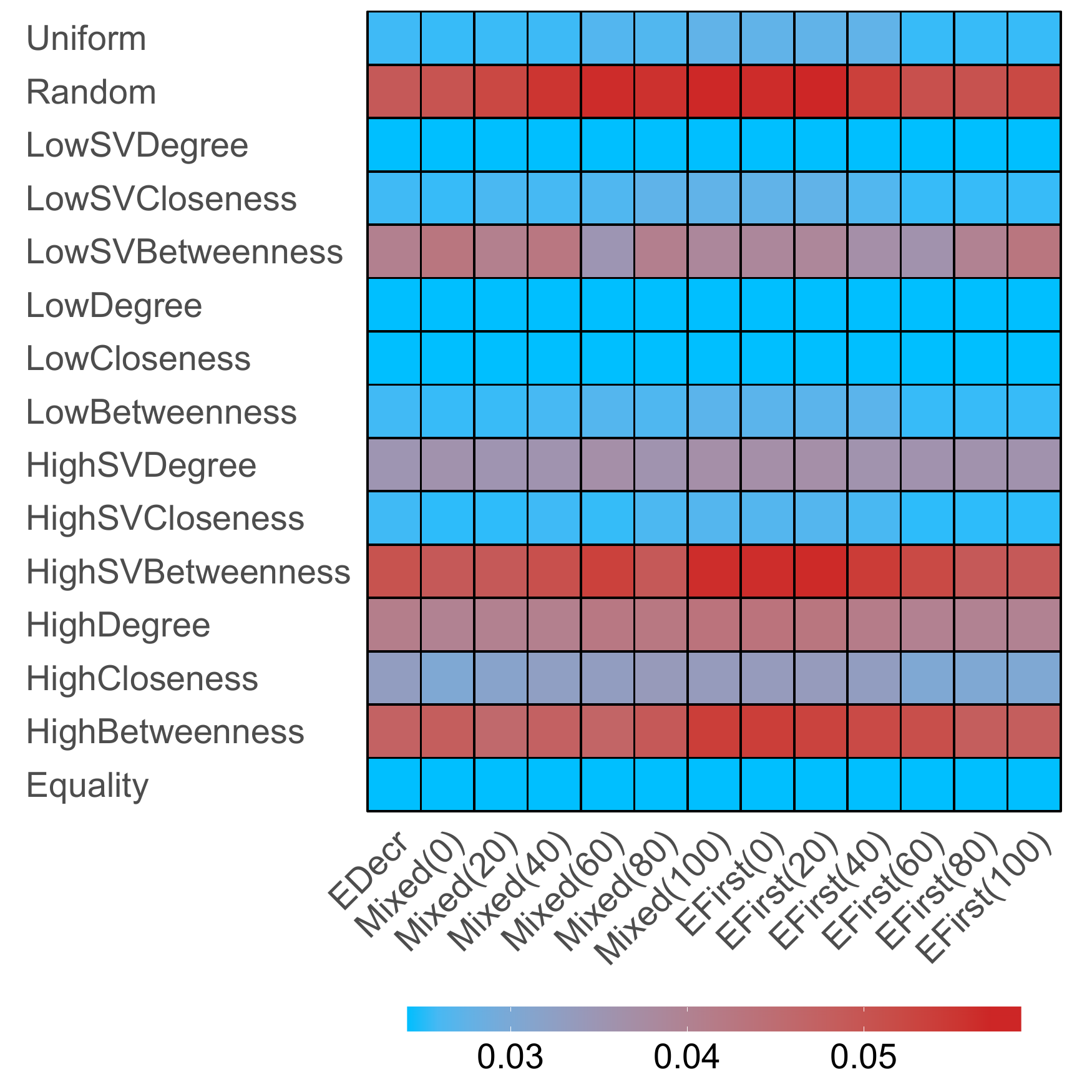} \\
\end{tabular}
\begin{tabular}{m{.32\textwidth}m{.32\textwidth}}
\multicolumn{1}{c}{Small world networks} &
\multicolumn{1}{c}{Random trees} \\
\includegraphics[width=\linewidth]{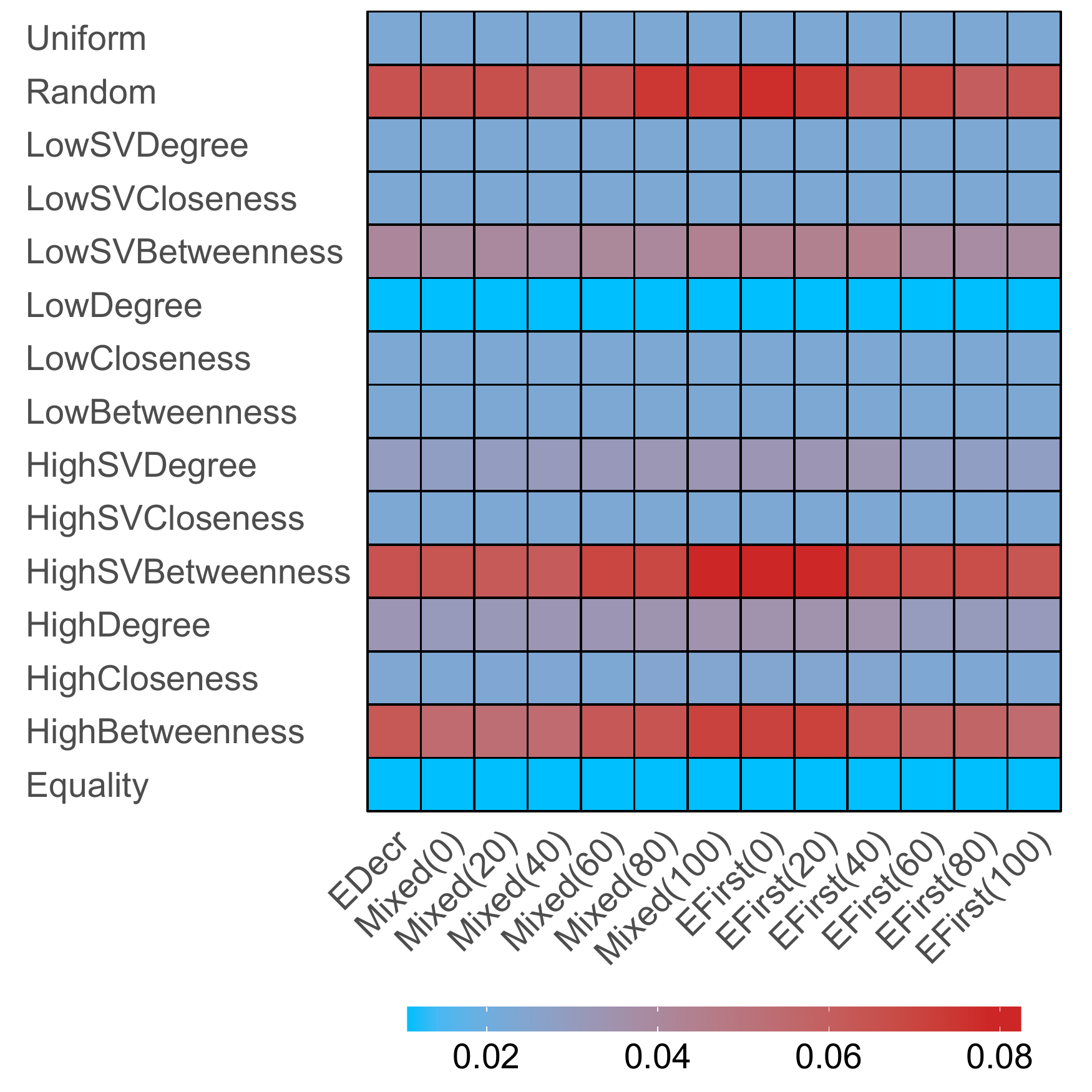} &
\includegraphics[width=\linewidth]{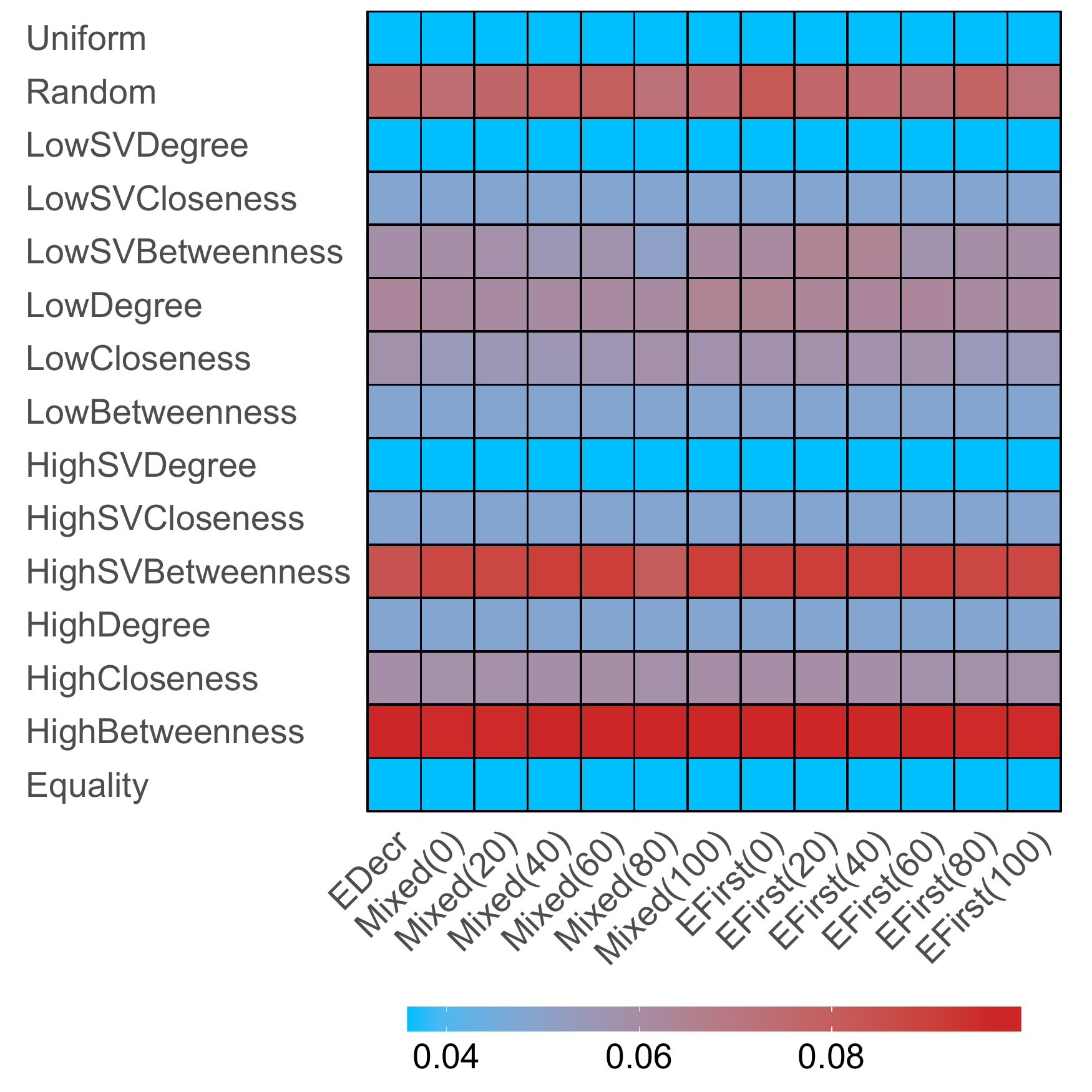} \\
\end{tabular}
\caption{
Comparison of effectiveness of defender's strategies for different attacker's strategies on networks with $80$ nodes in model with $\alpha=1.25$.
Color of each cell represents the expected percentage of nodes successfully activated by the attacker.
Results are taken as an average over $100$ simulations, with a new network generated for each simulation using one of the models.
}
\label{fig:heat-80-alt-models-1.25}
\end{figure}

\begin{figure}[t]
\centering
\setlength\tabcolsep{0pt}
\begin{tabular}{m{.5\textwidth}m{.5\textwidth}}
\multicolumn{1}{c}{$\alpha=0.8$} &
\multicolumn{1}{c}{$\alpha=1.25$} \\
\includegraphics[width=\linewidth]{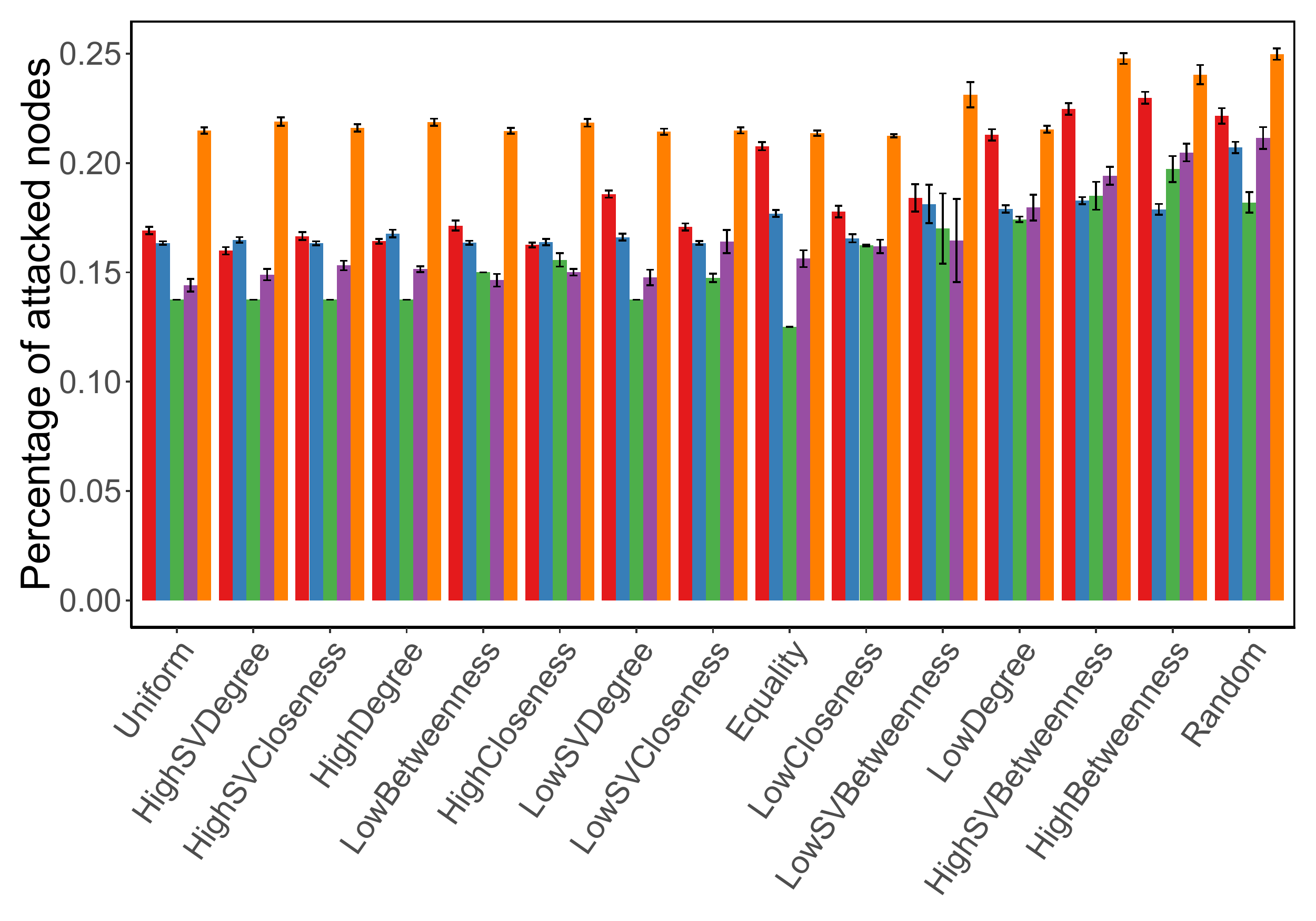} &
\includegraphics[width=\linewidth]{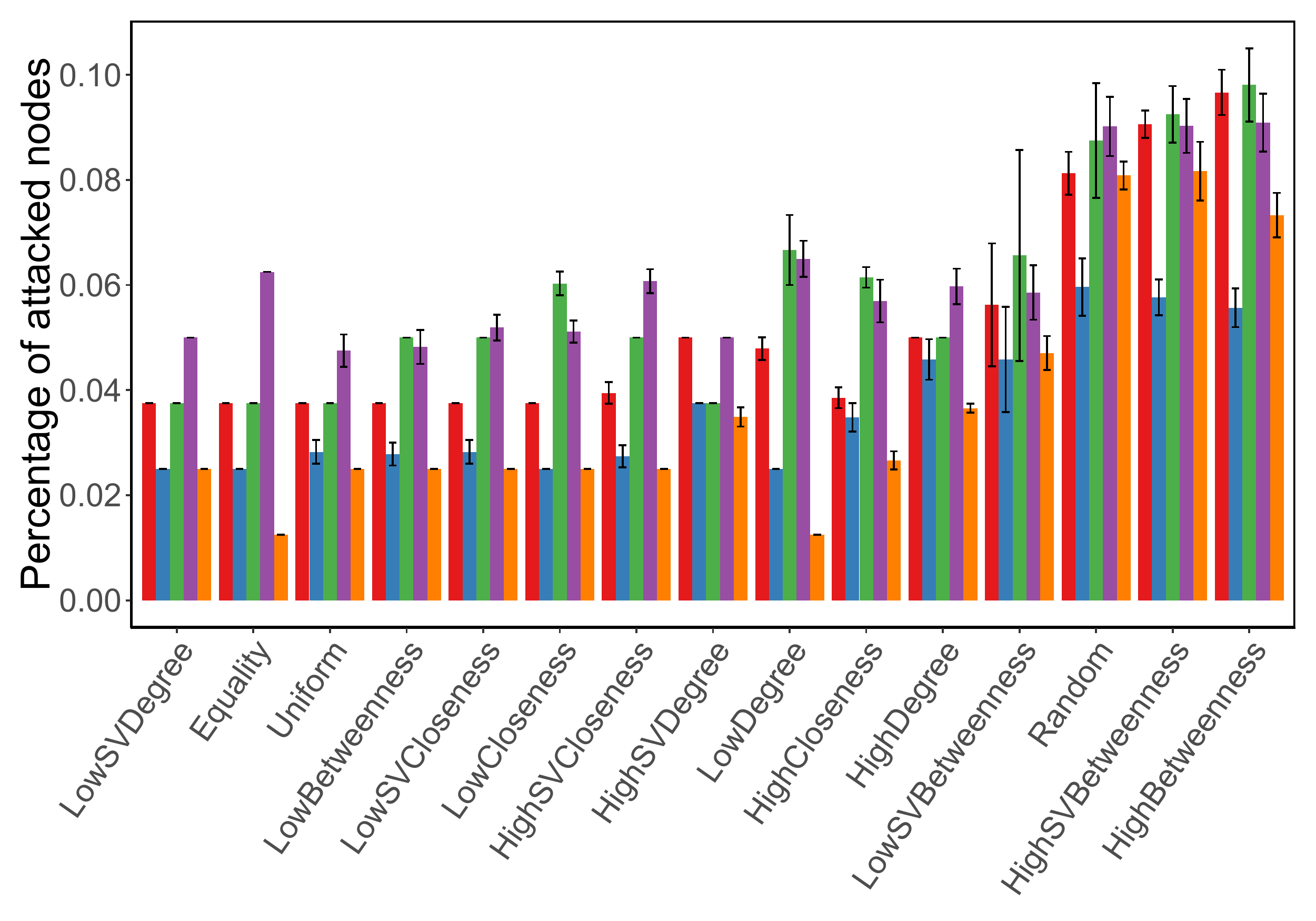} \\
\multicolumn{2}{c}{\includegraphics[width=.7\linewidth]{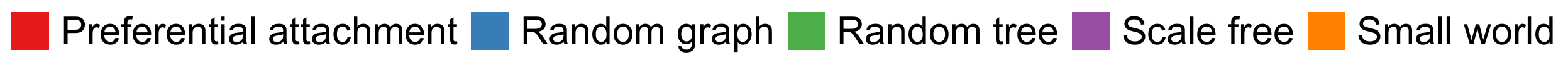}} \\
\end{tabular}
\caption{
Comparison of effectiveness of defender's strategies, under assumption that the attacker uses the best heuristic considered by us, for models with different values of $\alpha$.
Results are taken as an average over $100$ simulations, with a new network generated for each simulation using one of the models.
Defender's strategies are sorted from most to least effective on average.
Error bars represent $95\%$ confidence intervals.
}
\label{fig:best-bars-alt-models}
\end{figure}

\end{document}